\documentclass[11pt]{article}

\usepackage[margin=1in]{geometry}
\usepackage[comma]{natbib}
\usepackage{anysize}
\usepackage{color}
\usepackage{graphicx}
\usepackage{amsmath,amssymb,amsthm,amsopn,mathrsfs}
\usepackage{caption}
\usepackage{subcaption}
\usepackage[applemac]{inputenc}
\usepackage{booktabs}
\usepackage[makeroom]{cancel}

\def\thick#1{\hbox{\rlap{$#1$}\kern0.25pt\rlap{$#1$}\kern0.25pt$#1$}}

\def\be{{\bf e}}

\def\cL{{\mathcal{L}_\mathcal{C}}}

\def\Cov{\mbox{Cov}}

\def\diag{\mbox{diag}}

\def\real{{\mathbb R}}
\def\reald{{\mathbb R}^d}

\def\expect{{\mathbb E}}
\def\prob{\text{Pr}}
\def\variance{{\text{Var}}}

\def\cS{{\cal S}}

\newcommand{\der}{{\mathrm d}}

\def\cS{\mathcal{S}}

\def\spa{\mathbb{S}}
\def\powset{\mathbb{P}}
\def\index{\mathbb{I}}
\def\cD{\mathcal{D}}

\def\cN{\mathcal{N}}

\def\cL{\mathscr{L}}
\def\simplex{\mathbb{W}}

\def\cT{\mathcal{T}}
\def\der{\text{d}}
\def\skn{\cS\cN}
\def\skt{\cS\cT}

\def\wconv{\Rightarrow}
\def\goinf{n\rightarrow\infty}

\theoremstyle{definition}\newtheorem{defi}{Definition}
\theoremstyle{plain}\newtheorem{theo}{Theorem}\newtheorem{prop}{Proposition}\newtheorem{lem}{Lemma}
\theoremstyle{remark}

\begin{document}

\title{Models for extremal dependence derived from skew-symmetric families}

\author{Boris Beranger\\
{\it School of Mathematics and Statistics, University of New South Wales, Australia}\\
 Simone A. Padoan\\
 {\it Department of Decision Sciences, Bocconi University, Italy}\\
 Scott A. Sisson\\
{\it School of Mathematics and Statistics, University of New South Wales, Australia} 
}

\date{}

\maketitle
\begin{abstract}
Skew-symmetric families of distributions such as the
skew-normal and skew-$t$ represent supersets of the normal and $t$ distributions, and they exhibit richer classes of extremal behaviour.
By defining a non-stationary skew-normal process, which allows the easy handling of positive definite, non-stationary covariance functions, we derive a new family of max-stable processes -- the extremal-skew-$t$ process.
This process is a superset of non-stationary processes that include the stationary extremal-$t$ processes.
We provide the spectral representation and the resulting angular densities of the extremal-skew-$t$ process, and illustrate its practical implementation\\

\end{abstract}
%

%
%
%
\section{Introduction}

The modern-day  analysis of extremes  is based on results from the theory of stochastic processes.
In particular, max-stable processes \citep{dehaan1984} are a popular and useful tool when modelling extremal responses in environmental, financial and engineering applications.
Let $\spa\subseteq \real^k$ denote a $k$-dimensional region of space (or space-time) over which a real-valued stochastic process $\{Y(s)\}_{s\in\spa}$ with a continuous sample path on $\spa$ can be defined.
 Considering a sequence $Y_1,\ldots,Y_n$ of independent and identically distributed (iid) copies of $Y$,
the pointwise partial maximum can be defined as 
$$
M_n(s)=\max_{i=1,\ldots,n}Y_i(s),\quad s\in\spa.
$$ 
If there are sequences of real-valued functions, $a_n(s)>0$ and $b_n(s)$, for $s\in\spa$ and $n=1,2,\ldots$, such that 
$$
\left\{\frac{M_n(s)-b_n(s)}{a_n(s)}\right\}_{s\in \spa} \wconv \{U(s)\}_{s\in\spa},
$$
converges weakly as $n\rightarrow\infty$ to a process $U(s)$
with non-degenerate marginal distributions for all $s\in\spa$,
then $U(s)$ is known as a max-stable process \citep[][Ch. 9]{dehaan2006}.
In this setting, for a finite sequence of points $(s_j)_{j\in I}$ in $\spa$, where $I=\{1,\ldots,d\}$ is an index set, the finite-dimensional distribution of $U$ is then a multivariate extreme value distribution
\citep[][Ch. 6]{dehaan2006}.  This distribution has generalised extreme value univariate margins and, when parameterised with unit 
Fr\'echet margins, has a joint distribution function of the form
\begin{equation*}\label{eq:maxst_distr}
G(x_j, j\in I)=\exp\{-V(x_j,j\in I)\},\quad x_j>0,
\end{equation*}
where $x_j\equiv x(s_j)$. The exponent function $V$
describes the dependence between extremes, and can be expressed as
\begin{equation*}\label{eq:exponent}
V(x_j,j\in I)=\int_{\simplex}\max_{j\in I}(w_j/x_j) H(\der w_1,\ldots,\der w_d),
\end{equation*}
where the angular measure $H$ is a finite measure defined on the $d$-dimensional unit simplex 
$
\simplex=\{w\in\reald:w_1+\cdots+w_d=1\},
$ 
satisfying the moment conditions 
$
\int_\simplex w_j\; H(\der w)=1, j\in I,
$
\citep[][Ch. 6]{dehaan2006}.

In recent years a variety of specific max-stable processes have been developed, many of which have become popular as they can be practically amenable to statistical modelling \citep{davison2012b}. 
The extremal-$t$ process \citep[][]{opitz2013} is one of the best-known and widely-used max-stable processes,  
from which the Brown-Resnick process \citep[][]{brown1977, kabluchko2009stationary}, the Gaussian extreme-value process \citep[][]{Smith1990a} and the extremal-Gaussian processes
\citep[][]{schlather2002a} can be seen as special cases. In their most basic form, the Brown-Resnick and the extremal-$t$ processes can be respectively understood as the limiting extremal processes of strictly stationary Gaussian and Student-$t$ processes. 
However, in practice, data may be non-stationary and exhibit asymmetric distributions in many applications.
In these scenarios, skew-symmetric distributions \citep{azzalini2013skew, arellano2006unification, azzalini2005, genton2004, azzalini1985} provide simple models for modelling asymmetrically distributed data. However, the limiting extremal behaviour of these processes has not yet been established.

In this paper we characterise and develop statistical models for the extremal behaviour of skew-normal and skew-$t$ distributions. The joint tail behaviours of these skew distributions are capable of describing 
a far wider range of dependence levels than that obtained under the symmetric normal and $t$ distributions.
We provide a definition of a skew-normal process which is in turn a non-stationary process. This provides an accessible approach to constructing positive definite, non-stationary covariance functions when working with non-Gaussian processes.
Recently some forms of non-stationary dependent structures embedded into max-stable processes have been studied by \cite{huser+g14}.
We show that on the basis of the skew-normal process a new family of max-stable processes -- the extremal-skew-$t$ process -- can be obtained. 
This process is a superset of non-stationary processes that includes the stationary extremal-$t$ processes \citep[][]{opitz2013}. From the extremal-skew-$t$ process, a rich
family of non-stationary, isotropic or anisotropic extremal coefficient functions can be obtained. 

This paper is organised as follows: in Section \ref{sec:prelim} we first introduce a new variant of the extended skew-$t$ class of distributions, before developing a non-stationary version of the skew-normal process. In both cases we discuss
the stochastic behavior of their extreme values.
In Section \ref{sec:spectral} we derive the spectral representation of the extended extremal skew-$t$ process.
Section \ref{sec:inference} discusses inferential aspects of the extremal skew-$t$ dependence model, and
Section \ref{sec:application} provides a real data application. We conclude with a Discussion.

%
%
%
\section{Preliminary results on skew-normal processes and skew-$t$ distributions}\label{sec:prelim}

We introduce two preliminary results that will be used
in order to present our main contribution in Section \ref{sec:spectral}, the extremal-skew-$t$ process.
In Section \ref{sec:noncest} we define the {\it non-central} extended skew-$t$ family of distributions,
which is a new variant of the class introduced by \cite{arellano2010}, that allows a non-centrality parameter.
In Section \ref{sec:extskew} we present the development of a new non-stationary, skew normal random process.

Hereafter, we use $Y\sim \cD_d(\theta_1,\theta_2,\ldots)$ to denote that $Y$ is a $d$-dimensional random vector with probability law
$\cD$ and parameters $\theta_1, \theta_2,\ldots$. When $d=1$ the subscript is omitted for brevity.
Similarly, when a parameter is equal to zero or a scale matrix is equal to the identity (both in a vector and scalar sense) so that $\cD_d$ reduces to an obvious sub-family, it is also omitted.
%

%
%
%
\subsection{The non-central, extended skew-$t$ distribution}\label{sec:noncest}

While several skew-symmetric distributions have been developed \citep[see e.g.,][]{genton2004, azzalini2013skew}, we focus on the skew-normal and skew-$t$ distributions. 

Denote a $d$-dimensional skew-normally distributed random vector by 
$Y \sim \skn_d(\mu, \Omega, \alpha, \tau)$ \citep{arellano2010}. This random vector has probability density function (pdf)
\begin{equation}\label{eq:ext_skew_normal}
\phi_d(y;\mu,\Omega,\alpha,\tau)= \frac{\phi_d(y;\mu,\Omega)}
{\Phi\{\tau/\sqrt{1 + Q_{\bar{\Omega}}(\alpha)}\}} \, 
\Phi(\alpha^{\top} z + \tau), \qquad y \in \reald,
\end{equation}
where $\phi_d(y;\mu,\Omega)$ is a $d$-dimensional normal pdf with mean $\mu\in\reald$ and
$d\times d$ covariance matrix $\Omega$, $z=(y - \mu)/\omega$, $\omega=\diag(\Omega)^{1/2}$, 
$\bar{\Omega}=\omega^{-1} \, \Omega\, \omega^{-1}$,  $Q_{\bar{\Omega}}(\alpha)=\alpha^\top \bar{\Omega}\alpha$ and $\Phi(\cdot)$ is the standard univariate normal cumulative distribution function (cdf). The shape parameters  $\alpha \in \reald$ and $\tau\in\real$ are respectively {\it slant} and {\it extension} parameters.
The cdf associated with \eqref{eq:ext_skew_normal} is termed the extended skew-normal distribution \citep{arellano2010}
of which the skew-normal and normal distributions are special cases \citep{arellano2010,azzalini2013skew}. 
For example, in the case where $\alpha=0$ and $\tau=0$ the standard normal pdf is recovered.

\begin{defi}\label{def:noncen_skew} 

$Y$ is a $d$-dimensional, non-central extended skew-$t$ distributed random vector, denoted by
$Y\sim \skt_d(\mu, \Omega, \alpha, \tau,\kappa,\nu)$, if for  $y\in\reald$ it has pdf
\begin{equation}\label{eq:nc_ext_skew_t}
\psi_d(y;\mu,\Omega,\alpha,\tau,\kappa,\nu)=\frac{\psi_d(y;\mu,\Omega,\nu)}
{\Psi\left(\frac{\tau}{\sqrt{1+Q_{\bar{\Omega}}(\alpha)}};\frac{\kappa}{\sqrt{1+Q_{\bar{\Omega}}(\alpha)}},\nu\right)} 
\Psi \left\{(\alpha^{\top} z+\tau) \sqrt{\frac{\nu + d}{\nu + Q_{\bar{\Omega}^{-1}}(z)}};\kappa,\nu+d\right\},
\end{equation}
where  $\psi_d(y;\mu,\Omega,\nu)$ is the pdf of a $d$-dimensional $t$-distribution with location $\mu\in\reald$, $d\times d$ scale matrix
$\Omega$ and $\nu\in\mathbb{R}^+$ degrees of freedom, $\Psi(\cdot;a,\nu)$ denotes a univariate non-central $t$ cdf with non-centrality parameter $a\in\real$ and $\nu$ degrees of freedom, and $Q_{\bar{\Omega}^{-1}}(z)=z^{\top} \bar{\Omega}^{-1} z$. The remaining terms are as defined in \eqref{eq:ext_skew_normal}.
The associated cdf is
\begin{equation}\label{eq:nc_ext_skew_t_cdf}
\Psi_d(y;\mu,\Omega,\alpha,\tau,\kappa,\nu) = 
\frac{\Psi_{d+1} \left\{
\bar{z};
\Omega^*,\kappa^*, \nu
\right\}}
{\Psi\left(\bar{\tau}; \bar{\kappa},\nu\right)},
\end{equation}
where $\bar{z}=(z^\top,\bar{\tau})^{\top}$, $\Psi_{d+1}$ is a $(d+1)$-dimensional (non-central) $t$ cdf with covariance matrix
and non-centrality parameters
$$
\Omega^*=\left( \begin{array}{cc} 
\bar{\Omega} & - \delta\\
 - \delta^\top & 1 
 \end{array} \right),
\quad \kappa^*=\left( \begin{array}{c} 0 \\ \bar{\kappa} \end{array} \right),
$$
and $\nu$ degrees of freedom,  
and where 
\begin{equation}\label{eq:parameters}
\delta = \left\{ 1+Q_{\bar{\Omega}}(\alpha) \right\}^{-1/2}\,\bar{\Omega}\alpha,\quad
\bar{\kappa} = \left\{ 1+Q_{\bar{\Omega}}(\alpha) \right\}^{-1/2}\, \kappa, \quad
\bar{\tau} = \left\{ 1+Q_{\bar{\Omega}}(\alpha) \right\}^{-1/2}\,\tau. 
\end{equation}
\end{defi}
When the non-centrality parameter $\kappa$ is zero, then the extended skew-$t$ family of \citet{arellano2010}
is obtained.
For the non-central skew-$t$ family, we now demonstrate modified properties to those discussed in \citet{arellano2010}.
\begin{prop}[Properties]\label{pro:prop_nskewt} Let $Y\sim \skt_d(\mu,\Omega,\alpha,\tau,\kappa,\nu)$.
\begin{enumerate}
\item \label{prop1A} Marginal and conditional distributions. 
Let $I\subset\{1,\ldots,d\}$ and $\bar{I}=\{1,\ldots,d\}\backslash I$ identify the $d_I$- and $d_{\bar{I}}$-dimensional subvector partition of $Y$ such that 
$Y=(Y_I^\top,Y_{\bar{I}}^\top)^\top$, with corresponding partitions of the parameters $(\mu,\Omega,\alpha)$. Then
\begin{enumerate}
\item  $Y_I\sim \skt_{d_I}(\mu_I,\Omega_{II},\alpha^*_{I},\tau^*_{I},\kappa^*_{I},\nu)$,
where
\begin{align} \label{eq:margparams}
\begin{array}{ccc}
\alpha^*_{I} = \frac{\alpha_I + \bar{\Omega}_{II}^{-1}\bar{\Omega}_{I\bar{I}}\alpha_{\bar{I}}}
{\sqrt{1 + Q_{\tilde{\Omega}_{\bar{I}\bar{I} \cdot I} }( \alpha_{\bar{I}} ) }},
& 
\tau^*_{I} = \frac{\tau}{\sqrt{1 + Q_{\tilde{\Omega}_{\bar{I}\bar{I} \cdot I} }( \alpha_{\bar{I}} ) }},
&
\kappa^*_{I} = \frac{\kappa}{\sqrt{1 + Q_{\tilde{\Omega}_{\bar{I}\bar{I} \cdot I} }( \alpha_{\bar{I}} ) }},
\end{array}
\end{align}
given 
$
\tilde{\Omega}_{\bar{I}\bar{I}\cdot I} 
= \bar{\Omega}_{\bar{I}\bar{I}} - \bar{\Omega}_{\bar{I}I}\bar{\Omega}_{II}^{-1}\bar{\Omega}_{I\bar{I}}.
$

\item $(Y_{\bar{I}}|Y_I=y_I)\sim \skt_{d_{\bar{I}}}(\mu_{\bar{I}\cdot I},\Omega_{\bar{I}\cdot I},\alpha_{\bar{I} \cdot I},\tau_{\bar{I} \cdot I},\kappa_{\bar{I} \cdot I},\nu_{\bar{I} \cdot I})$,
where
$ \mu_{\bar{I}\cdot I} = \mu_{\bar{I}} + \Omega_{I\bar{I}}\Omega_{II}^{-1} ( y_I - \mu_I )$,
$ \Omega_{\bar{I}\cdot I} = \zeta_I \Omega_{\bar{I}\bar{I} \cdot I} $,
$ \zeta_I = \{ \nu + Q_{\Omega_{II}^{-1}}(z_I)\} / ( \nu + d_I ) $,
$ z_I = \omega_I^{-1}( y_I - \mu_I) $,
$ \omega_I = \diag(\omega_{II})^{1/2} $,
$Q_{\Omega_{II}^{-1}}(z_I)=z_I^{\top} \Omega_{II}^{-1} z_I$,
$ \Omega_{\bar{I}\bar{I} \cdot I} = \Omega_{\bar{I}\bar{I}} - \Omega_{\bar{I}I}\Omega_{II}^{-1}\Omega_{I\bar{I}} $,
$ \alpha_{\bar{I}\cdot I} = \omega_{\bar{I} \cdot I} \omega_{\bar{I}}^{-1} \alpha_{\bar{I}} $,
$ \omega_{\bar{I} \cdot I} = \diag(\Omega_{\bar{I}\bar{I} \cdot I})^{1/2} $,
$ \omega_{\bar{I}} = \diag(\omega_{\bar{I}\bar{I}})^{1/2} $,
$ \tau_{\bar{I} \cdot I} = \zeta_I^{-1/2}\{  
(\alpha_{\bar{I}}^\top \bar{\Omega}_{\bar{I}I}\bar{\Omega}_{II}^{-1} + \alpha_I^\top )z_I + \tau \}$,
$ \kappa_{\bar{I} \cdot I} = \zeta_I^{-1/2} \kappa$
and
$ \nu_{\bar{I}\cdot I} = \nu + d_I $.
\end{enumerate}

\item Conditioning type stochastic representation. We can write $Y=\mu+\Omega Z$, where
$
Z = (X| \alpha^\top X+\tau>X_0),
$
and where $X\sim \cT_d(\bar{\Omega},\nu)$ is independent of $X_0\sim \cT(\kappa,\nu)$.
\item \label{qqq} Additive type stochastic representation. We can write $Y=\mu+\Omega Z$, where
$
Z=\sqrt{\frac{\nu + \tilde{X}_0^2}{\nu +1}} X_1 + \delta \tilde{X}_0,
$
$X_1 \sim \cT_d(\Omega-\delta \delta^\top, \bar{\kappa}, \nu+1)$ 
is independent of $\tilde{X}_0 = (X_0 | X_0 + \bar{\tau} > 0 )$, 
$X_0 \sim \cT(\bar{\kappa}, \nu )$, $\delta\in(-1,1)^d$ and where $\bar{\tau}$ and $\bar{\kappa}$ are as in \eqref{eq:parameters}. 
\end{enumerate}
Proof in Appendix \ref{ssec:prop_nskewt}
\end{prop}

We conclude by presenting a final property of the non-central skew-$t$ family.
The next result describes the extremal behaviour of observations drawn from a member of this class.
\begin{prop}\label{pro:limit_skew_t} 
Let $Z_1,\ldots,Z_n$ be iid copies of $Z\sim \skt_d(\bar{\Omega},\alpha,\tau,\kappa,\nu)$ and $M_n$ be the componentwise sample maxima.
Define $a_n=(a_{n,1},\ldots,a_{n,d})^\top$, where
$$
a_{n,j}=\left\{
\frac{n\{\Gamma(\nu/2)\}^{-1}
\Gamma\{(\nu+1)/2\} \nu^{(\nu-2)/2}\,\Psi(\alpha^*_{j}\sqrt{\nu+1};\kappa,\nu+1)}
{\sqrt{\pi}\Psi\left(\tau^*_{j}/\{1+Q_{\bar{\Omega}}(\alpha^*_{j})\}^{1/2};\kappa^*_{j}/\{1+Q_{\bar{\Omega}}(\alpha^*_{j})\},\nu\right)}\right\}^{1/\nu}
$$
where $\alpha^*_{j}=\alpha^*_{\{j\}}$, $\tau^*_{j}=\tau^*_{\{j\}}$ and $\kappa^*_{j}=\kappa^*_{\{j\}}$ are the marginal parameters (\ref{eq:margparams}) under Proposition \ref{pro:prop_nskewt}(\ref{prop1A}).
Then $M_n/a_n \wconv U$ as $n\rightarrow+\infty$, where $U$ has univariate 
$\nu$-Fr\'{e}chet marginal distributions (i.e. $e^{-x^{-\nu}}$, $x>0$), and exponent function 
\begin{equation}\label{eq:extst_expo}
V(x_j, j\in I)=
\sum_{j=1}^d x_j^{-\nu} \Psi_{d-1}
\left(\left(
\sqrt{\frac{\nu+1}{1-\omega^2_{i,j}}}
\left(
\frac{x^+_i}{x^+_j} - \omega_{i,j}
\right),i\in I_j\right)^\top; \bar{\Omega}^+_j, \alpha^+_j, \tau^+_j, \nu+1
\right),
\end{equation}
where $\Psi_{d-1}$ is a $(d-1)$-dimensional central extended skew-$t$ distribution with 
 correlation matrix $\bar{\Omega}^+_j$, shape and extension parameters $\alpha^+_j$ and $\tau^+_j$, and 
$\nu+1$ degrees of freedom, 
$I=\{1,\ldots,d\}$, $I_j=I\backslash\{j\}$, 
and $\omega_{i,j}$ is the $(i,j)$-th element of $\bar{\Omega}$.

Proof (and further details) in Appendix \ref{ssec:limit_skew_t}.
\end{prop}

As the limiting distribution \eqref{eq:extst_expo} is the same as that of the classic skew-$t$ distribution \citep[see][]{padoan2011}, 
it exhibits identical upper and lower tail dependence coefficients \citep[e.g.][Ch 5]{joe1997}. That is, the extension and non-centrality parameters, $\tau$ and $\kappa$, do not affect the extremal behavior.
%

%
%
%
\subsection{A non-stationary, skew-normal random process}\label{sec:extskew}

While there are several definitions of a stationary skew-normal process \citep[e.g.][]{minozzo2012existence}, 
stationarity is incompatible with the requirement that all finite-dimensional distributions of the process are skew-normal.
We now construct a non-stationary version of the skew-normal process  through the additive-type stochastic representation \citep[e.g.][Ch. 5]{azzalini2013skew}. A similar approach was  explored by \cite{zhang2010spatial} for the stationary case.

\begin{defi}\label{def:skn_proc}
Let \{$X(s)\}_{s\in\spa}$ be a stationary Gaussian random process on $\spa$ with zero mean, unit variance and correlation function $\rho(h)=\expect\{X(s)X(s+h)\}$ for $s\in\spa$ and $h\in \real^k$.
For $X'\sim\cN(0,1)$ independent of $X(s)$,  $\varepsilon\in\mathbb{R}$ and a function $\delta: \spa \mapsto (-1,1)$,
define
\begin{eqnarray}
X''(s) &:=& X' | X'+\varepsilon >0,\qquad \forall\; s\in\spa\nonumber \\
Z(s)&:=&\sqrt{1-\delta(s)^2}X(s)+\delta(s)X''(s),\quad  s\in\spa.\label{eq:skew_field}
\end{eqnarray}
Then $Z(s)$ is a skew-normal random process.

\end{defi}
We refer to $\delta(s)$ as the slant function.
From \eqref{eq:skew_field}, if $\delta(s)\equiv0$ for all $s\in\spa$, then $Z$ is a Gaussian random process. Note that $Z$ is a random process with a consistent family of distribution functions, since $Z(s)=a(s)X(s)+b(s)Y(s)$ where
$a$ and $b$ are bounded functions and $X$ and $Y$ are random processes with a consistent family of distribution functions.
For any finite sequence of points $s_1,\ldots,s_d\in\spa$ 
the joint distribution of $Z(s_1), \ldots, Z(s_d)$ is $\skn_d(\bar{\Omega},\alpha,\tau)$, where
\begin{align}\label{eq:param}
\nonumber\bar{\Omega}&=D_\delta (\bar{\Sigma}+(D_\delta^{-1}\delta)(D_\delta^{-1}\delta)^\top)D_\delta\\
\alpha&=\{1+(D_\delta^{-1}\delta)^\top \bar{\Sigma}^{-1}(D_\delta^{-1}\delta)\}^{-1/2}D_\delta^{-1}\bar{\Sigma}^{-1}(D_\delta^{-1}\delta) \\
\nonumber\tau&=\{1+Q_{\bar{\Omega}}(\alpha)\}^{1/2}\,\varepsilon
\end{align}
and where $\bar{\Sigma}$ is the $d\times d$ correlation matrix of $X$, $\delta=(\delta(s_1),\ldots,\delta(s_d))^\top$ and  
$D_\delta=\{1_d-\text{diag}(\delta^2)\}^{1/2}$, where $1_d$ is the identity matrix 
\citep[][Ch. 5]{azzalini2013skew}. As a result, for any lag $h\in \real^k$, the distributions of
$\{Z(s_1),\ldots,Z(s_d)\}$ and $\{Z(s_1+h),\ldots,Z(s_d+h)\}$ will differ unless $\delta(s)=0$ for all $s\in\spa$. Hence, the distribution of $Z$ is not translation invariant and the process is not strictly stationary.
For $s\in\spa$ and $h\in\real^k$, the mean $m(s)$ and covariance function $c_s(h)$ 
of the skew-normal random process are
\begin{equation*}\label{eq:expect_ext_skew}
m(s)=\expect\{Z(s)\}=\delta(s)\phi(\varepsilon)/\Phi(\varepsilon)
\end{equation*}
and
\begin{equation}\label{eq:cov_ext_skew}
c_s(h)=\Cov\{Z(s), Z(s+h)\}=\rho(h)\sqrt{\{1-\delta^2(s)\}\{1-\delta^2(s+h)\}}+\delta(s)\delta(s+h)(1-r),
\end{equation}
where 
$
r = \left\{\frac{\phi(\varepsilon)}{\Phi(\varepsilon)}
\left(\varepsilon+\frac{\phi(\varepsilon)}{\Phi(\varepsilon)}\right)\right\}.
$
Hence, the mean is not constant and
the covariance does not depend only on the lag $h$, unless $\delta(s)=\delta_0\in(-1,1)$ for all $s\in\spa$. In the latter case
the skew-normal random process is weakly stationary \citep{zhang2010spatial}. 

One benefit of working with a skew-normal random field is that
the non-stationary covariance function \eqref{eq:cov_ext_skew} is positive definite if the covariance function of $X$
is positive definite, and if $-1<\delta(s)<1$ for all $s\in\spa$.
Hence, a valid model is directly obtainable by means of standard parametric correlation models $\rho(h)$ and any bounded function $\delta$ in $(-1,1)$.
If the Gaussian process correlation function satisfies  $\rho(0)=1$ and $\rho(h)\rightarrow 0$ as $\|h\|\rightarrow +\infty$,
then the correlation of the skew-normal process satisfies $\rho_s(0)=1$ and
$$
\rho_s(h)=\frac{c_s(h)}{\sqrt{c_s(0)c_s(h)}}\approx\frac{\delta(s)\delta(s+h)(1-r)}{\sqrt{(1-\delta^2(s)r)(1-\delta^2(s+h)r)}},
$$
as $\|h\|\rightarrow+\infty$.
Hence $\rho_s(h)=0$ if either $\delta(s)$ or $\delta(s+h)$ are zero. Conversely,  
if both $\delta(s)\rightarrow \pm 1$ and $\delta(s+h)\rightarrow \pm 1$ then $\rho_s(h)\rightarrow \pm 1$.

The increments $Z(s+h)-Z(s)$ are skew-normal distributed for any fixed $s\in\spa$ and $h\in\real^k$
\citep[see][Ch. 5]{azzalini2013skew}
and the variogram $2\gamma_{s}(h)=\variance\{Z(s+h)-Z(s)\}$ is equal to
$$
2\gamma_{s}(h)=2\left(
1-c_s(h)-\frac{\delta^2(s+h)+\delta^2(s)}{2/r}
\right).
$$
When $h=0$ the variogram is zero, and when $\|h\|\rightarrow +\infty$ the variogram  approaches a constant $\leq 2$,
respectively resulting in spatial independence or dependence for large distances $h$.
We can now infer the conditions required so that $Z(s)$ has a continuous sample path.
\begin{prop}\label{pro:con_der} 
Assume that $\spa\subseteq\real$. A skew-normal process $\{Z(s), s\in\spa\}$ has a continuous sample path
if $\delta(s+h)-\delta(s)=o(1)$ and
$1-\rho(h)=O(|\log |h||^{-a})$ for some $a>3$, as $h\rightarrow 0$. 
\end{prop}
This result follows by noting that $r_s(h)=\rho(h)+\delta^2(s)(1-\rho(h))+o(1)$
as $h\rightarrow 0$ and this is a consequence of the continuity assumption on $\delta(s)$, where
$r_s(h)=c_s(h)+r\{\delta^2(s+h)+\delta^2(s)\}/2$. Therefore, $1-r_s(h)=O(|\log|h||^{-a})$
as $h\rightarrow 0$. Thus, the proof follows from the results in \citet[][page 48]{lindgren2012}.
This means that  continuity of the skew-normal process is assured if $\delta(s)$ is
a continuous function, in addition to the usual condition on the correlation function of the
generating Gaussian process \citep[e.g.][Ch. 2]{lindgren2012}.

Figure \ref{fig01} illustrates trajectories of the skew-normal process for $k=1$,  
with $X(s)$ a zero mean unit variance Gaussian process on $[0,1]$ with isotropic power-exponential correlation function
\begin{equation}\label{eq:corr_power_exp}
\rho(h;\vartheta)=
\exp
\{-
\left(
h/\lambda
\right)^{\xi}
\}, \quad \vartheta=(\lambda,\xi),\: \lambda>0,\; 0<\xi\leq2,\; h>0,
\end{equation}
with $\xi=1.5$, $\lambda=0.3$ and $h\in[0,1]$. 
\begin{figure}[t!]
\centering 
\makebox{\includegraphics[width=12cm]{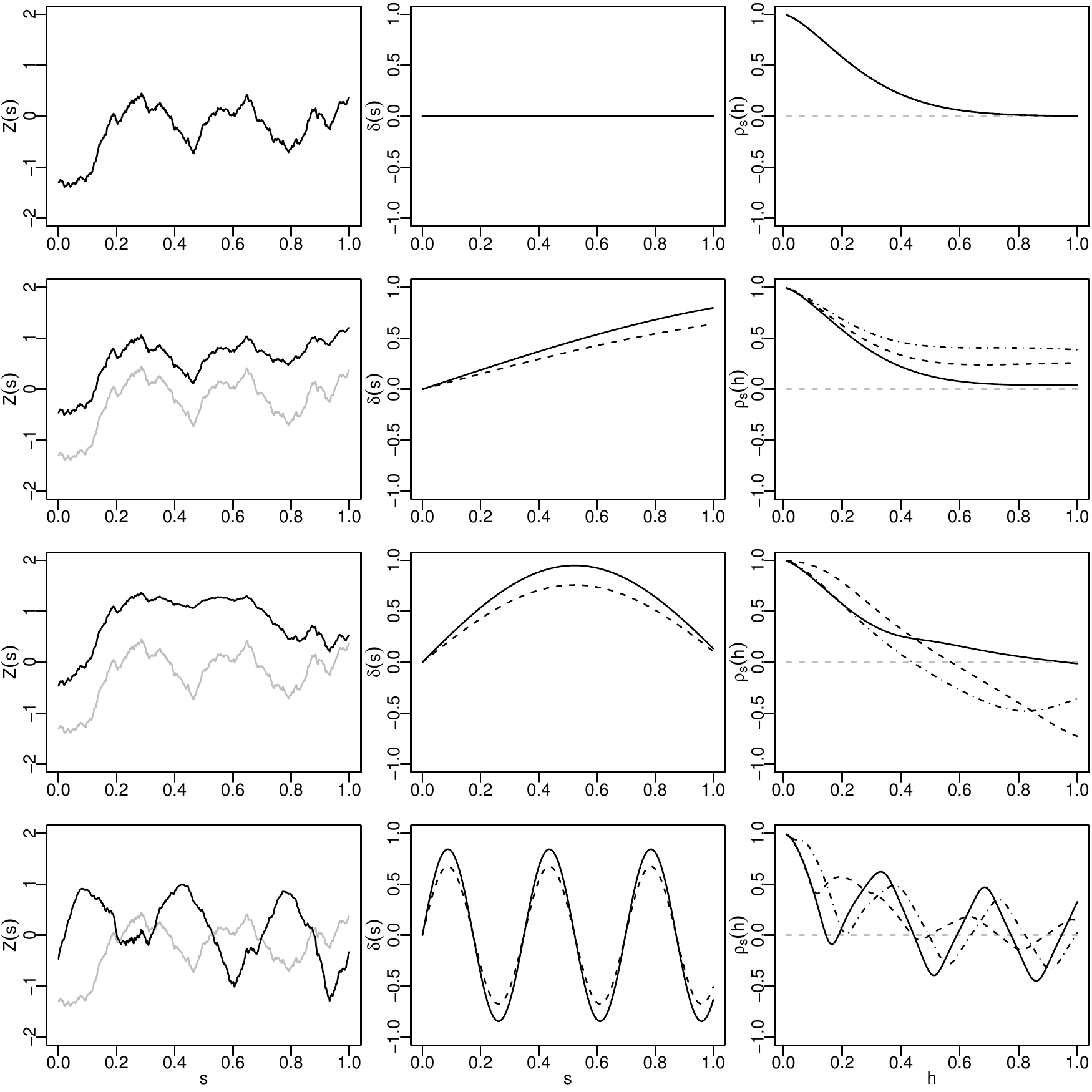}}
\caption{\label{fig01}\small Simulations from four univariate skew-normal random processes on $[0,1]$ with
$\varepsilon = 0$.
The left column shows the sample path (solid line) of the simulated 
process $Z(s)$ and of the generating Gaussian process $X(s)$ (grey line).
The middle column illustrates the slant function $\delta(s)$ (solid line) and the 
mean $m(s)$ of the process (dashed line).
The right column displays the non-stationary correlation functions at locations  
$s =  0.1$ (solid line), $0.5$ and $0.75$ (dot-dash). 
Rows 1--3 use slant function $\delta(s) = a \sin (bs)$
with $a = 0.95$ and $b = 0,1$ and $3$ respectively, whereas  row 4 uses   
$\delta(s) = a^2 \sin (bs) \cos (bs)$ with $a = 1.3$ and $b = 0.9$.}
\end{figure}
The first row shows the standard stationary case.
The second row illustrates the non-stationary correlation function obtained with $s=0.1$ (solid line) behaving close to the stationary correlation, however decaying more slowly as $s$ increases and approaching, but not reaching zero exactly.
The third row demonstrates both that points may be negatively correlated
and that $\rho_s(h)$ is not necessarily a decreasing function in $h$.
The bottom row highlights this even more clearly -- correlation functions need not be monotonically decreasing -- implying that
pairs of points far apart can be more dependent than nearby points.

Simulating a skew-normal random process is computationally cheap through Definition \ref{def:skn_proc}, with the
 simulation of the required stationary Gaussian process achievable through many fast algorithms 
 \citep[e.g.,][]{wood1994, chan1997}.
Rather than relying on \eqref{eq:param}, for practical purposes, to directly simulate from a skew-normal process with given parameters $\alpha$, $\bar{\Omega}$ and $\tau$, a conditioning sampling approach can be adopted \citep[][Ch. 5]{azzalini2013skew}.

Specifically, let $X(s)$ define a zero-mean, unit variance stationary Gaussian random field on $\spa$ with correlation function $\omega(h)=\expect\{X(s)X(s+h)\}$ and let $\bar{\Omega}$ be the 
$d\times d$ correlation matrix of $X(s_1),\ldots,X(s_d)$. Specify $\alpha: \spa \mapsto \real$ to be a continuous 
square-integrable function and let $\langle\alpha, X\rangle = \int_{\spa}\alpha(s)X(s)\,\der s$
be  the inner product. Let $X'$ be a standard normal random variable independent of $X$ and $\tau\in \real$. If  we define
\begin{equation}\label{eq:skew_rf_cond}
Z(s)=\left\{X(s) \vert \langle\alpha, X\rangle > X'-\tau\right\},\qquad s\in\spa
\end{equation}
then, for any finite set $s_1,\ldots,s_d\in\spa$, the distribution of $Z(s_1),\ldots,Z(s_d)$ is $\skn(\bar{\Omega},\alpha,\tau)$, where 
$\alpha\equiv\{\alpha(s_1),\ldots,\alpha(s_d)\}$. For simplicity we also refer to $\alpha(s)$ as the slant function.
More efficient simulation of skew-normal processes can be achieved by considering the form
$Z(s)= X(s)$  if $\langle\alpha, X\rangle > X'-\tau$ and $Z(s)=-X(s)$ otherwise \cite[e.g.][Ch. 5]{azzalini2013skew}.

We conclude this section by discussing some extremal properties of the skew-normal process $Z(s)$.
For a finite sequence of points $s_1,\ldots,s_d\in\spa$, with $d\geq 2$. 
Each margin $Z(s_i)$ follows a skew-normal distribution \citep{azzalini2013skew} and so is
in the domain of attraction of a Gumbel distribution \citep{chang2007, padoan2011}.
Further, each pair $(Z(s_i),Z(s_j))$ is asymptotically independent \citep{bortot2010, lysenko2009}.
However, in this case a broad class of tail behaviours can still be obtained by assuming that the
joint survival function is regularly varying at $+\infty$ with index $-1/\eta$ \citep{ledford1996}, so that
\begin{equation}\label{eq:asymptotic_inp}
 \prob(Z(s_i)>x,Z(s_j)>x)=x^{-1/\eta}\, \cL(x), \qquad x \rightarrow +\infty,
\end{equation}
where $\eta\in(0,1]$ is the coefficient of tail dependence and $\cL(x)$ is a  slowly varying function i.e., $\cL(ax)/\cL(x)\rightarrow 1$ as $x\rightarrow+\infty$, for fixed $a>0$.
Considering $\cL$ as a constant, at extreme levels margins 
are negatively associated when $\eta<1/2$, independent when $\eta=1/2$ and positively associated when $1/2<\eta< 1$. When
$\eta=1$ and $\cL(x)\nrightarrow0$  asymptotic dependence  is obtained.
We derive the asymptotic behavior of the joint survival function \eqref{eq:asymptotic_inp} 
for a pair of skew-normal margins. As our primary interest is in spatial applications, we focus on 
the joint upper tail of the skew-normal distribution when the variables are positively correlated or uncorrelated.
\begin{prop}\label{pro:tails_sn}
Let $Z\sim\skn_2(\bar{\Omega},\alpha)$, where $\alpha=(\alpha_1,\alpha_2)^\top$ and  $\bar{\Omega}$ is a correlation matrix with 
off-diagonal term $\omega\in[0,1)$.
The joint survivor function of the bivariate skew-normal distribution with unit Fr\'{e}chet margins 
behaves asymptotically as \eqref{eq:asymptotic_inp}, where:
\begin{enumerate}
\item \label{case:a} when either  $\alpha_1,\alpha_2\geq0$,  or $\omega>0$ and $\alpha_j\leq 0$ and $\alpha_{3-j}\geq-\omega^{-1}\alpha_j$ for $j=1,2$, then
\begin{flushleft}
\begin{tabular}{l}
$\eta=(1+\omega)/2$, \qquad 
$\cL(x)=\frac{2\,(1+\omega)}{1-\omega}(4\pi\log x)^{-\omega/(1+\omega)};$
\end{tabular}
\end{flushleft}
\item \label{case:b} when $\omega>0$, $\alpha_j< 0$, and $-\omega\,\alpha_j\leq\alpha_{3-j}<-\omega^{-1}\alpha_j$, for  $j=1,2$, then 
\begin{enumerate}
\item If  $\alpha_{3-j}>-\alpha_j/\bar{\alpha}_{j}$ then
\begin{flushleft}
\begin{tabular}{l}
$\eta=\frac{(1-\omega^2)\bar{\alpha}_{j}^2}{1-\omega^2+(\bar{\alpha}_{j}-\omega)^2}$,\qquad
$\cL(x)=\frac{2\,\bar{\alpha}_{j}^2(1-\omega^2)}
{(\bar{\alpha}_{j}^2-\omega)(1-\omega\bar{\alpha}_{j})}(4\pi\log x)^{1/2\eta-1};$
\end{tabular}
\end{flushleft}
\item If $\alpha_{3-j}<-\alpha_j/\bar{\alpha}_{j}$ then
\begin{flushleft}
\begin{tabular}{l}
$\eta=\left[
\frac{1-\omega^2+(\bar{\alpha}_{j}-\omega)^2}
{(1-\omega^2)\bar{\alpha}_{j}^2}
+\left(
\alpha_{3-j}+\frac{\alpha_j}{\bar{\alpha}_{j}}
\right)^2
\right]^{-1}$,\\\\
$\cL(x)=\frac{-2^{3/2}\pi^{1/2}\bar{\alpha}_{j}^2(1-\omega^2)(\alpha_{3-j}+\alpha_j/\bar{\alpha}_{j})^{-1}}
{(\bar{\alpha}_{j}-\omega)
\{
1-\omega\bar{\alpha}_{j}+\alpha_j(\alpha_j+\alpha_{3-j}\bar{\alpha}_{j})(1-\omega^2)
\}}(4\pi\log x)^{1/2\eta-3/2};$
\end{tabular}
\end{flushleft}
\end{enumerate}
\item \label{case:c} when either $\alpha_1,\alpha_2<0$, or $\omega>0$, $\alpha_j<0$ and $0<\alpha_{3-j}<-\omega\,\alpha_j$ for $j=1,2$, then
\begin{flushleft}
\begin{tabular}{l}
$\eta=\left\{
\frac{1}{1-\omega^2}
\left(
\frac{\alpha_{3-j}^2(1-\omega^2)+1}{\bar{\alpha}_{3-j}^2}+
\frac{\alpha_j^2(1-\omega^2)+1}{\bar{\alpha}_{j}^2}+
\frac{2(\alpha_{3-j}\alpha_j(1-\omega^2)-\omega)}{\bar{\alpha}_{3-j}\bar{\alpha}_{j}}
\right)
\right\}^{-1}$,\\\\
$\cL(x)=\frac{-2^{3/2}\pi^{1/2}\bar{\alpha}_{j}^{3/2}\bar{\alpha}_{3-j}^2(1-\omega^2)
(\alpha_i\bar{\alpha}_{j}+\alpha_j\bar{\alpha}_{3-j})^{-1}}
{(\bar{\alpha}_{j}-\omega\bar{\alpha}_{3-j})
\{
1-\omega\bar{\alpha}_{j}+\alpha_j(\alpha_j+\alpha_{3-j}\bar{\alpha}_{j}/\bar{\alpha}_{3-j})(1-\omega^2)
\}}(4\pi\log x)^{1/2\eta-3/2};$
\end{tabular}
\end{flushleft}
\end{enumerate}
where 
$
\bar{\alpha}_{j}=\sqrt{1+\alpha^{*2}_{j}}$
and
$\alpha^*_{j}:=\alpha^*_{\{j\}}=\frac{\alpha_j+\omega\alpha_{3-j}}{\sqrt{1+\alpha_{3-j}(1-\omega^2)}}.
$

Proof in Appendix \ref{ssec:tails_sn}.
\end{prop}

As a result, when both marginal parameters are non-negative (case 1) then $1/2\leq \eta <1$, with $\eta=1/2$ occurring when $\omega=0$.
As a consequence, as for the Gaussian distribution (for which $\alpha=0$) the marginal extremes are either positively associated or exactly independent.
The marginal extremes are also completely dependent when $\omega=1$, regardless of the values of the slant parameters, 
$\alpha$.
When one marginal parameter is positive and one is negative (case 2) then
$\eta>(1+\omega)/2$. 
In this case the extreme marginals are also positively associated, but the dependence is greater than when the random variables
are normally distributed.
Finally, when both marginal parameters are negative (case 3), then $0<\eta<1/2$ implying 
that the extreme marginals are negatively associated, although $\omega>0$. 
It should be noted that differently from the Gaussian case ($\alpha=0$) where $\omega>0$ implies a positive association,
in this case it is not necessarily true.
In summary, the degree of dependence in the upper tail of the skew-normal distribution ranges from negative to positive association and including independence.

%
%
%

%
\section{Spectral representation for the extremal-skew-$t$ process}\label{sec:spectral}

The spectral representation of stationary max-stable processes with common unit Fr\'{e}chet margins can be constructed using 
the fundamental procedures introduced by \cite{dehaan1984} and \cite{schlather2002a}  \cite[see also][Ch. 9]{dehaan2006}. 
This representation can be formulated in broader terms resulting in max-stable processes with $\nu$-Fr\'{e}chet univariate marginal distributions, with $\nu>0$ \citep{opitz2013}.
In order to state our result we rephrase the spectral representation so to also take into 
account non-stationary processes.

Let $\{Y(s)\}_{s \in \spa}$ be a non-stationary real-valued stochastic process with continuous sample path on $\spa$ such that
$
\expect\left\{\sup_{s\in\spa}Y(s)\right\}<\infty$ and $m^+(s) =\expect[\{Y^+ (s)\}^{\nu}]<\infty, \forall s\in \spa
$
for $\nu>0$,
where $Y^+ ( \cdot ) = \max\{Y (\cdot) ,0\}$ denotes the positive part of $Y$.
Let $\{R_i\}_{i\geq1}$ be the points of an inhomogeneous Poisson point process on $(0,\infty)$ 
with intensity
$\nu r^{-(\nu+1)}$, $\nu>0$, which are independent of $Y$.
Define
\begin{equation}\label{eq:max-stable}
U (s) = \max_{i=1,2,\ldots} \{R_i Y_i^+ (s)\} /\{m^+(s)\}^{1/\nu},\quad s \in \spa,
\end{equation}
where $Y_1,Y_2,\ldots$ are iid copies of $Y$.
Then $U$ is a max-stable process with common $\nu$-Fr\'{e}chet univariate margins. In particular, for fixed $s\in\spa$ and $x(s)>0$ we have
$$
\prob(U(s)\leq x(s))=\exp\left[-\frac{\expect\{Y^+(s)\}^\nu}{x^\nu(s)m^+(s)}\right]=\exp\{-1/x^\nu(s)\},
$$
and for fixed $s_1,\ldots,s_d$
the finite dimensional distribution of $U$ has exponent function 
\begin{equation}\label{eq:expo_mea}
V(x(s_1),\ldots,x(s_d))=\expect\left(\max_{j}\left[\frac{\{Y^+(s_j)/x(s_j)\}^\nu}{m^+(s_j)}\right]\right),\quad x(s_j)>0,\, j=1,\ldots,d 
\end{equation}
\cite[][Ch. 9]{dehaan2006}.

In this construction, the impact of a non-stationary process $Y(s)$ would be that the dependence structure of the max-stable process
$U(s+h)$ depends on both the separation $h$ and the location $s\in\spa$, and would therefore itself be non-stationary.
The below theorem derives a max-stable process $U(s)$ when $Y(s)$ is the skew-normal random field introduced in Section \ref{sec:extskew}.

\begin{theo}[Extremal skew-$t$ process]\label{teo:main_res}
Let $Y(s)$ be a skew-normal random field on $s\in\spa$ with finite dimensional distribution $\skn_d(\bar{\Omega},\alpha,\tau)$, as defined in equation \eqref{eq:skew_rf_cond}. 
Then the max-stable process $U(s)$, given by \eqref{eq:max-stable}, has $\nu$-Fr\'{e}chet univariate marginal distributions and exponent function
\begin{equation}\label{eq:ncextst_expo}
V(x_j,j\in I)=
\sum_{j=1}^d x_j^{-\nu} \Psi_{d-1}
\left(\left(
\sqrt{\frac{\nu+1}{1-\omega^2_{i,j}}}
\left(
\frac{x^\circ_i}{x^\circ_j} - \omega_{i,j}
\right),i\in I_j\right)^\top; \bar{\Omega}^\circ_j, \alpha^\circ_j, \tau^\circ_j, \kappa^\circ_j,\nu+1
\right),
\end{equation}
where $x_j\equiv x(s_j)$, $\Psi_{d-1}$ is a $(d-1)$-dimensional non-central extended skew-$t$ distribution (Definition \ref{def:noncen_skew}) with 
correlation matrix $\bar{\Omega}^\circ_j$, shape, extension and non-centrality parameters $\alpha^\circ_j, \tau^\circ_j$ and $\kappa^\circ_j$, $\nu+1$ degrees of freedom,
$I=\{1,\ldots,d\}$, $I_j=I\backslash\{j\}$,
and $\omega_{i,j}$ is the $(i,j)$-th element of $\bar{\Omega}$.

Proof (and further details) in Appendix \ref{ap:proof_th_1}. 
\end{theo}
We call the process $U(s)$ with exponent function 
\eqref{eq:ncextst_expo} an extremal skew-$t$ process.

Note that in Theorem \ref{teo:main_res} when $\tau=0$, and the slant function is such that $\alpha(s)\equiv 0$ for all $s\in\spa$, then the exponent function \eqref{eq:ncextst_expo} becomes
\begin{equation}\label{eq:expo_ext_t_paper}
V(x_j,j\in I)=\sum_{j\in I}
x_j^{-\nu}
\Psi_{d-1}
\left[  
\left(
\sqrt{\frac{\nu+1}{1-\omega^2_{i,j}}}
\left(
\frac{x_i}{x_j}- \omega_{i,j}
\right),i\in I_j\right)^\top; \bar{\Omega}^\circ_{j},\nu+1
\right].
\end{equation}
This is the exponent function of the extremal-$t$ process as discussed in \citet{opitz2013}. 

If we assume $\tau=0$ in \eqref{eq:skew_rf_cond}, then the bivariate exponent function of the extremal skew-$t$ process seen as a function of the separation $h$ is equal to
\begin{equation*}
V\{x(s),x(s+h)\}=
\frac{\Psi(b(x^*_s(h));\alpha^*_s(h),\tau^*_s(h),\nu+1)}{x^\nu(s)}+
\frac{\Psi(b(x^+_s(h));\alpha^+_s(h),\tau^+_s(h),\nu+1)}{x^\nu(s+h)}
\end{equation*}
where $\Psi$ is a univariate extended skew-$t$ distribution,
$
b(\cdot)=\sqrt{\frac{\nu+1}{1-\omega^2(h)}}(\cdot-\omega(h)),
$

\begin{equation*}
\begin{array}{rclrcl}
x^*_s(h) & = & \frac{x(s+h)\Gamma_s(h)}{x(s)}, & x^+_s(h) & = & \frac{x(s)}{x(s+h)\Gamma_s(h)},\\\\
\alpha^*_s(h)&=&\alpha(s+h)\sqrt{1-\omega^2(h)}, & \alpha_s^+(h)&=&\alpha(s)\sqrt{1-\omega^2(h)},\\\\
\tau_s^*(h)&=&\sqrt{\nu+1}\{\alpha(s)+\alpha(s+h)\omega(h)\},& \tau_s^+(h)&=&\sqrt{\nu+1}\{\alpha(s+h)+\alpha(s)\omega(h)\},\\
\end{array}
\end{equation*}
and
$$
\Gamma_s(h)=
\left(
\frac{\Psi\left[\alpha(s)+\alpha(s+h)\omega(h)\sqrt{\frac{\nu+1}{\alpha^2(s+h)\{1-\omega^2(h)\}}};\nu+1\right]}
{\Psi\left[\alpha(s+h)+\alpha(s)\omega(h)\sqrt{\frac{\nu+1}{\alpha^2(s)\{1-\omega^2(h)\}}};\nu+1\right]}
\right)^{1/\nu}.
$$
Clearly, as the dependence structure depends on both correlation function $\omega(h)$ and the slant function $\alpha(s)$, and
therefore on the value of $s\in\spa$, it is a non-stationary dependence structure. 
From the bivariate exponent function we can derive
the non-stationary extremal coefficient function, using the relation $\theta_s(h)=V(1,1)$, which gives
\begin{equation}\label{eq:extremal_coeff}
\theta_s(h)=
\Psi(b(\Gamma_s(h));\alpha^*_s(h),\tau^*_s(h),\nu+1)+
\Psi(b(1/\Gamma_s(h));\alpha^+_s(h),\tau^+_s(h),\nu+1).
\end{equation}
\begin{figure}[t!]
\begin{center}
\includegraphics[scale=0.32]{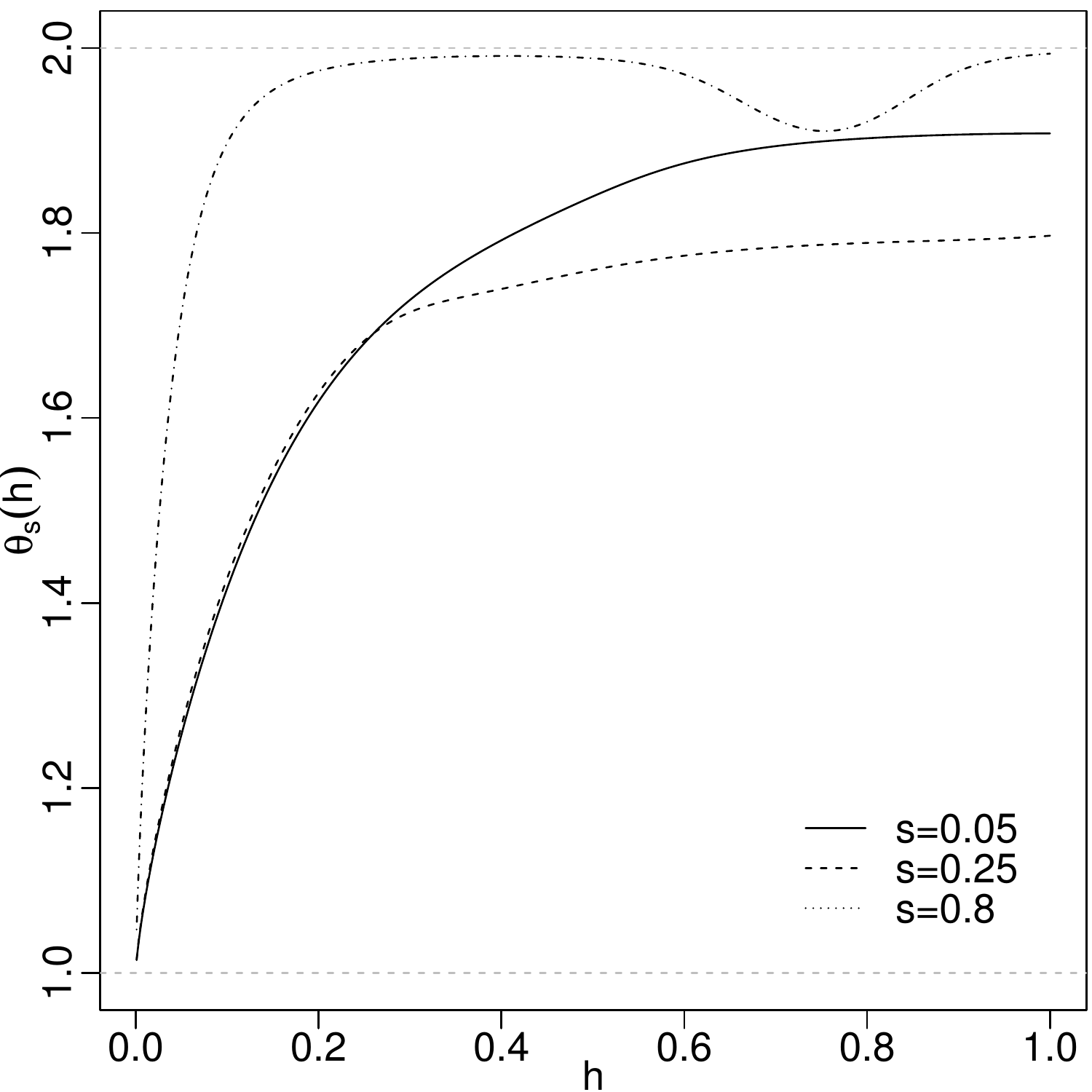}
\includegraphics[scale=0.32]{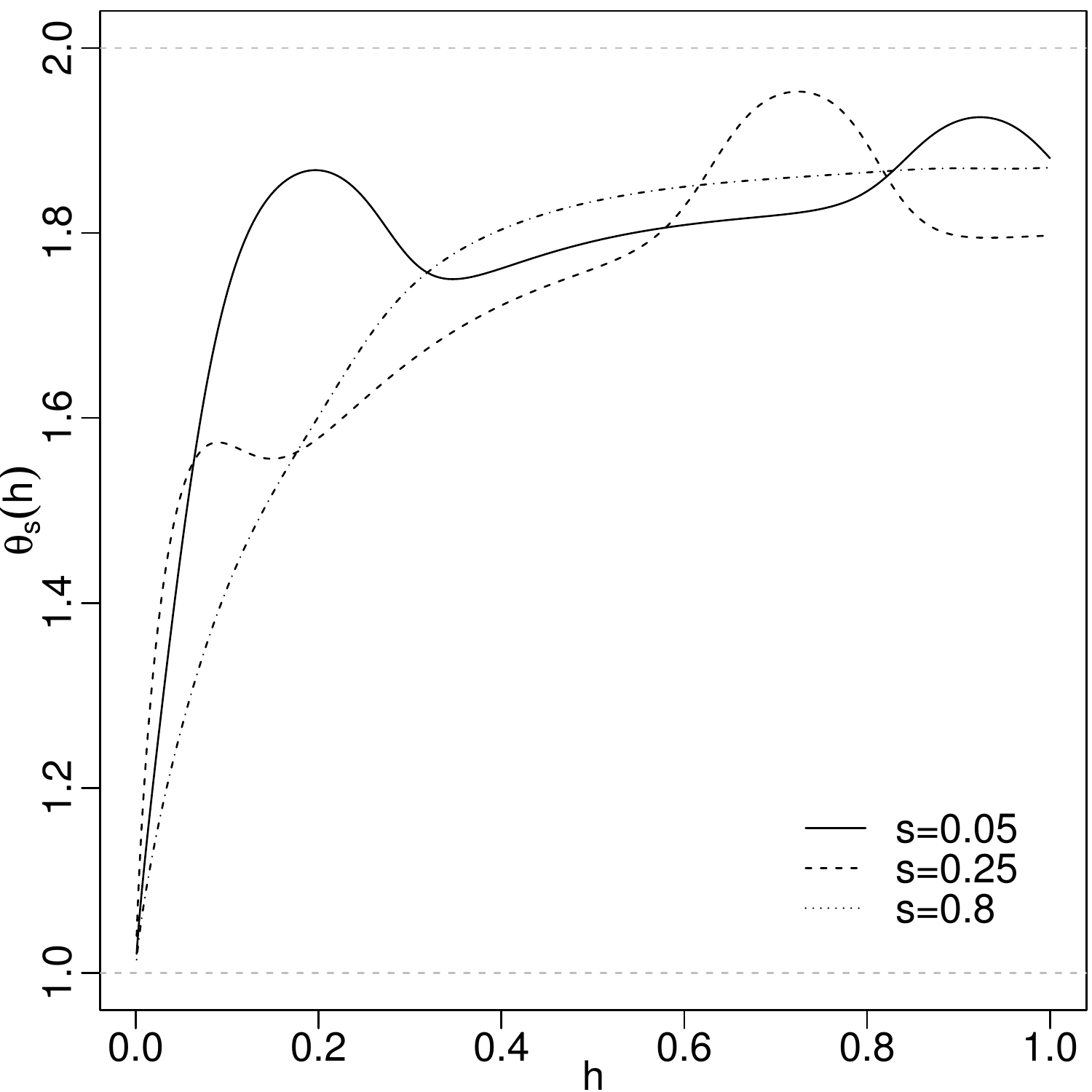}
\includegraphics[scale=0.32]{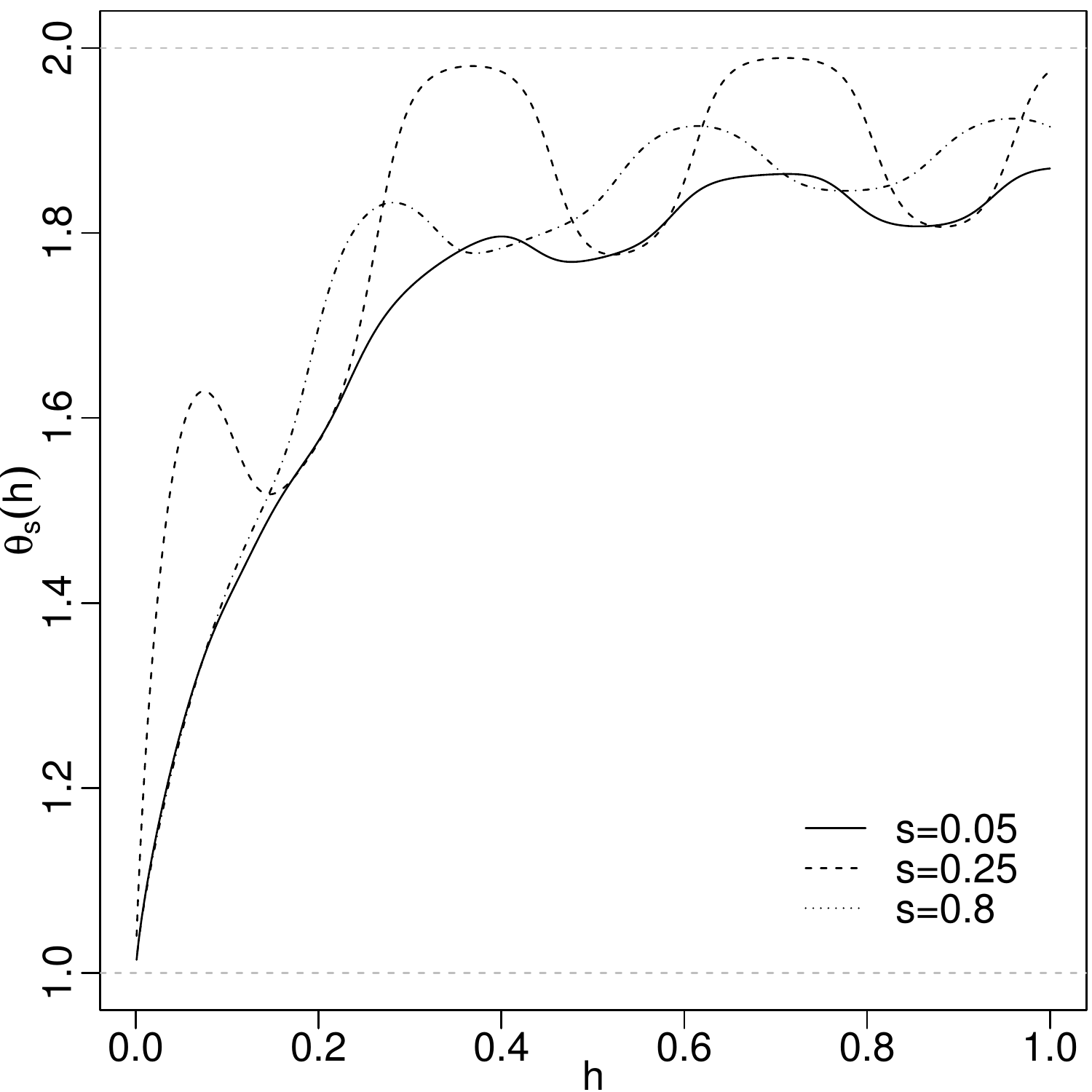}
\end{center}
\caption{\label{fig:iso_nonst_ec}\small 
Examples of univariate ($k=1$) non-stationary isotropic extremal coefficient functions $\theta_s (h)$, for the extremal skew-$t$ process over 
$s \in [0,1]$,  using correlation function (\ref{eq:corr_power_exp}) where $h \in [0,1]$, $\lambda = 1.5$
and $\gamma = 0.3$. 
Slant functions are  (left to right panels): $\alpha(s) = -1-s+\exp \{ \sin(5s) \}, \alpha(s)=1+1.5s- \exp\{ \sin(8s)\}$
and $\alpha(s)= 2.25 \sin (9s) \cos(9s)$. Solid, dashed and dot-dashed lines represent 
the fixed locations $s = 0.05, 0.25$ and $0.8$ respectively.}
\end{figure}
\begin{figure}[h!]
\begin{center}
\includegraphics[scale=0.32]{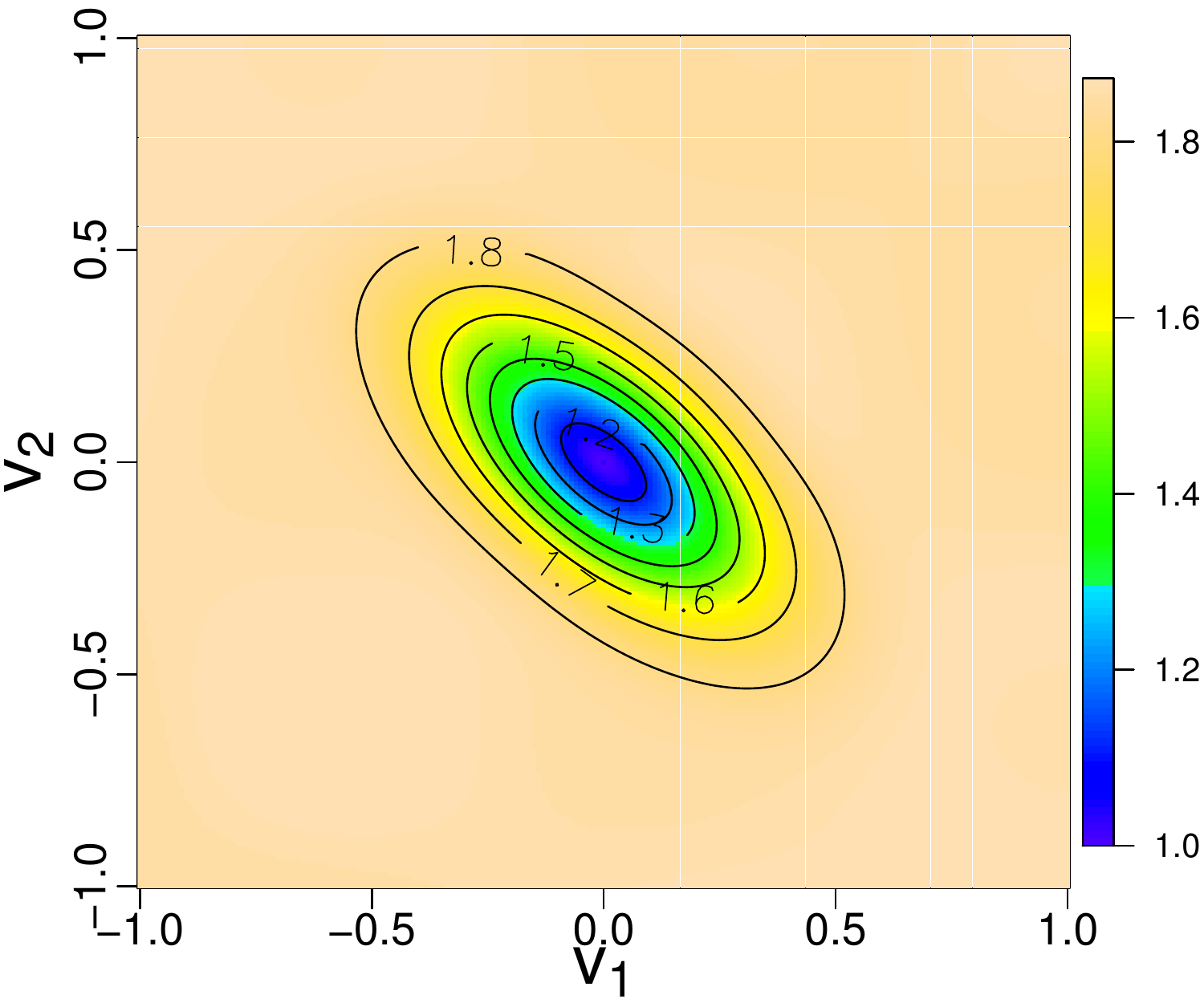}
\includegraphics[scale=0.32]{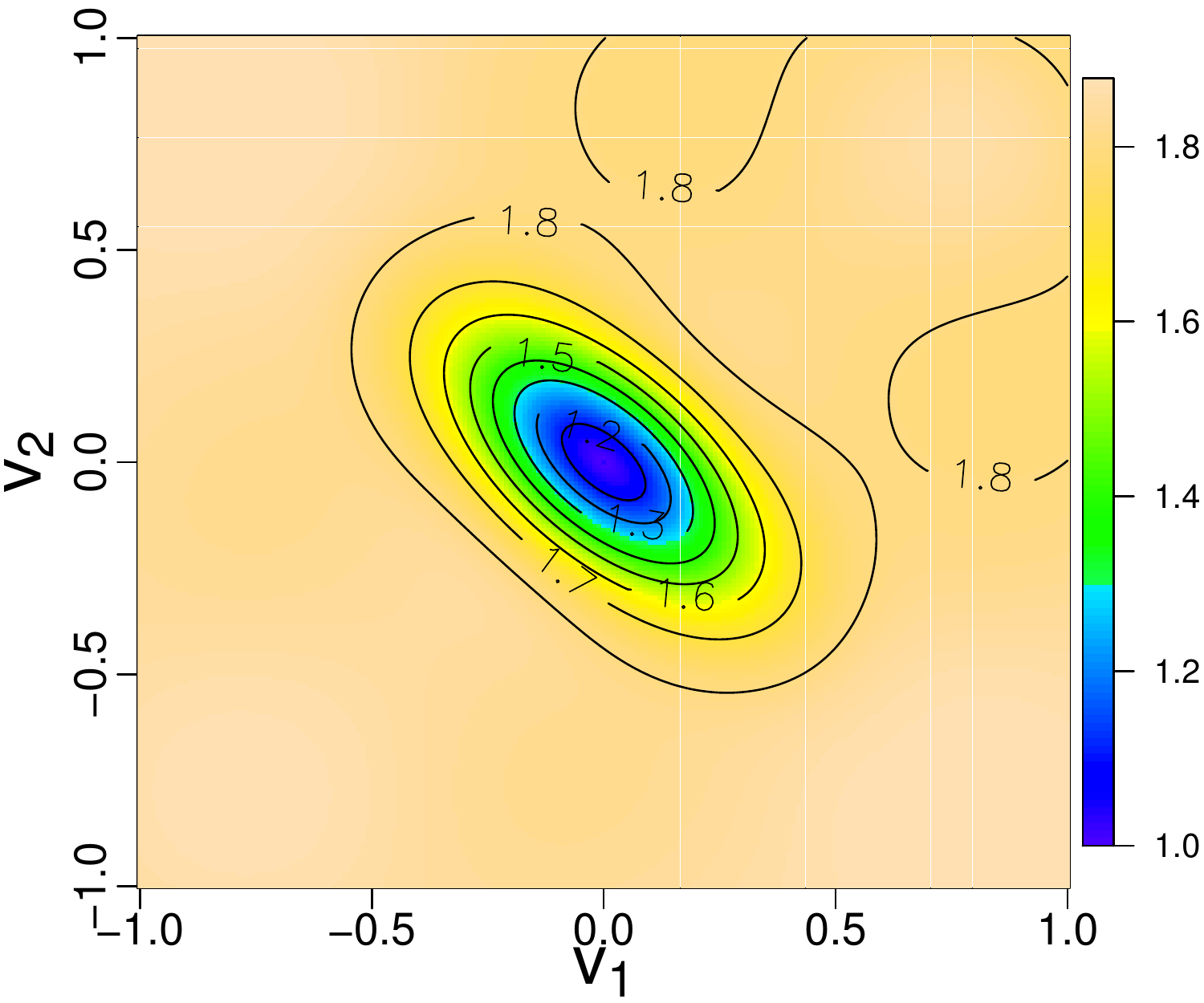}
\includegraphics[scale=0.32]{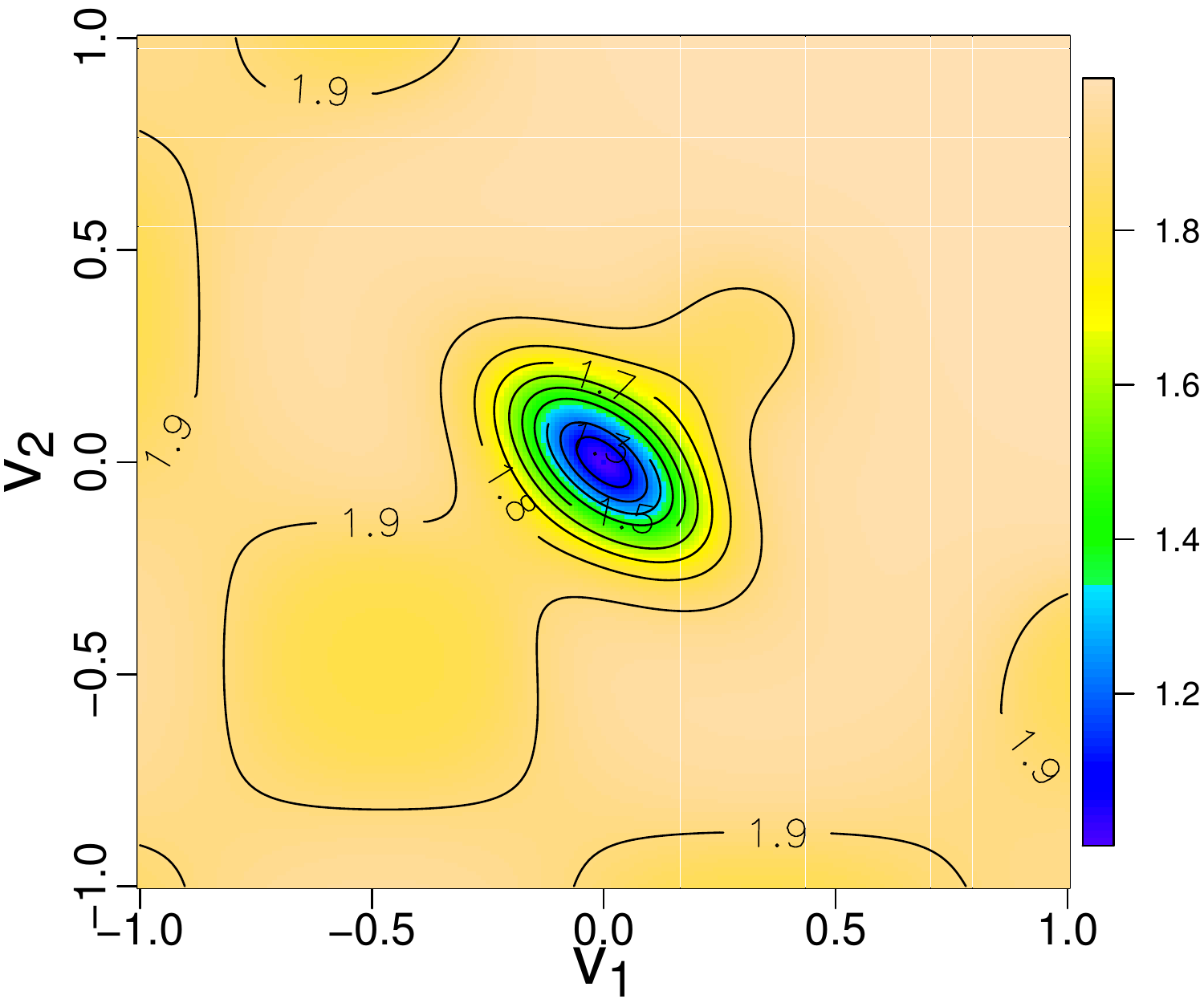}\\
\includegraphics[scale=0.32]{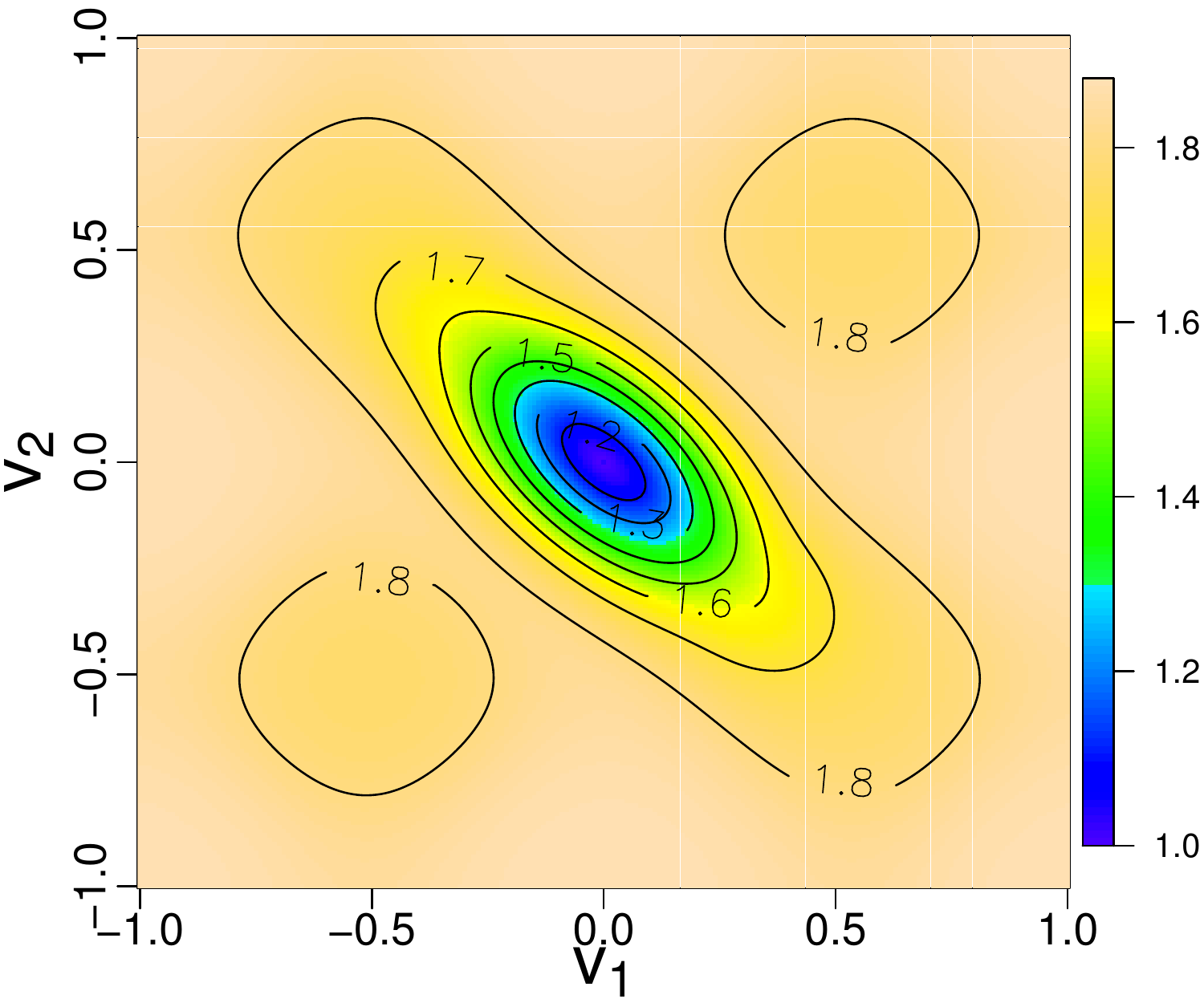}
\includegraphics[scale=0.32]{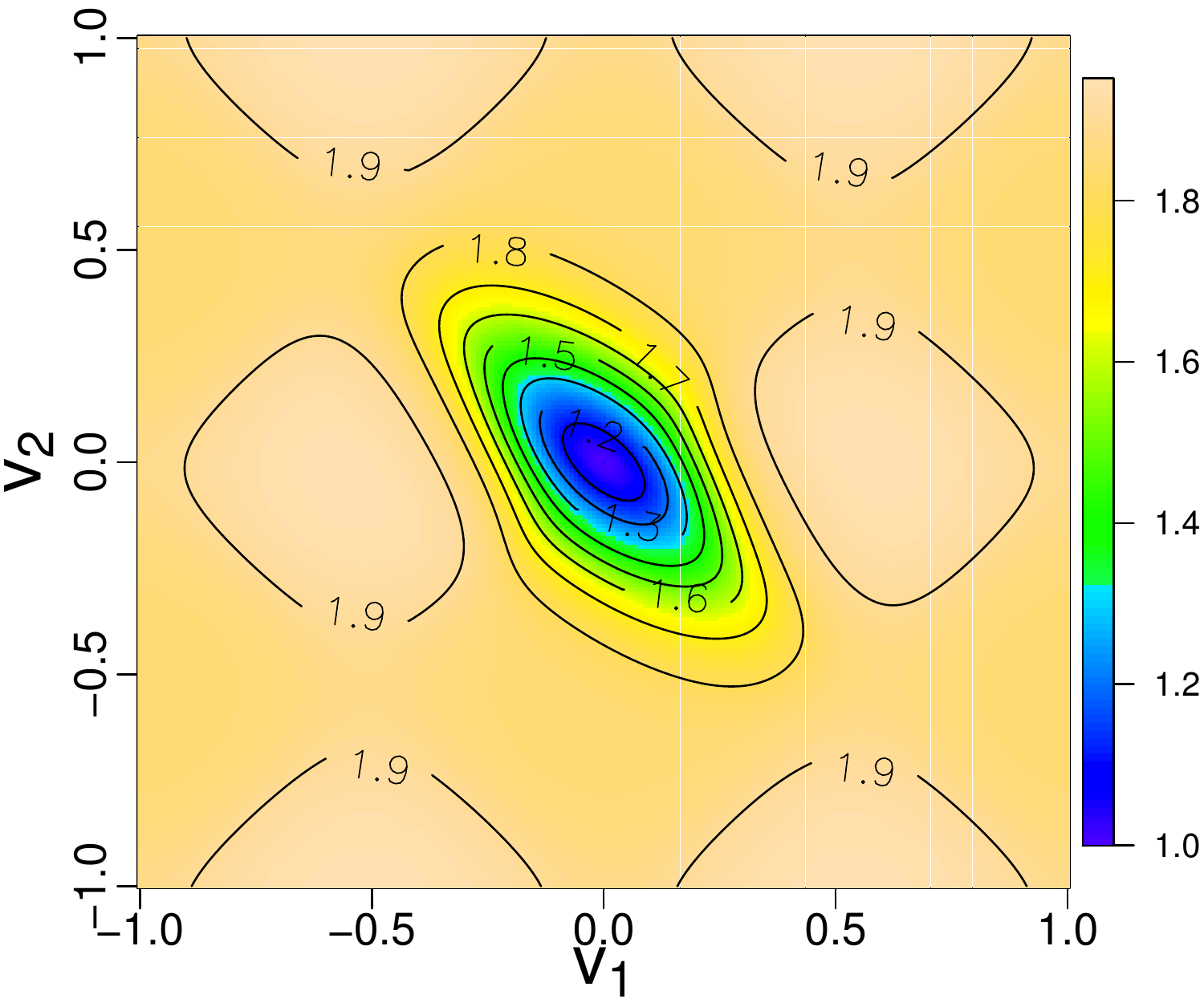}
\includegraphics[scale=0.32]{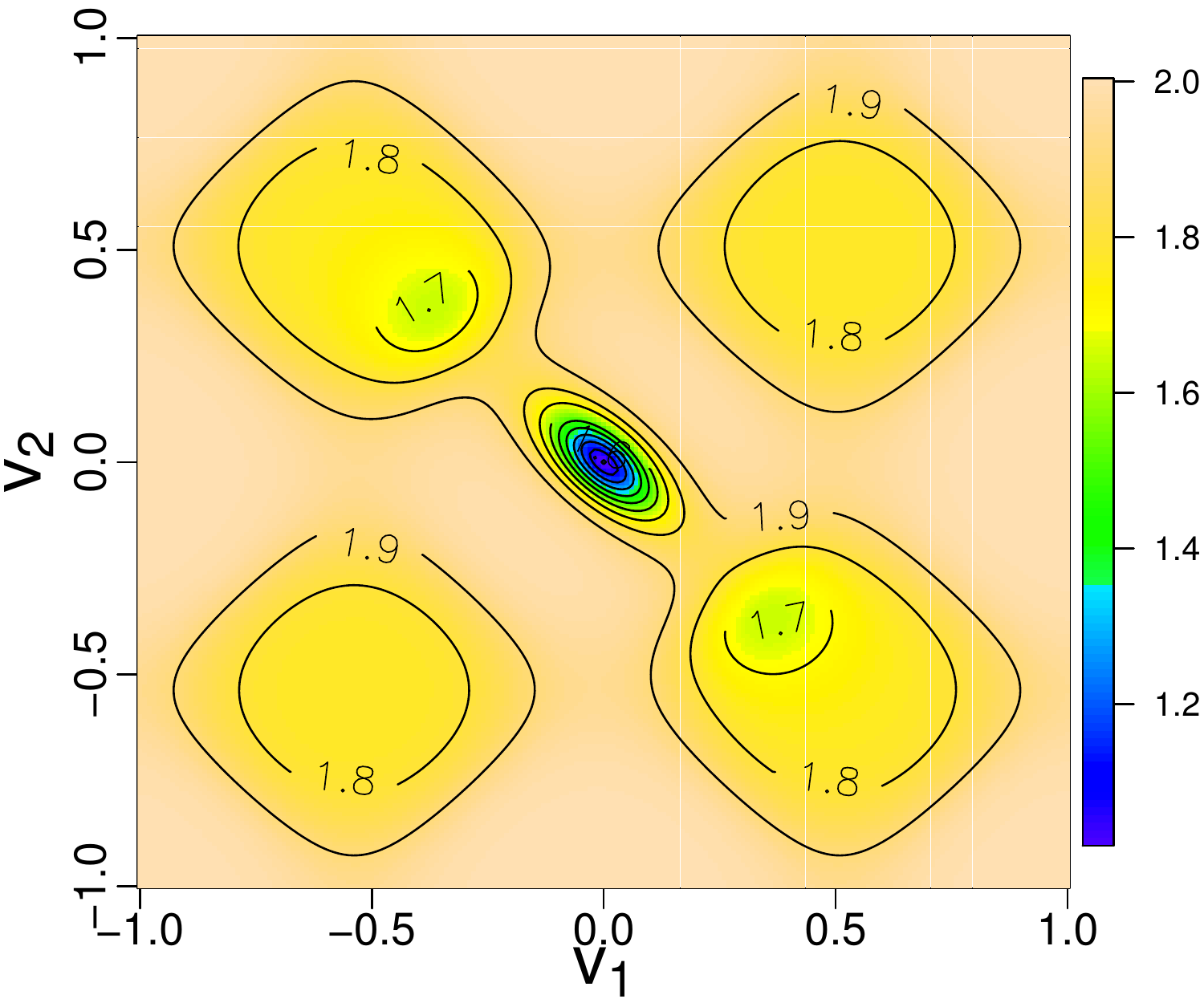}
\end{center}
\caption{\label{fig:aniso_nonst_ec}\small
Bivariate ($k=2$) geometric anisotropic non-stationary extremal coefficient functions $\theta_s (h)$, for the extremal skew-$t$ process on $s \in [0,1]^2$,  based on extremal coefficient 
function (\ref{eq:extremal_coeff}) with $\lambda = 1.5$ and $\gamma = 0.3,$ where $h=v^\top R v$, $v=(v_1,v_2)^\top\in[-1,1]^2$ and $R$ is a $2\times 2$ matrix whose diagonal elements are $2.5$ 
and off-diagonal elements $1.5$. 
Slant functions are $\alpha(s)=\exp\{\sin(4s_1)\sin(4s_2)-s_1s_2-1\}$ (top panels) and 
$\alpha(s)=2.25\{\sin(3s_1)\cos(3s_1)+\sin(3s_2)\cos(3s_2)\}$ (bottom), with $s=(s_1,s_2)^\top\in[0,1]^2$.
Left to right,  
panels are based on fixing 
$s=(0.2,0.2)^\top$, $s=(0.4,0.4)^\top$ and $s=(0.85,0.85)^\top$ (top panels)
and $s=(0.25,0.25)^\top$, $s=(0.25,0.8)^\top$ and $s=(0.8,0.8)^\top$ (bottom).
}
\end{figure}

Figure \ref{fig:iso_nonst_ec} shows some examples of univariate ($k=1$) non-stationary  isotropic extremal coefficient functions obtained from \eqref{eq:extremal_coeff} using
 the power-exponential correlation function \eqref{eq:corr_power_exp}.
Each panel illustrates a different slant function $\alpha(s)$, with the line-types indicating the fixed location value of $s\in\spa.$
The extremal coefficient functions $\theta_s(h)$
increase as the value of $h$ increases, meaning that the dependence of extremes decreases with the distance.
$\theta_s(h)$ grows with different rates depending on the location $s\in\spa$.
Although the ergodicity and mixing properties of the process must be investigated, numerical results show that
for some $s$, $\theta_s(h)\rightarrow 2$ as $|h|\rightarrow +\infty$.
By increasing the complexity of the slant function (e.g. centre and right panels) it is possible to construct extremal coefficient functions which exhibit stronger dependence for larger distances, $h$, compared to shorter distances.
Similarly Figure \ref{fig:aniso_nonst_ec} illustrates examples of  bivariate ($k=2$) non-stationary geometric anisotropic extremal coefficient functions, $\theta_s(h)$, also obtained from \eqref{eq:extremal_coeff}. 
Similar interpretations to the univariate case can be made (Figure \ref{fig:iso_nonst_ec}), in addition to noting that the level of dependence is affected by the direction (from the origin).

\section{Inference for extremal skew-$t$ processes}\label{sec:inference}

Parametric inference for the extremal-skew-$t$ process can be performed in two ways.
The first uses the marginal composite-likelihood approach \citep[e.g.][]{padoan2010, davison2012c, huser2013},
since only marginal densities of dimension up to $d=4$ are practically available
(see the Supporting Information).

Let $\vartheta\in\Theta\subseteq \mathbb{R}^p$, $p=1,2,\ldots$, 
denote the vector of dependence
parameters of the extremal-skew-$t$ process. 
Consider a sample $(x_i,i=1,\ldots,n$) with $x_i\in\mathbb{R}^d_+$ of $n$ iid replicates of the process
observed over a finite number of points $(s_j,j\in I)$ with $s_j\in\spa$. 
For simplicity, it is assumed that the univariate marginal distributions are unit Fr\'{e}chet.
The pairwise or triplewise  ($m=2,3$) log-composite-likelihood is defined by
$$
\ell_m(\vartheta;x)=\sum_{i=1}^n\sum_{E\in E_m}\log f(x_i\in E;\vartheta),\quad m=2,3,
$$ 
where $x=(x_1,\ldots,x_n)^\top$ with $x_i\in\mathbb{R}_+^{m}$ and $f$ is a marginal extremal-skew-$t$ pdf associated with each member of a set of marginal events $E_m$. See e.g. \cite{varin2011} for a complete description of composite likelihood methods.

A second approach is to use the approximate likelihood function introduced by \cite{coles1994}, which is constructed on the space of angular densities.
The angular measure of the extremal-skew-$t$ dependence model \eqref{eq:ncextst_expo} places
mass on the interior as well as on all the other subspaces of the simplex, such as the edges and the vertices.
We derive some of these densities following the results in \citet{coles1991}.

Let $J$ be an index set that takes values in $\index=\powset(\{1,\ldots,d\})\backslash\emptyset$, where $\powset(I)$ is the power set of $I$.
For any fixed $d$ and all $J\in\index$, the sets 
$$
\simplex_{d,J}=(w \in \simplex : w_j =0, \text{ if } j \notin J;\, w_j>0\text{ if } j\in J)
$$
provide a partition of the $d$-dimensional simplex $\simplex$ into $2^d-1$ subsets. Let $k=|J|$ be the size of $J$.
Let $h_{d,J}$ denote the density that lies on the subspace 
$\simplex_{d,J}$, which has $k-1$ free parameters $w_j$ such that $j\in J$. 
When $J=\{1,\ldots,d\}$ 
the angular density in the interior of the simplex is
\begin{equation} \label{eq:spec_den_ext_swt}
h(w) = 
\frac{
\psi_{d-1} 
\left( 
\left[
 \sqrt{ \frac{\nu+1}{1-\omega_{i,1}^2 } }
\left\{  
\left( \frac{w^\circ_i }{w^\circ_1}\right)^{1/\nu} -\omega_{i,1} 
\right\}, i\in I_1 \right]^\top; \Omega_1^\circ, \alpha_1^\circ, \tau_1^\circ, \kappa_1^\circ, \nu+1
\right) }
{w_1^{(d+1)} \left\{ \prod_{i=2}^d \frac{1}{\nu} \sqrt{ \frac{\nu+1}{1-\omega_{i,1}^2 } }
 \Big( \frac{w^\circ_i}{w^\circ_1}\Big )^{\frac{1}{\nu}-1} \frac{m_i^+}{m_1^+}
 \right\}^{-1}},\quad w\in\simplex
\end{equation}
where $\psi_{d-1}$ denotes the $d-1$-dimensional skew-$t$ density, $I_j=\{1,\ldots,d\}\backslash j$ and where the parameters $\Omega_1^\circ, \alpha_1^\circ, \tau_1^\circ, \kappa_1^\circ$ and $w^\circ_i=w_i(m_i^+)^{1/\nu}$
are given in the proof to Theorem \ref{teo:main_res} (Appendix \ref{ap:proof_th_1}).
When $J=\{i_1,\ldots,i_k\}\subset \{1,\ldots,d\}$,  the angular density
for any $x\in\real^d_+$ is
\begin{equation}\label{eq:spectralden}
h_{d,J} \left( \frac{x_{i_1}}{\sum_{\substack{i\in J}} x_i}, \cdots , \frac{x_{i_{k-1}}}{\sum_{i \in J} x_i} \right)=
-\left(\sum_{\substack{i\in J}}  x_i \right)^{k+1} 
\lim_{\substack{x_j\rightarrow0,\\ j\notin J}}
\frac{\partial^k V}{\partial x_{i_1} \cdots \partial x_{i_k}} (x).
\end{equation}
Thus, when $J=\{j\}$ for any $j \in \{1,\ldots,d\}$ then $\simplex_{d,J}$ is a vertex 
$\be_j$ of the simplex
and the density is a point mass, denoted $h_{d,J}=H(\{\be_j\})$.
In this case
\eqref{eq:spectralden} reduces to
\begin{equation}\label{eq:mass_ext_skt}
h_{d,J} = \Psi_{d-1} 
\left\{ 
\left( 
-\sqrt{ \frac{\nu+1}{1-\omega_{i,j}^2 } } \omega_{i,j},
{i \in I_j}\right)^\top; 
\Omega^\circ_j,  \alpha_j^\circ, \tau_j^\circ, \kappa_j^\circ, \nu+1
\right\}, 
\end{equation}
where $\Psi_{d-1}$ denotes the $d-1$-dimensional skew-$t$ distribution with parameters  again given in the proof to Theorem \ref{teo:main_res} (Appendix \ref{ap:proof_th_1}).

Computations of all $2^{d}-1$ densities that lie on the edges and vertices of the simplex
are available for $d=3$.
 In this case, the angular densities on the interior and vertices of the simplex
can be deduced from \eqref{eq:spec_den_ext_swt} and \eqref{eq:mass_ext_skt}. 
For all $i,j\in J=\{1,2,3\}$, with $i\neq j$, the angular density on the edges of $\simplex_{d,J}$
for $w\in\simplex_{d,J}$ is given by 
\begin{align}\label{eq:biv_spec_den_ext_swt}
h_{3, \{i,j\}}(w)
&= 
\sum_{u,v \in \{i,j\}, u \neq v} 
\left(
\frac{\psi (b_{u,v}^\circ;\nu+1 )}
{\Psi ( \bar{\tau}_u ;\nu+1 )}
\Psi_2 
\left[
\left\{ y_1^\circ(u,v),y_2^\circ(u,v) \right\}^\top; \bar{\Omega}^{\circ\circ}_u, \nu+2
\right]
\right. \nonumber \\
& \qquad \times \frac 1 w_1
\left\{ \frac{\der^2 b_{u,v}^\circ}{\der w_u \der w_v} 
+ \frac{\der b_{u,v}^\circ}{\der w_v} 
\left( \frac{\der b_{u,v}^\circ}{\der w_u} \frac{(\nu+2) b_{u,v}^\circ}{\nu +1 + b_{u,v}^{\circ 2}} - \frac 1 w_1 \right)
\right\} \nonumber \\
& \qquad + \psi\{ y_1^\circ(u,v); \nu +2\} \sqrt{\frac{\nu+2}{1-\Omega^{\circ 2}_{u,[1,2]}} } 
\frac{b_{u,v}^\circ c_{u,\bar{k}} + \Omega^{\circ 2}_{u,[1,2]} (\nu+1)}{(\nu+1+b_{u,v}^{\circ 2})^{3/2}} \nonumber \\
& \qquad \times \Psi \left( 
\frac{\sqrt{\nu+3} \left\{
z_2^\circ(u,v) \Omega^{\circ \circ }_{u,[1,1]} - z_1^\circ(u,v) \Omega^{\circ\circ}_{u,[1,2]}
\right\}}
{\sqrt{ \left[\Omega^{\circ \circ}_{u,[1,1]} \{\nu +1 +b_{u,v}^{\circ 2}\} + z_1^{\circ 2}(u,v)\right]
\det(\Omega^{\circ\circ}_u)}}
; \nu+3\right) \\
& \qquad + \psi\{ y_2^\circ(u,v); \nu +2\} \sqrt{\frac{\nu+2}{1-\Omega^{*2}_{u,[1,3]}} } 
\frac{x(u,v) \bar{\tau}_u + \Omega^{*2}_{u,[1,3]} (\nu+1)}{\{\nu+1+b_{u,v}^{\circ 2}\}^{3/2}} \nonumber \\
& \qquad \left.\times \Psi \left\{ 
\frac{\sqrt{\nu+3} \left\{
z_1^\circ(u,v) \Omega^{\circ\circ}_{u,[2,2]} - z_2^\circ(u,v) \Omega^{\circ\circ}_{u,[1,2]}
\right\}}
{\sqrt{ \left(\Omega^{\circ\circ}_{u,[2,2]} \{\nu +1 + b_{u,v}^{\circ 2}\} + z_2^\circ(u,v)^2\right)
\det(\Omega^{\circ\circ}_u)}}
; \nu+3\right\}
\right)\nonumber,
\end{align}
where  for all $u,v \in J$, with $u \neq v$, and $\bar{k} \notin \{i,j\}$,
$$
y_\ell^\circ (u,v) = \frac{z_\ell^\circ (u,v)}{\sqrt{\Omega_{u,[\ell,\ell]}^\circ}} \sqrt{\frac{\nu+2}{\nu+1+b_{u,v}^{\circ 2}}}, \; \ell=1,2, \quad
z_1^\circ (u,v) = c_{u,\bar{k}} - \Omega^\circ_{u,[1,2]} b_{u,v}^\circ,
$$ 
$$
c_{u,v} = -\omega_{u,v} \sqrt{\frac{\nu+1}{1-\omega^2_{u,v}}},\quad
 z_2^\circ (u,v) = \bar{\tau}_u - \Omega^\circ_{u,[1,3]},\quad  
 b_{u,v}^\circ = \sqrt{\frac{\nu+1}{1-\omega^2_{u,v}}} \left( \left( \frac{w_v^\circ}{w_u^\circ} \right)^{1/\nu} - \omega_{u,v} \right),
$$
$$ 
\quad
\Omega^\circ_u =
\left[
\begin{array}{cc}
\bar{\Omega}_u & -\delta_u \\
-\delta_u^\top & 1
\end{array}
\right],
\quad
\delta_u^\top = \bar{\Omega}_{u} \left( \alpha_v \sqrt{1-\omega_{u,v}^2}, \alpha_k \sqrt{1-\omega_{u,k}^2}\right)^\top,\;
\bar{\Omega}^{\circ\circ}_u =
{\omega_u^{\circ}}^{-1/2} \Omega^{\circ\circ}_u {\omega_u^{\circ}}^{-1/2},
$$
$\omega_u^\circ=\diag(\Omega^{\circ\circ}_u)$, $\Omega^{\circ\circ}_u = \Omega^{\circ}_{u,[-1,-1]} - \Omega^{\circ}_{u,[-1,1]} \Omega^{\circ}_{u,[1,-1]}$.
Components of $\Omega^\circ_u$ and $\Omega^{\circ\circ}_u$ are respectively given by $\Omega^\circ_{u,[i,j]}$ and
$\Omega^{\circ\circ}_{u,[i,j]}$ for $i,j \in J$. See also Appendix \ref{ap:proof_th_1} for further details.
When, $\tau=0$  and $\alpha(s)=0$, then 
the densities \eqref{eq:spec_den_ext_swt}, \eqref{eq:mass_ext_skt} and
\eqref{eq:biv_spec_den_ext_swt} reduce to the densities of the extremal-$t$ dependence model.
A graphical illustration that shows the difference between the two dependence models is provided
in the Supporting Information.

Therefore, for $d=3$ the estimation of dependence parameters can be based on the following approach.
Let $\{(r_i,w_i): i=1,\ldots,n \}$ be the set of observations,
where $r_i=\sum_{j\in I }x_{i,j}$ and $w_i=x_i/r_i$, with $x_i=(x_{i,j})_{j\in I}$, 
are pseudo-polar radial and angular components. Then the approximate log-likelihood is defined by
\begin{equation}\label{eq:app_llik}
\ell(\vartheta;\tilde{w})=\sum_{\substack{i=1,\ldots,n:\\ r_i>r_0}}\log h(w_i;\vartheta),
\end{equation}
where $\tilde{w}=(w_1,\ldots,w_n)^\top$, 
for some radial threshold $r_0>0$,
and where $h$ is the angular density function of the extremal-skew-$t$ dependence model.
The components of the sum in \eqref{eq:app_llik} comprise the three types of angular densities lying on the interior, edges and vertices of the simplex.
Whether an angular component belongs either to the interior, an edge or a vertex of the simplex, producing the associated density, is determined according the following criterion. 
We select a threshold $c \in [0,0.1]$ and we construct the following partitions for an arbitrary observation
$w_i=(w_{i,j},w_{i,k},w_{i,l})$, $i=1,\ldots,n$. Set $w\equiv w_i$ for simplicity. 
When $\mathcal{C}_j=\{w_{j}>1-c; j\in I\}$  
then an observation belongs to vertex $\boldsymbol{e}_j$.
When $\mathcal{E}_{j,k}=\{w_{j},w_{k}<1-c, w_{l}<c, w_{j}>1-2w_{k}, w_{k}>1-2w_{j};  j\in I, k\in I_j, l \in I_j\backslash\{k\}\}$, 
then an observation belongs to edge between the $j$th and $k$th components.
When $\mathcal{I}=\{w_j >c; j\in I\}$ then an observation belongs to the interior 
(see the Supporting Information for more details). 
The components of the angular density $h(w)$ then require rescaling so that they satisfy the constraints of valid angular densities -- namely that they integrate to the number of components of $w$ (3 in the trivariate case) -- while also respecting the partition of  $\mathbb{W}$ implied by $c$. Without this rescaling then the likelihood of e.g. the model that places mass on all subsets of the simplex is not comparable with that of models that places mass only on subsets of the simplex.
Specifically
$$
\int_{\simplex}h(w)\der w=K_{\mathcal{C}}\sum_{j\in I}\int_{\mathcal{C}_j}h_{3,\{j\}}\der w+
\sum_{\substack{j=1,2\;\\ k=j+1,3}}K_{\mathcal{E}_{j,k}}\int_{\mathcal{E}_{j,k}}h_{3,\{j,k\}}(w)\der w+
K_{\mathcal{I}}\int_{\mathcal{I}}h_{3,\{1,2,3\}}(w)\der w=3,
$$
where
$$
K_{\mathcal{C}} = \frac{4}{\sqrt{3}c^2},\quad
K_{\mathcal{E}_{j,k}}=\frac{2\int_0^1 h_{3,\{j,k\}}(w) \der w} {c \sqrt{3}(1-2c)},\quad
K_{\mathcal{I}} = \frac{\int_0^1 \int_0^1 h_{3,\{1,2,3\}}(w) \der w}{\int_c^{1-2c} \int_c^{1-2c} h_{3,\{1,2,3\}}(w) \der w},
$$
and $h_{3,\{j\}}$, $h_{3,\{j,k\}}(w)$ and $h_{3,\{1,2,3\}}(w)$ are defined above.
Note that for $j,k\in I$ with $j\neq k$, we have that $h_{3,\{j,k\}}(w)=h_{3,\{k,j\}}(w)$.
In the bivariate case ($d=2$), the appropriate modification only considers the mass on the vertices and interior.

We now illustrate the ability of the approximate likelihood in estimating the extremal dependence parameters in the bivariate and trivariate cases.
We generate $500$ replicate datasets of sizes $5000$ (bivariate) and $1000$ (trivariate), with parameters  
$\vartheta_2 = (\omega, \nu) = (0.6, 1.5)$ and $\vartheta_3 =( \omega_{1,2}, \omega_{1,3}, \omega_{2,3}, \nu ) = (0.6, 0.8, 0.7, 1)$.
Each dataset is transformed to pseudo-polar coordinates and the 100 observations with the largest radial component are retained.
Parameters are estimated through the profile likelihood where the dependence parameter $\omega$ is the parameter of interest and the degree of freedom $\nu$ is considered as a nuisance parameter.
Parameters are estimated for different values of the threshold $c=0, 0.02, 0.04, 0.06, 0.08, 0.1$.
In order to compare likelihoods for different values of $c$, the likelihood functions are evaluated using those data points considered to belong to the interior of the simplex, multiplied by the mass at the corners and/or edges in proportion to their rescaling constants.

\begin{figure}[tbh]
\centering
\includegraphics[width=0.6\textwidth]{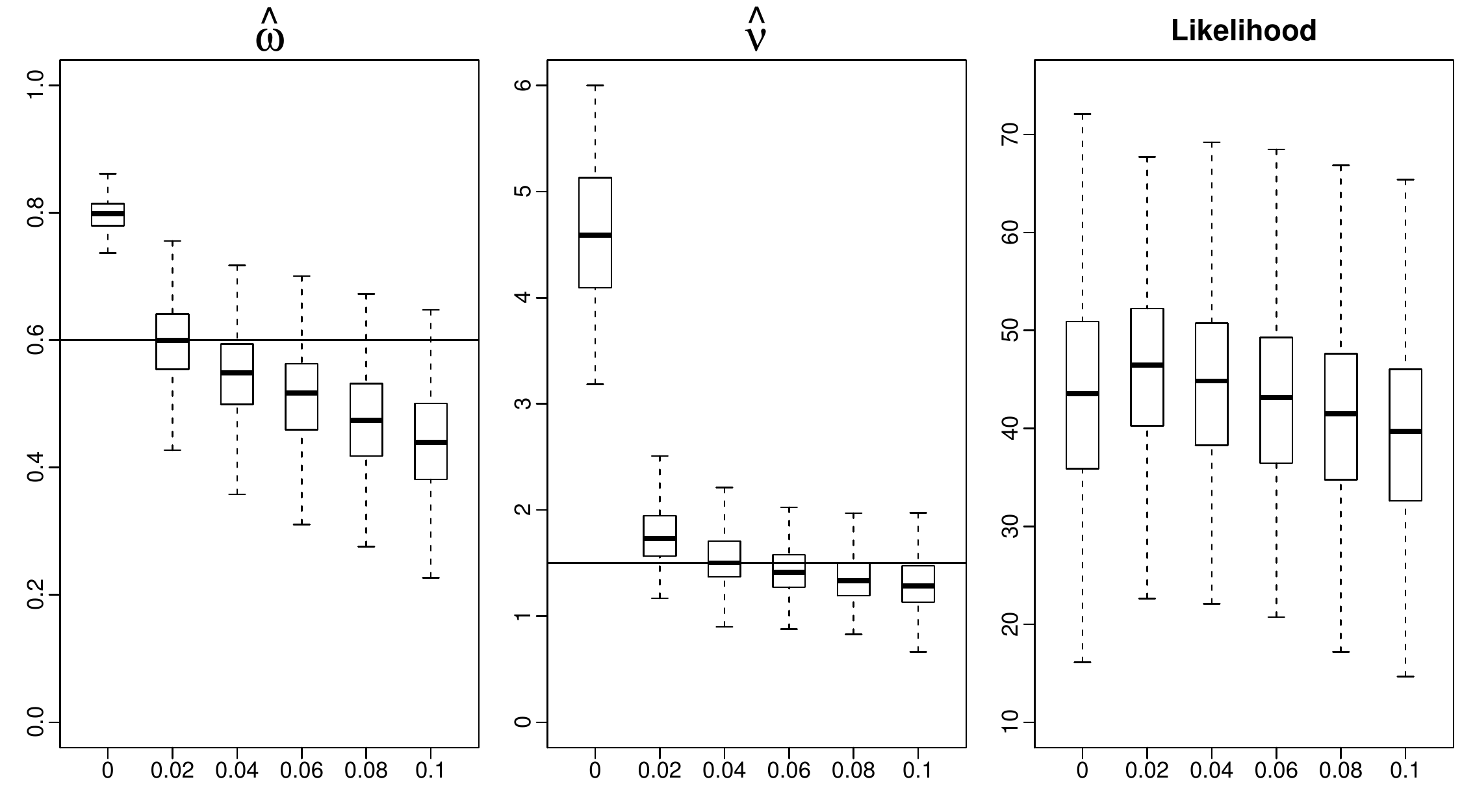}
\caption{\label{fig:profileBiv}\small Left to right: Boxplots of the estimates of the dependence parameter $\omega$, the degree of freedom $\nu$ and the associated maximum of the likelihood function based on the rescaled angular density, when $c=0, 0.02, 0.04, 0.06, 0.08$ and $0.1$. Boxplots are constructed from $500$ replicate datasets of size 5000. Horizontal lines indicate the true values $\omega=0.6$ and $\nu=1.5$.}
\end{figure}

Figures \ref{fig:profileBiv} and \ref{fig:profileTriv} provide (left to right) boxplots of the resulting estimates of the dependence parameter(s) $\omega$, the degree of freedom $\nu$ and of the likelihood function for increasing values of $c$, for the 500 replicate datasets for both bivariate and trivariate cases. The true parameter values are indicated by the horizontal lines.

\begin{figure}[tbh]
\centering
\includegraphics[width=\textwidth]{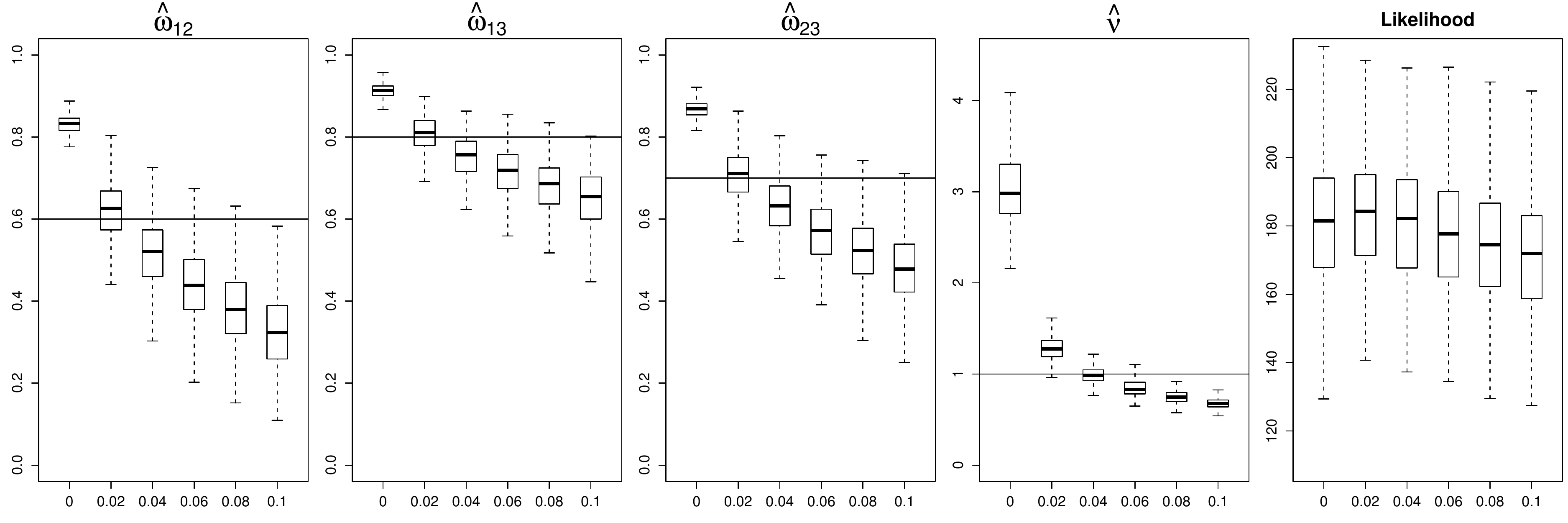}
\caption{\label{fig:profileTriv}\small Left to right: Boxplots of the estimates of the dependence parameter $\omega= (\omega_{1,2}, \omega_{1,3}, \omega_{2,3})$, the degree of freedom $\nu$ and the associated maximum of the likelihood function based on the rescaled angular density, when $c=0, 0.02, 0.04, 0.06, 0.08$ and $0.1$. Boxplots are constructed from 500 replicate datasets of size 1000. Horizontal lines indicate the true values $\omega_{1,2}=0.6, \omega_{1,3}=0.7, \omega_{2,3}=0.7$ and $\nu=1$.}
\end{figure}

In the rightmost  panel of each Figure, the largest values of the log-likelihood are globally obtained for $c=0.02$,  for which the most accurate estimates of $\omega$ and $\nu$ are also obtained. 
Conditional on $c=0.02$ the mean estimates are $\hat{\omega}=0.55$ and $\hat{\nu}=1.79$ in the bivariate case and $\hat{\omega}=(0.62,0.80,0.71)$ and $\hat{\nu}=1.27$ in the trivariate case.
Note that the degree of freedom $\nu$ appears to be slightly overestimated, and appears to be better estimated for slightly larger values of $c$.
Overall this procedure appears capable of efficiently estimating the model parameters.
Note that increased precision of estimates can be obtained by considering a denser range of threshold values $c$.

An independent study comparing the efficiency of the maximum pairwise and triplewise composite likelihood estimators is provided in the Supporting Information.

%
%
%
\section{Application to wind speed data}\label{sec:application}

We illustrate the use of the extremal skew-$t$ process using wind speed data (the weekly maximum 
wind speed in km/h), collected from 4 monitoring stations across Oklahoma, USA, over the March-May period during 1996--2012, as part of a larger dataset of 99 stations.
An analysis establishing the significant marginal, station-specific skewness of these data is presented in the Supporting Information. Here, we
focus on the dependence structure between stations, where 
for simplicity the data is marginally transformed to unit Fr\'{e}chet distributions.
Only  extremal-$t$ and extremal skew-$t$ models are considered, and parameter estimation is performed via pairwise composite likelihoods as detailed at the beginning of Section \ref{sec:inference}.

Model comparison is performed through the composite likelihood information criterion (CLIC; \citealp{varin2011}) given by 
$$
\textrm{CLIC} = - 2 \left[ \ell_2(\hat{\vartheta};x) - \mathrm{tr}\{ \hat{J}(\hat{\vartheta}) \hat{H}(\hat{\vartheta})^{-1}  \} \right],
$$
where $\hat{\vartheta}$ is the maximum composite likelihood estimate of $\vartheta$, $\ell_2 (\hat{\vartheta};x)$ is the maximised pairwise composite likelihood, and $\hat{J}$ and $\hat{H}$ are estimates of 
$
J(\vartheta) = \mathrm{Var}_{U} ( \nabla \ell_2(\vartheta;U))
$ and 
$
H(\vartheta) = \expect_{U} (-\nabla^2 \ell_2(\vartheta;U))
$, 
the variability and sensibility (hessian) matrices, where $U$ is a bivariate random vector with extremal skew-$t$ distribution.

\begin{table}[h!]
\begin{center}
\begin{tabular}{clcccc}
Stations & Model & $\hat{\omega}$ & $\hat{\alpha}$& $\hat{\nu}$ & CLIC \\
 \hline
(CLOU,CLAY,SALL) & ex-$t$ & $(0.67, 0.57, 0.69)$ & $-$ & $2.89$ & $5395.73$ \\
 & ex-skew-$t$ & $(0.42, 0.74, 0.52)$ & $(-0.80, 2.88, -0.23)$ & $2.06$ & $5385.07$ \\
 &&&se: $(0.04,0.14,0.03)$\\
 \hline
(CLOU,CLAY,PAUL) & ex-$t$ & $(0.59, 0.50, 0.69)$ & $-$ & $2.53$ & $5503.54$ \\
 & ex-skew-$t$ & $(0.45, 0.29,  0.65)$ & $(-0.68, 21.07, 23.41)$ & $2.16$ & $5496.90$ \\
  &&&se: $(0.05,0.97,1.09)$\\
   \hline
(CLAY,SALL,PAUL) & ex-$t$ & $(0.65, 0.61, 0.53)$ & $-$ & $1.55$ & $5086.13$ \\
 & ex-skew-$t$ & $(0.56, 0.51, 0.39)$ & $(3.55, 2.36, 8.49)$ & $1.29$ & $5075.87$ \\
  &&&se: $(0.17,0.15,0.63)$\\
   \hline
(CLOU,SALL,PAUL) & ex-$t$ & $(0.37, 0.40, 0.42)$ & $-$ & $1.88$ & $5428.04$ \\
 & ex-skew-$t$ & $(0.29, 0.30, 0.37)$ & $(-0.14, 1.04, 34.70)$ & $2.11$ & $5419.27$ \\
  &&&se: $(0.03,0.02,3.49)$\\
 \hline
\end{tabular}
\end{center}
\caption{\small Pairwise composite likelihood estimates $\hat{\vartheta} = (\hat{\omega},\hat{\nu})$ and $\hat{\vartheta} = (\hat{\omega},\hat{\alpha},\hat{\nu})$  of the extremal-$t$ (ext-$t$) and extremal skew-$t$ (ex-skew-$t$) models respectively, for all possible triplets of the four locations CLOU, CLAY, PAUL and SALL. Standard errors (se) are shown for $\hat{\alpha}$ only.}
\label{tabTriv}
\end{table} 

Table \ref{tabTriv} presents the pairwise composite likelihood estimates of $\omega = (\omega_{12},\omega_{13},\omega_{23})$, $\alpha = (\alpha_1, \alpha_2, \alpha_3)$ and $\nu$ for the extremal-$t$ and extremal skew-$t$ models, obtained for all triplewise combinations of the four locations CLOU, CLAY, PAUL and SALL.
For each triple the extremal skew-$t$ model achieves a lower CLIC score than  the extremal-$t$ model, indicating its greater suitability.
Moreover the standard errors of the estimated slant parameters $\hat{\alpha}$, 
clearly indicate that these parameters are non-zero, strengthening the argument of a significantly better fit from the extremal skew-$t$ model

For each location triple $(X, Y, Z)$ we can also evaluate the conditional probability of exceeding some fixed threshold $(x,y,z)$ using each parametric model.
Table \ref{tabExceed} presents estimated probabilities of  the two cases
$
\prob(X > x| Y > y, Z > z)
$
and 
$
\prob(X > x, Y > y | Z > z),
$
along with the associated empirical probabilities and their $95\%$ confidence intervals (CI) for a range of thresholds.
For these specific thresholds, the extremal skew-$t$ model provides estimates of the conditional probabilities that fall within the $95\%$ empirical CI.
However, four probabilities estimated with the extremal-$t$ model are not consistent with the empirical CI.
This indicates that the additional flexibility of the extremal skew-$t$ model allows it to more accurately characterise the dependence structure evident in the observed data.

\begin{table}[h!]
\begin{center}
\begin{tabular}{@{\extracolsep{4pt}}ccccc@{}}
 & Threshold & Extremal-$t$ & Extremal skew-$t$ & Empirical ($95\%$ CI) \\
\hline
$X|Y,Z$ & $(q^{90}_{\textrm{CO}}, q^{70}_{\textrm{CA}}, q^{70}_{\textrm{PA}})$ & $0.2587$ & $0.2737$ & $0.3333 \,(0.2706, 0.3960)$  \\
 & $(q^{90}_{\textrm{SA}}, q^{70}_{\textrm{CA}}, q^{70}_{\textrm{PA}})$ & $0.3268$ & $0.3305$ & $0.2973 \,(0.2356, 0.3590)$  \\
 & $(q^{90}_{\textrm{PA}}, q^{70}_{\textrm{CA}}, q^{70}_{\textrm{SA}})$ & $0.3752$ & $0.3356$ & $0.2857 \,(0.2247, 0.3467)$  \\
 & $(q^{90}_{\textrm{CO}}, q^{70}_{\textrm{SA}}, q^{70}_{\textrm{PA}})$ & $0.2686$ & $0.3150$ & $0.3333 \,(0.2706, 0.3960)$  \\
\hline
$X,Y|Z$ & $(q^{90}_{\textrm{CO}}, q^{90}_{\textrm{CA}}, q^{70}_{\textrm{SA}})$ & $0.1196$ & $0.0789$ & $0.0781 \,(0.0420, 0.1142)$  \\
 & $(q^{90}_{\textrm{CA}}, q^{90}_{\textrm{PA}}, q^{70}_{\textrm{CO}})$ & $0.1236$ & $0.0776$ & $0.0938 \,(0.0546, 0.1330)$  \\
 & $(q^{90}_{\textrm{CO}}, q^{90}_{\textrm{SA}}, q^{70}_{\textrm{PA}})$ & $0.0896$ & $0.1048$ & $0.0938 \,(0.0550, 0.1326)$  \\
 & $(q^{90}_{\textrm{SA}}, q^{90}_{\textrm{PA}}, q^{70}_{\textrm{CO}})$ & $0.1038$ & $0.1071$ & $0.0769 \,(0.0415, 0.1123)$  \\
 \hline
\end{tabular}
\end{center}
\caption{\small Extremal-$t$ and extremal skew-$t$ conditional probabilities of exceeding particular fixed thresholds of the form  $\prob(X > x| Y > y, Z > z)$ and $\prob(X > x, Y > y | Z > z)$, along with empirical estimates.
The windspeed thresholds $(x,y,z)$ are constructed from the marginal quantiles $q^{70}= ( q_{\textrm{CO}}^{70}, q_{\textrm{CA}}^{70}, q_{\textrm{SA}}^{70}, q_{\textrm{PA}}^{70} ) = ( 18.04, 20.33, 24.18, 23.61 )$
and
$q^{90}= ( q_{\textrm{CO}}^{90}, q_{\textrm{CA}}^{90}, q_{\textrm{SA}}^{90}, q_{\textrm{PA}}^{90} ) = ( 22.11, 24.33, 29.05, 28.26 )$ at each location. }
\label{tabExceed}
\end{table}

\begin{figure}[tbh]
\centering $
\begin{array}{cccc}
\includegraphics[width=0.25\textwidth]{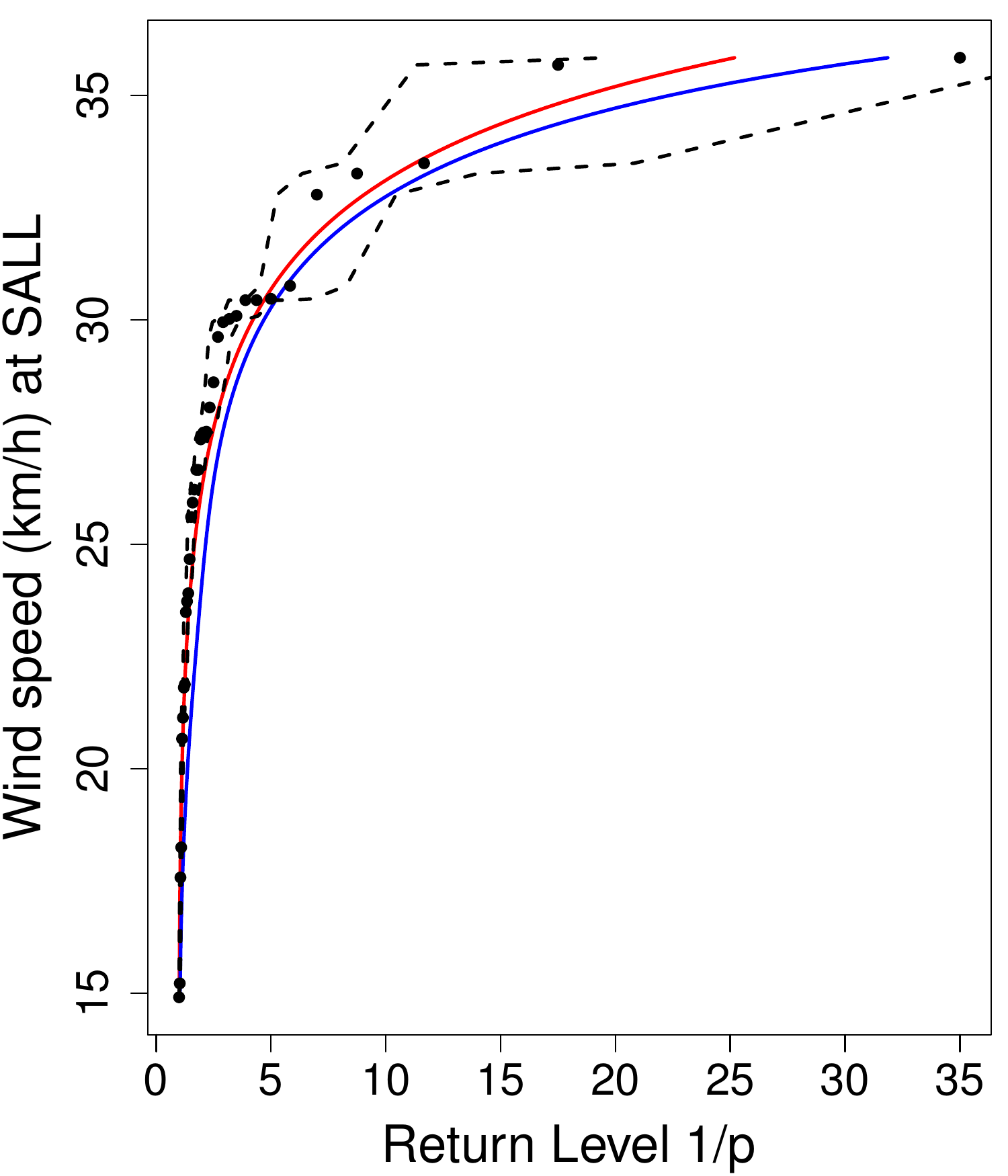} &
\includegraphics[width=0.25\textwidth]{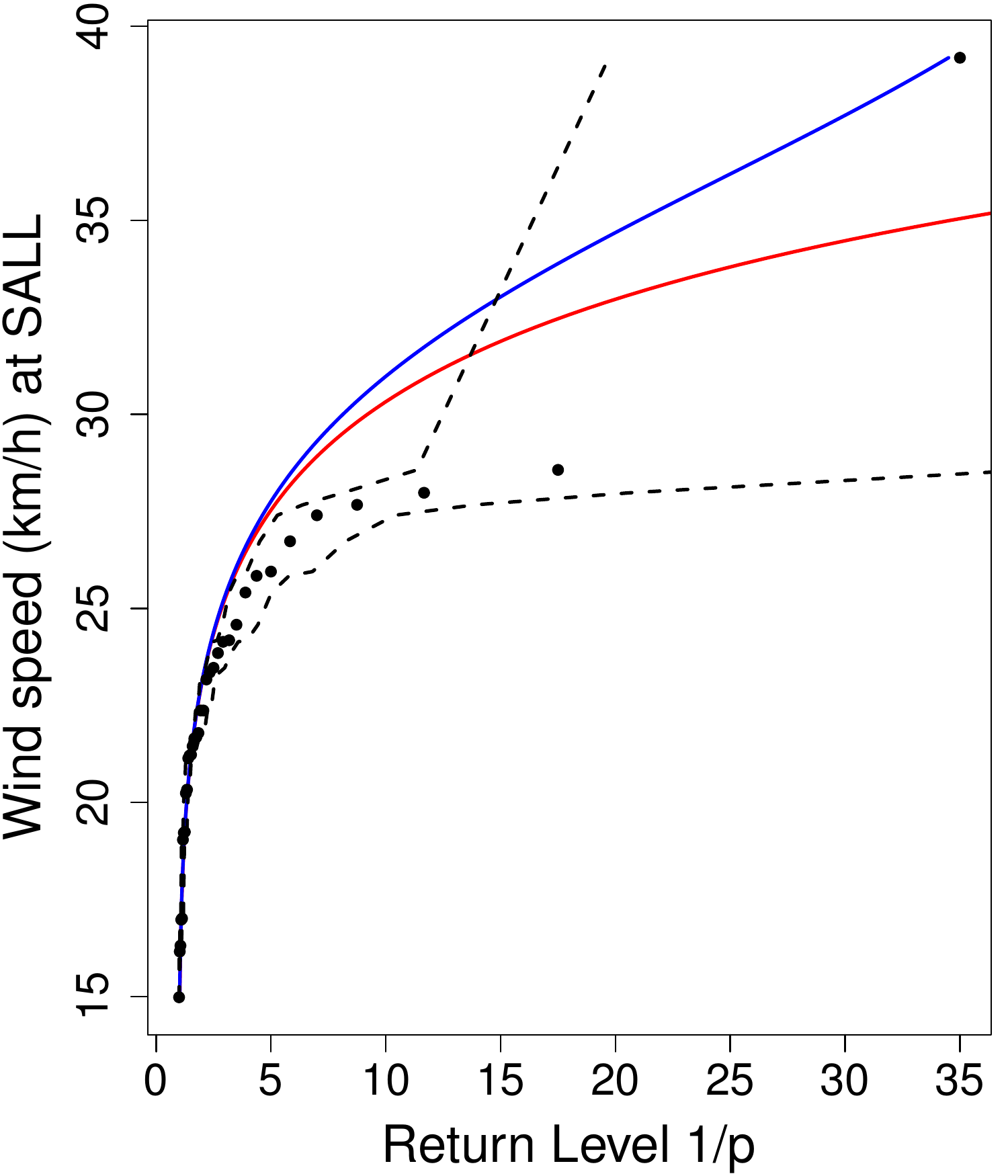} &
\includegraphics[width=0.25\textwidth]{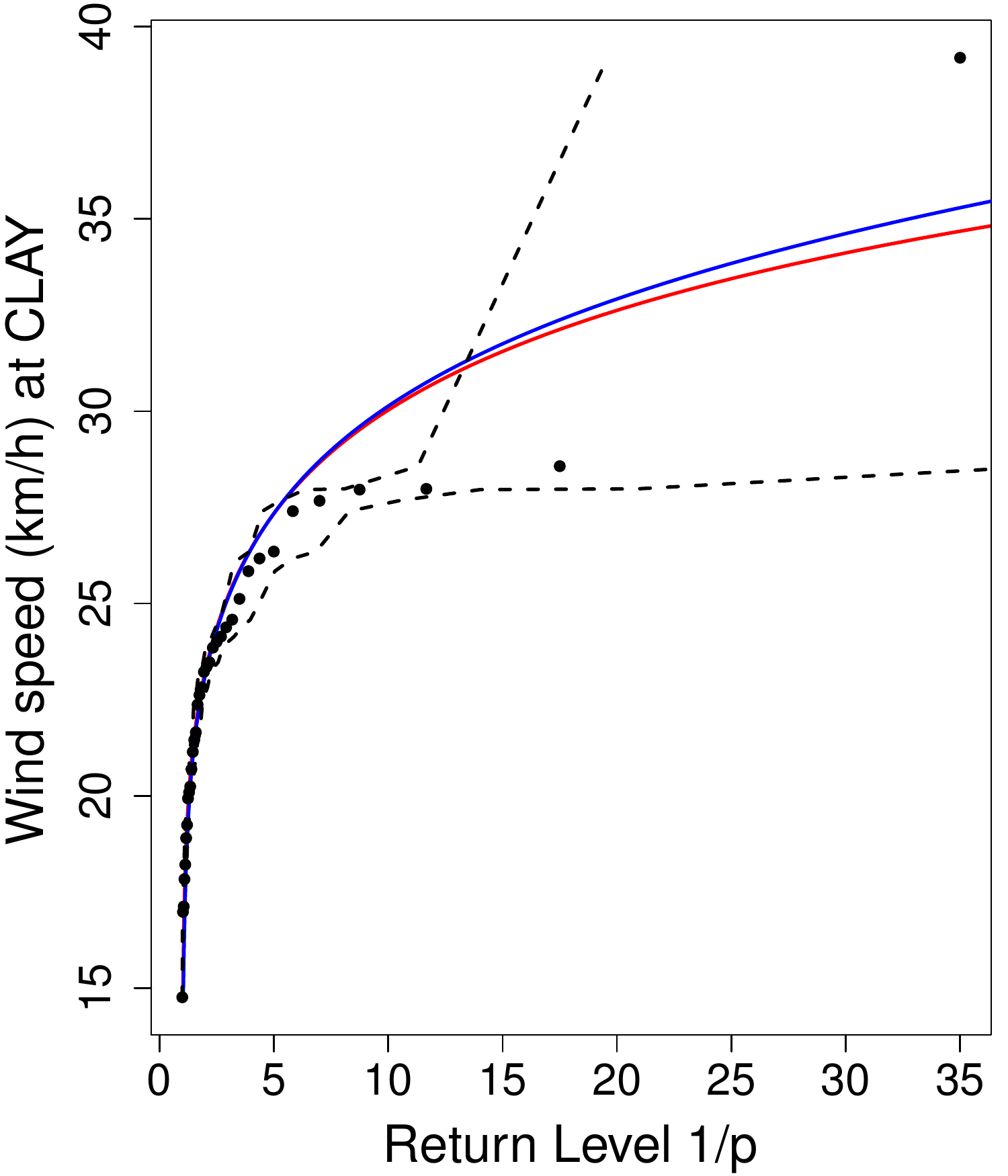} &
\includegraphics[width=0.25\textwidth]{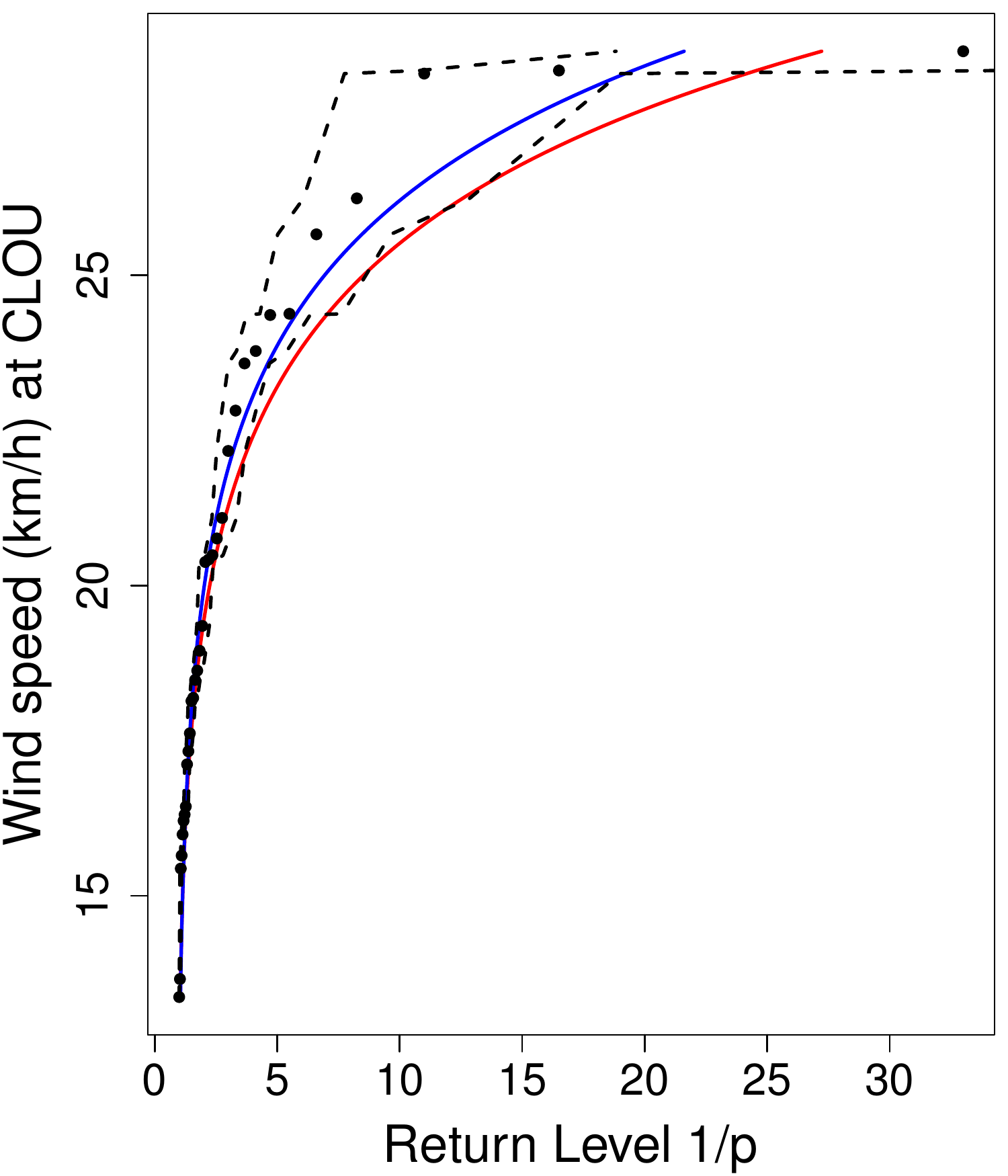} \\ 
\includegraphics[width=0.25\textwidth]{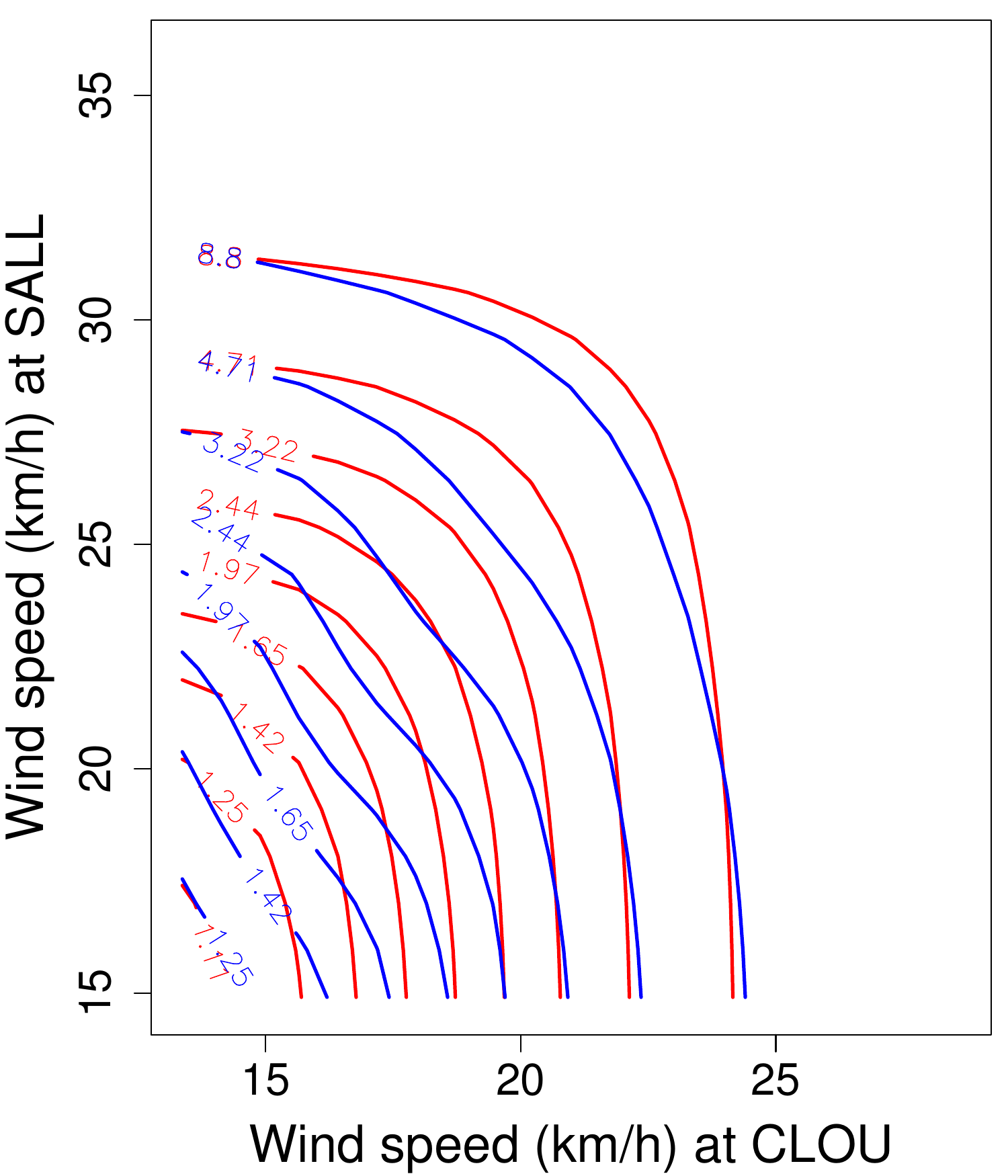} &
\includegraphics[width=0.25\textwidth]{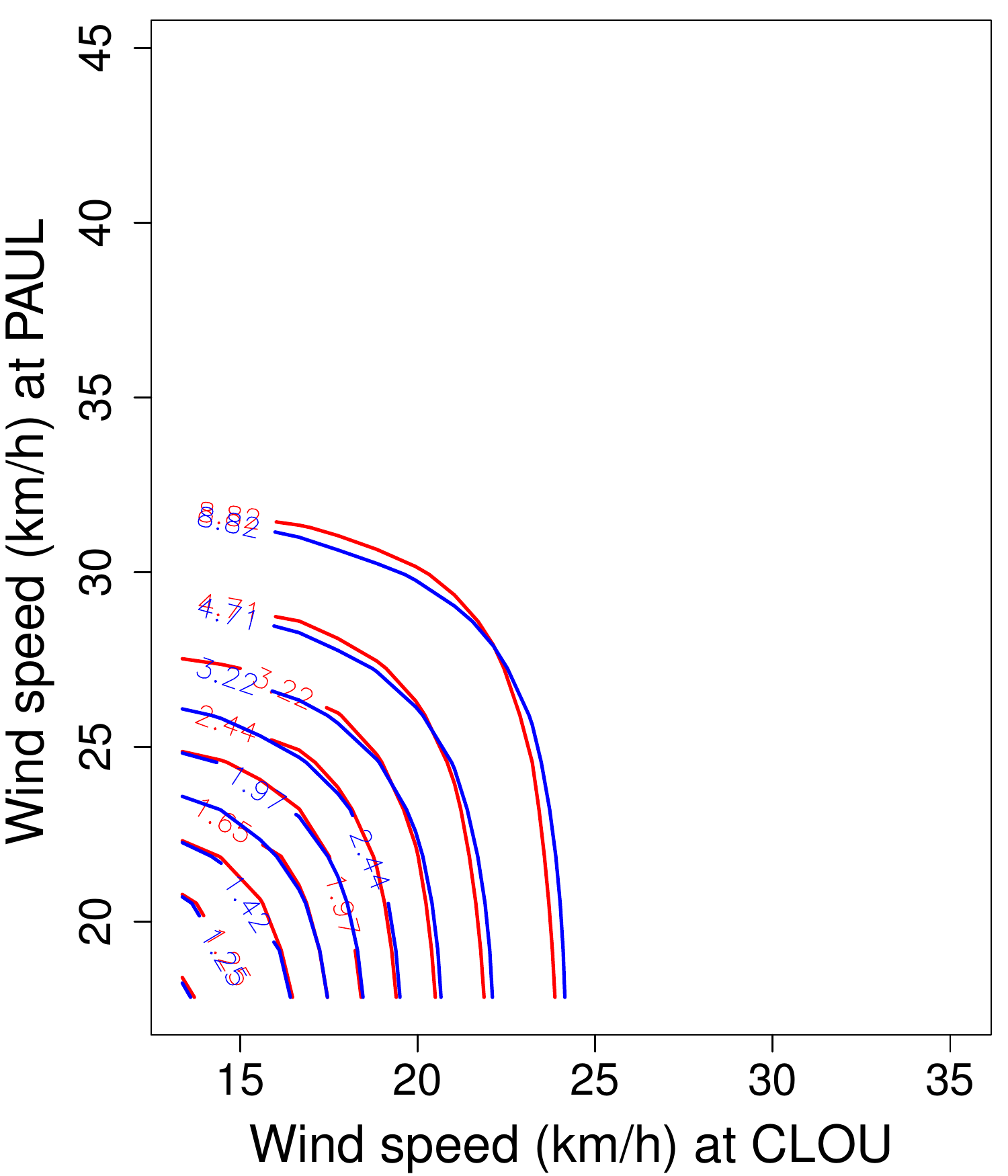} &
\includegraphics[width=0.25\textwidth]{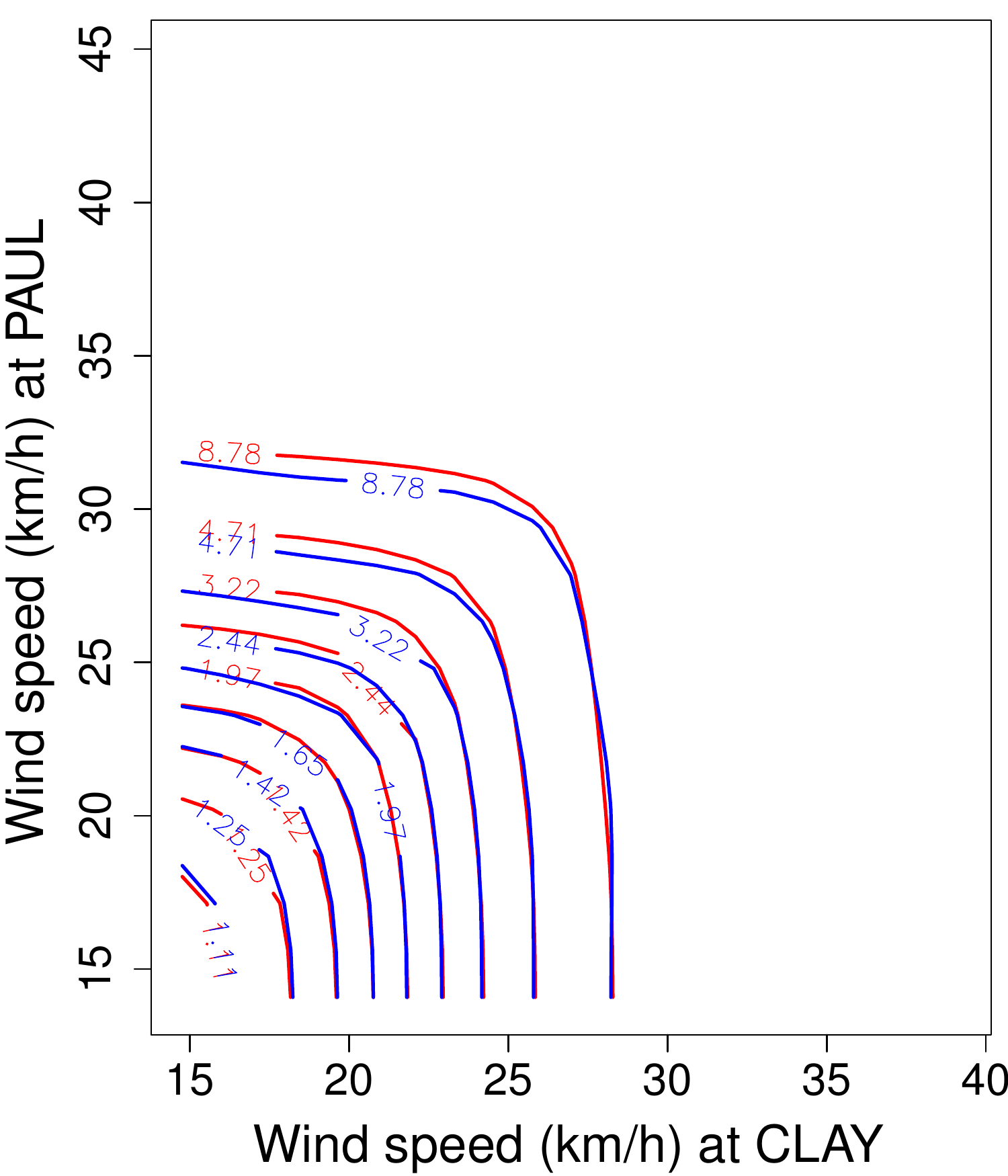} &
\includegraphics[width=0.25\textwidth]{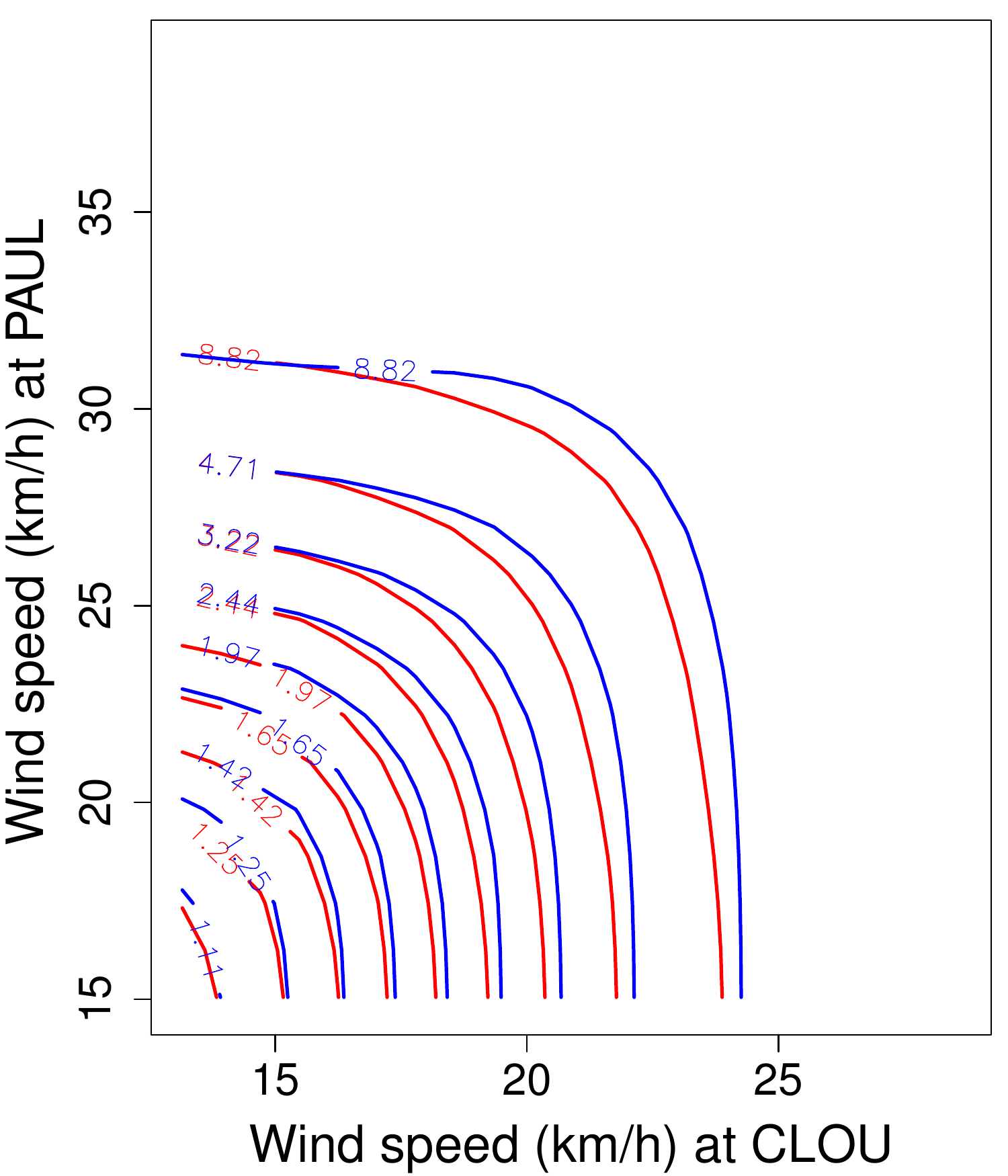} 
\end{array}$
\caption{\label{fig:figRL}\small Univariate (top row) and bivariate (bottom) conditional return levels for the triples (left-to-right):  (CLOU, CLAY, SALL), (CLOU, CLAY, PAUL), (CLAY, SALL, PAUL) and (CLOU, SALL, PAUL). Red and blue lines respectively indicate return levels  calculated from extremal-$t$ and extremal skew-$t$ models. Points indicate the empirical observations and the black dashed lines their $95\%$ confidence interval. 
}
\end{figure}
    
Finally, 
Figure \ref{fig:figRL} provides examples of univariate (top panels) and bivariate (bottom) conditional return levels for each triple of sites.
The return levels are computed 
conditionally on the wind at the remaining station(s) being higher than their upper $70\%$ marginal quantile.
For the univariate conditional return levels (top panels),
both the extremal-$t$ and extremal skew-$t$ model fits are strongly influenced by the windspeed outlier of $\sim 40$ km/h observed at CLAY station (centre two panels). This phenomenon, whereby the far tails of extremal model fits can be dominated by a single extreme outlier, is not uncommon in practice \cite[e.g.][]{coles+ps03}. Being the more flexible model, the extremal skew-$t$ model is better able to follow this extreme outlier compared to the extremal $t$. When the outlier is not present (in the two outer panels), the extremal skew-$t$ model provides a better visual fit to the observed data and spends more time within the empirical confidence intervals, indicating a superior model fit.

The primary differences in the bivariate conditional return levels (bottom panels, Figure \ref{fig:figRL}) are the possibility of asymmetric contour levels under the extremal skew-$t$ model (blue line) in contrast with symmetric contours under the extremal-$t$ model (red line).
The difference is more noticeable in the leftmost and rightmost panel. 
The leftmost panel indicates lower return levels for the extremal skew-$t$ model, which occurs because (CLOU, SALL) have negative slant parameters (Table \ref{tabTriv}, top row) and so the joint tail is shorter than that of the extremal $t$.
 Conversely, the rightmost panel exhibits larger return levels for the extremal skew-$t$ model, as a result of the small negative and very large slant parameters for (CLOU, PAUL) (Table \ref{tabTriv}, bottom row), and so the joint tail is longer than that of the extremal-$t$. 
 The differences in the centre two panels are less pronounced.
For the second panel, the slant parameters of (CLOU, PAUL) similarly take a large positive and a small  negative value (Table \ref{tabTriv}, row 2). However, as the parameter for CLAY is also a large positive value this means that there is little difference between the joint tails of the two models. Finally, for the third panel, the slant parameters of (CLAY, PAUL, SALL) are relatively small and positive (Table \ref{tabTriv}, row 3) and so there is little difference between the joint tails of the two models. 

In summary, for these wind speed data, the more flexible extremal skew-$t$ model is demonstrably superior to the extremal-$t$ model in describing the extremes of both the univariate marginal distributions, and the extremal dependence between locations.

%
%
%
\section{Discussion}

Appropriate modelling of extremal dependence is critical for producing realistic and precise estimates of future extreme events. In practice this is a hugely challenging task, as extremes in different application areas may exhibit different types of dependence structures,  asymptotic dependence levels, exchangeability, and stationary or non-stationary behaviour.

Working with families of skew-normal distributions and processes we have derived flexible  new classes of extremal dependence models.  Their flexibility arises as they include a wide range of dependence structures, while also incorporating several previously developed and popular models, such as the stationary extremal-$t$ process and its sub-processes, as special cases.
These include dependence structures that are asymptotically independent, which is useful for describing the dependence of variables that are not exchangeable, and  a wide class of non-stationary, asymptotically dependent models, suitable for the  modelling of spatial extremes.

In terms of future development,  semi-parametric estimation methods would provide powerful techniques to fully take advantage of the flexibility offered by non-stationary max-stable models. Such methods can be computationally demanding, however.
An interesting further direction would be to design simple and interpretable families of covariance functions for skew-normal processes for which it is then possible to derive max-stable dependence models that are useful in practical applications.

The code used to perform the simulations studies and real data analysis in Section \ref{sec:inference} and \ref{sec:application} as well as in the Supporting Information, is available in the {\tt R} \citep{rteam10} package {\tt ExtremalDep} \citep{ExtremalDep} available
at {\tt https://r-forge.r-project.org/R/?group\_id=1998}.

\section*{Supporting Information}
Additional information for this article is available online.\\
Description: additional derivations, simulations and figures.

\section*{Acknowledgements}
We would like to thank the referees, associate editor and editor for useful comments which led to improved
presentation of the material.

\bibliographystyle{chicago}
\bibliography{bibliography}

\begin{thebibliography}{}

\bibitem[\protect\citeauthoryear{Arellano-Valle and Azzalini}{Arellano-Valle
  and Azzalini}{2006}]{arellano2006unification}
Arellano-Valle, R.~B. and A.~Azzalini (2006).
\newblock On the unification of families of skew-normal distributions.
\newblock {\em Scand. J. Statist.\/}, 561--574.

\bibitem[\protect\citeauthoryear{Arellano-Valle and Genton}{Arellano-Valle and
  Genton}{2010}]{arellano2010}
Arellano-Valle, R.~B. and M.~G. Genton (2010).
\newblock Multivariate extended skew-{$t$} distributions and related families.
\newblock {\em Metron\/}~{\em 68\/}(3), 201--234.

\bibitem[\protect\citeauthoryear{Azzalini}{Azzalini}{1985}]{azzalini1985}
Azzalini, A. (1985).
\newblock A class of distributions which includes the normal ones.
\newblock {\em Scand. J. Statist.\/}, 171--178.

\bibitem[\protect\citeauthoryear{Azzalini}{Azzalini}{2005}]{azzalini2005}
Azzalini, A. (2005).
\newblock The skew-normal distribution and related multivariate families.
\newblock {\em Scand. J. Statist.\/}~{\em 32\/}(2), 159--200.
\newblock With discussion by Marc G. Genton and a rejoinder by the author.

\bibitem[\protect\citeauthoryear{Azzalini}{Azzalini}{2013}]{azzalini2013skew}
Azzalini, A. (2013).
\newblock {\em The skew-normal and related families}, Volume~3.
\newblock Cambridge University Press.

\bibitem[\protect\citeauthoryear{Beranger, Marcon, and Padoan}{Beranger
  et~al.}{2015}]{ExtremalDep}
Beranger, B., G.~Marcon, and S.~Padoan (2015).
\newblock {\em ExtremalDep: Extremal Dependence Modeling}.
\newblock R package version 0.1-2/r76.

\bibitem[\protect\citeauthoryear{Bortot}{Bortot}{2010}]{bortot2010}
Bortot, P. (2010).
\newblock Tail dependence in bivariate skew-normal and skew-$t$ distributions.
\newblock Unpublished manuscript.

\bibitem[\protect\citeauthoryear{Brown and Resnick}{Brown and
  Resnick}{1977}]{brown1977}
Brown, B.~M. and S.~I. Resnick (1977).
\newblock Extreme values of independent stochastic processes.
\newblock {\em J. Appl. Probab.\/}, 732--739.

\bibitem[\protect\citeauthoryear{Chan and Wood}{Chan and Wood}{1997}]{chan1997}
Chan, G. and A.~T. Wood (1997).
\newblock Algorithm {AS} 312: An algorithm for simulating stationary {Gaussian}
  random fields.
\newblock {\em J. R. Stat. Soc. Ser. C. Appl. Stat.\/}~{\em 46\/}(1), 171--181.

\bibitem[\protect\citeauthoryear{Chang and Genton}{Chang and
  Genton}{2007}]{chang2007}
Chang, S.-M. and M.~G. Genton (2007).
\newblock Extreme value distributions for the skew-symmetric family of
  distributions.
\newblock {\em Comm. Statist. Theory Methods\/}~{\em 36\/}(9), 1705--1717.

\bibitem[\protect\citeauthoryear{Coles, Pericchi, and Sisson}{Coles
  et~al.}{2003}]{coles+ps03}
Coles, S.~G., L.~R. Pericchi, and S.~A. Sisson (2003).
\newblock A fully probabilistic approach to extreme value modelling.
\newblock {\em Journal of Hydrology\/}~{\em 273}, 35--50.

\bibitem[\protect\citeauthoryear{Coles and Tawn}{Coles and
  Tawn}{1991}]{coles1991}
Coles, S.~G. and J.~A. Tawn (1991).
\newblock Modelling extreme multivariate events.
\newblock {\em J. R. Stat. Soc. Ser. B. Stat. Methodol.\/}~{\em 53\/}(2), pp.
  377--392.

\bibitem[\protect\citeauthoryear{Coles and Tawn}{Coles and
  Tawn}{1994}]{coles1994}
Coles, S.~G. and J.~A. Tawn (1994).
\newblock Statistical methods for multivariate extremes: An application to
  structural design.
\newblock {\em J. R. Stat. Soc. Ser. C. Appl. Stat.\/}~{\em 43\/}(1), pp.
  1--48.

\bibitem[\protect\citeauthoryear{Davison and Gholamrezaee}{Davison and
  Gholamrezaee}{2012}]{davison2012c}
Davison, A.~C. and M.~M. Gholamrezaee (2012).
\newblock Geostatistics of extremes.
\newblock {\em Proceedings of the {R}oyal {S}ociety of {L}ondon {S}eries {A}:
  {M}athematical and {P}hysical {S}ciences\/}~{\em 468}, 581--608.

\bibitem[\protect\citeauthoryear{Davison, Padoan, and Ribatet}{Davison
  et~al.}{2012}]{davison2012b}
Davison, A.~C., S.~A. Padoan, and M.~Ribatet (2012).
\newblock Statistical modeling of spatial extremes.
\newblock {\em Statist. Sci.\/}~{\em 27}, 161--186.

\bibitem[\protect\citeauthoryear{de~Haan}{de~Haan}{1984}]{dehaan1984}
de~Haan, L. (1984).
\newblock A spectral representation for max-stable processes.
\newblock {\em Ann. Appl. Probab.\/}~{\em 12\/}(4), 1194--1204.

\bibitem[\protect\citeauthoryear{de~Haan and Ferreira}{de~Haan and
  Ferreira}{2006}]{dehaan2006}
de~Haan, L. and A.~Ferreira (2006).
\newblock {\em Extreme value theory}.
\newblock Springer Series in Operations Research and Financial Engineering.
  Springer, New York.
\newblock An introduction.

\bibitem[\protect\citeauthoryear{Dutt}{Dutt}{1973}]{dutt1973}
Dutt, J.~E. (1973).
\newblock A representation of multivariate normal probability integrals by
  integral transforms.
\newblock {\em Biometrika\/}~{\em 60\/}(3), 637--645.

\bibitem[\protect\citeauthoryear{Feller}{Feller}{1968}]{feller1968}
Feller, W. (1968).
\newblock {\em An Introduction to Probability Theory and Its Applications.
  Volume I.}
\newblock John Wiley \&amp; Sons London-New York-Sydney-Toronto.

\bibitem[\protect\citeauthoryear{Genton}{Genton}{2004}]{genton2004}
Genton, M. (2004).
\newblock {\em Skew-elliptical distributions and their applications}.
\newblock Chapman \& Hall/CRC, Boca Raton, FL.
\newblock A journey beyond normality, Edited by Marc G. Genton.

\bibitem[\protect\citeauthoryear{Huser and Davison}{Huser and
  Davison}{2013}]{huser2013}
Huser, R. and A.~C. Davison (2013).
\newblock Composite likelihood estimation for the {B}rown-{R}esnick process.
\newblock {\em Biometrika\/}~{\em 100\/}(2), 511--518.

\bibitem[\protect\citeauthoryear{Huser and Genton}{Huser and
  Genton}{2015}]{huser+g14}
Huser, R. and M.~Genton (2015).
\newblock Non-stationary dependence structures for spatial extremes.
\newblock {\em arXiv:1411.3174v1\/}.

\bibitem[\protect\citeauthoryear{Jamalizadeh, Mehrali, and
  Balakrishnan}{Jamalizadeh et~al.}{2009}]{jamalizadeh2009}
Jamalizadeh, A., Y.~Mehrali, and N.~Balakrishnan (2009).
\newblock Recurrence relations for bivariate $t$ and extended skew-$t$
  distributions and an application to order statistics from bivariate $t$.
\newblock {\em Comput. Statist. Data Anal.\/}~{\em 53\/}(12), 4018--4027.

\bibitem[\protect\citeauthoryear{Joe}{Joe}{1997}]{joe1997}
Joe, H. (1997).
\newblock {\em Multivariate models and dependence concepts}, Volume~73 of {\em
  Monographs on Statistics and Applied Probability}.
\newblock Chapman \& Hall, London.

\bibitem[\protect\citeauthoryear{Kabluchko, Schlather, and De~Haan}{Kabluchko
  et~al.}{2009}]{kabluchko2009stationary}
Kabluchko, Z., M.~Schlather, and L.~De~Haan (2009).
\newblock Stationary max-stable fields associated to negative definite
  functions.
\newblock {\em Ann. Appl. Probab.\/}, 2042--2065.

\bibitem[\protect\citeauthoryear{Ledford and Tawn}{Ledford and
  Tawn}{1996}]{ledford1996}
Ledford, A.~W. and J.~A. Tawn (1996).
\newblock Statistics for near independence in multivariate extreme values.
\newblock {\em Biometrika\/}~{\em 83\/}(1), 169--187.

\bibitem[\protect\citeauthoryear{Lindgren}{Lindgren}{2012}]{lindgren2012}
Lindgren, G. (2012).
\newblock {\em Stationary Stochastic Processes: Theory and Applications}.
\newblock CRC Press.

\bibitem[\protect\citeauthoryear{Lysenko, Roy, and Waeber}{Lysenko
  et~al.}{2009}]{lysenko2009}
Lysenko, N., P.~Roy, and R.~Waeber (2009).
\newblock Multivariate extremes of generalized skew-normal distributions.
\newblock {\em Statist. Probab. Lett.\/}~{\em 79\/}(4), 525--533.

\bibitem[\protect\citeauthoryear{Minozzo and Ferracuti}{Minozzo and
  Ferracuti}{2012}]{minozzo2012existence}
Minozzo, M. and L.~Ferracuti (2012).
\newblock On the existence of some skew-normal stationary processes.
\newblock {\em Chil. J. Stat.\/}~{\em 3}, 157--170.

\bibitem[\protect\citeauthoryear{Nikoloulopoulos, Joe, and Li}{Nikoloulopoulos
  et~al.}{2009}]{nikoloulopoulos2009}
Nikoloulopoulos, A.~K., H.~Joe, and H.~Li (2009).
\newblock Extreme value properties of multivariate {$t$} copulas.
\newblock {\em Extremes\/}~{\em 12\/}(2), 129--148.

\bibitem[\protect\citeauthoryear{Opitz}{Opitz}{2013}]{opitz2013}
Opitz, T. (2013).
\newblock Extremal $t$ processes: Elliptical domain of attraction and a
  spectral representation.
\newblock {\em J. Multivariate Anal.\/}~{\em 122\/}(0), 409 -- 413.

\bibitem[\protect\citeauthoryear{Padoan}{Padoan}{2011}]{padoan2011}
Padoan, S.~A. (2011).
\newblock Multivariate extreme models based on underlying skew-$t$ and
  skew-normal distributions.
\newblock {\em J. Multivariate Anal.\/}~{\em 102\/}(5), 977 -- 991.

\bibitem[\protect\citeauthoryear{Padoan, Ribatet, and Sisson}{Padoan
  et~al.}{2010}]{padoan2010}
Padoan, S.~A., M.~Ribatet, and S.~A. Sisson (2010).
\newblock Likelihood-based inference for max-stable processes.
\newblock {\em J. Amer. Statist. Assoc.\/}~{\em 105\/}(489), 263--277.

\bibitem[\protect\citeauthoryear{Schlather}{Schlather}{2002}]{schlather2002a}
Schlather, M. (2002).
\newblock Models for stationary max-stable random fields.
\newblock {\em Extremes\/}~{\em 5\/}(1), 33--44.

\bibitem[\protect\citeauthoryear{Smith}{Smith}{1990}]{Smith1990a}
Smith, R.~L. (1990).
\newblock Max-stable processes and spatial extremes.
\newblock University of Surrey 1990 technical report.

\bibitem[\protect\citeauthoryear{Varin, Reid, and Firth}{Varin
  et~al.}{2011}]{varin2011}
Varin, C., N.~Reid, and D.~Firth (2011).
\newblock An overview of composite likelihood methods.
\newblock {\em Statist. Sinica\/}~{\em 21\/}(1), 5--42.

\bibitem[\protect\citeauthoryear{Wood and Chan}{Wood and Chan}{1994}]{wood1994}
Wood, A. T.~A. and G.~Chan (1994).
\newblock Simulation of stationary {Gaussian} processes in $[0, 1]^d$.
\newblock {\em J. Comput. Graph. Statist.\/}~{\em 3\/}(4), 409--432.

\bibitem[\protect\citeauthoryear{Zhang and El-Shaarawi}{Zhang and
  El-Shaarawi}{2010}]{zhang2010spatial}
Zhang, H. and A.~El-Shaarawi (2010).
\newblock On spatial skew-{Gaussian} processes and applications.
\newblock {\em Environmetrics\/}~{\em 21\/}(1), 33--47.

\end{thebibliography}

Simone A. Padoan, Department of Decision Sciences, Bocconi University, via Roentgen, 1, 20136, Milan, Italy.\\
Email: simone.padoan@unibocconi.it. 

%
%
\appendix
\section{Appendix A: Proofs}
\subsection{Proof of Proposition \ref{pro:prop_nskewt}}\label{ssec:prop_nskewt}

Items (1)--(3) are easily derived following the proof of Propositions (1)--(4) of \citet{arellano2010} and taking into account the next result.
\begin{lem}\label{lem:cond_nc_t}
Let $Y=(Y_1^\top,Y_2^\top)^\top \sim \cT_d ( \mu , \Omega, \kappa, \nu )$, where
$Y_1 \in \real$ and $Y_2 \in \real^{d-1}$ with the corresponding partition of the parameters $(\mu,\Omega,\nu)$ and $\kappa=(\kappa_1,0^\top)^\top$ with
$\kappa_1\in \real$.
Then, 
$$( Y_1 | Y_2 = y_2 ) \sim 
\cT( \mu_{1 \cdot2}, \Omega_{11\cdot2}, \kappa_{1\cdot2}, \nu_{1\cdot2} ), \qquad y_2 \in \real^{d-1} 
$$
where
$\mu_{1\cdot2} = \mu_1 + \Omega_{12} \Omega_{22}^{-1} (y_2-\mu_2)$,  
$\Omega_{1\cdot2} = \zeta_2 \Omega_{11\cdot 2}$,
$\zeta_2=\{\nu+ Q_{\Omega_{22}^{-1}}(z_2)\}/(\nu+d_2)$,
$z_2 = \omega_{2}^{-1}(y_2 - \mu_2)/\Omega_2$,
$\omega_2=\diag(\Omega_{22})^{1/2}$,
$\Omega_{11\cdot 2} = \Omega_{11} - \Omega_{12}\Omega_{22}^{-1}\Omega_{21}$,
$\kappa_{1\cdot 2} = \zeta_2^{-1/2} \kappa$,
$\nu_{1\cdot 2}=\nu +d-1$.
\end{lem}
\begin{proof}[Proof of Lemma \ref{lem:cond_nc_t}]
The marginal density of $Y_2$ is equal to
$$
f_{Y_2}(y_2) = \int_0^{\infty} \frac{v^{\nu/2-1} e^{-v}}{\Gamma ( \nu/2 )}
\phi_{d-1} \left( \frac{y_2 - \mu_2}{\sqrt{\frac{\nu}{2v}}}; \Omega_{22}  \right) 
\left(\frac{2v}{\nu}\right)^{(d-1)/2} \der v =  \psi_{d-1}(y_2;\mu_2,\Omega_{22},\nu),
$$
namely it is a $(d-1)$-dimensional central $t$ pdf. 
The joint density of $Y$ is equal to
\begin{align*}
& f_{Y_2}(y_2) f_{Y_1|Y_2=y_2}(y_1)\\ 
&= 
\psi_{d-1}(y_2;\mu_2,\Omega_{22},\nu)
\int_0^{\infty} \frac{v^{(\nu+d-1)/2-1} e^{-v}}{\Gamma(\frac{\nu+d-1}{2})}
\phi \left\{ (\Omega_{1\cdot 2})^{-1/2}(y_1 - \mu_{1\cdot 2})
\sqrt{\frac{2v}{\nu+d-1}} 
-(\Omega_{11\cdot 2})^{-1/2}\kappa_1 \right\} \der v \\
&= \int_0^{\infty}\frac{(\Omega_{11\cdot 2})^{-1/2}v^{\nu/2-1} e^{-v}}{\Gamma(\frac{\nu}{2})}
\left(\frac{2v}{\nu}\right)^{d/2}
\phi_{d-1}\left(\frac{y_2-\mu_2}{\sqrt{\frac{\nu}{2v}}} \right)
\phi \left\{
(\Omega_{11\cdot 2})^{1/2}\left(y_1 - \mu_{1\cdot 2} \right)
\sqrt{\frac{2v}{\nu}}-\kappa_1  
\right\} \der v \\
&= \int_0^{\infty}\frac{v^{\nu/2-1} e^{-v}}{\Gamma(\frac{\nu}{2})} 
\phi_d \left\{
\left(\begin{array}{c} y_1 - \mu_1 - \kappa_1 \sqrt{\frac{\nu}{2v}} \\ 
y_2 -\mu_2 \end{array}\right);
\sqrt{\frac{\nu}{2v}} \Omega 
\right\} \der v.
\end{align*}
\end{proof}
%



%
\subsection{Proof of Proposition \ref{pro:limit_skew_t}}\label{ssec:limit_skew_t}

Let $Z\sim \skt(\alpha,\tau,\kappa,\nu)$. Then $1-\Psi(x;\alpha,\tau,\nu)\approx x^{-\nu}\cL(x;\alpha,\tau,\nu)$ as $x\rightarrow+\infty$, 
for any $\nu>1$, where
$$
\cL(x;\alpha,\tau,\kappa,\nu)=
\frac{\Gamma\{(\nu+1)/2\}\Psi(\alpha\sqrt{\nu+1};\nu+1)}
{\Gamma(\nu/2)\sqrt{\pi}\nu^{3/2}\Psi(\tau/\sqrt{1+\alpha^2};\kappa/\sqrt{1+\alpha^2},\nu)}
\left(\frac{1}{x^2}+\frac{1}{\nu}\right)^{-(\nu+1)/2}
$$
is a slowly varying function \citep[e.g][Appendix B]{dehaan2006}. From Corollary 1.2.4 in \citet{dehaan2006}, it 
follows that the normalisation constants are $a_n=\Psi^{\leftarrow}(1-1/n;\alpha,\tau,\kappa,\nu)$, where $\Psi^{\leftarrow}$ is
the inverse function of $\Psi$, and $b_n=0$, and therefore
$a_n= \{n\cL(\alpha,\tau,\kappa,\nu)\}^{1/\nu}$, where $\cL(\alpha,\tau,\kappa,\nu)\equiv \cL(\infty;\alpha,\tau,\kappa,\nu)$.
Applying Theorem 1.2.1 in \citet{dehaan2006} we obtain that  $M_n/a_n\wconv U$, where $U$ has $\nu$-Fr\'{e}chet univariate marginal distributions.

Let $Z\sim\skt_d(\bar{\Omega},\alpha,\tau,\kappa,\nu)$. For any $j\in\{1,\ldots,d\}$ consider the partition $Z=(Z_{j},Z_{I_j}^\top)^\top$, where $I_j=\{1,\ldots,d\}\backslash j$ and $Z_j=Z_{\{j\}}$,
and
the respective partition of $(\bar{\Omega},\alpha)$. Define $a_n=(a_{n,1},\ldots,a_{n,d})$, where
$a_{n,j}=\{n\cL(\alpha^*_{j},\tau^*_{j},\kappa^*_{j},\nu)\}^{1/\nu}$ and $\alpha^*_j=\alpha^*_{\{j\}}$, $\tau^*_{j}=\tau^*_{\{j\}}$ and $\kappa^*_{j}=\kappa^*_{\{j\}}$ are the marginal parameters (\ref{eq:margparams}) under Proposition \ref{pro:prop_nskewt}(\ref{prop1A}).
From Theorem 6.1.1 and Corollary 6.1.3 in \citet{dehaan2006}, $M_n/a_n\wconv U$, where the distribution of $U$ is 
$G(x)=\exp\{-V(x)\}$ with 
$V(x)=\lim_{n\rightarrow+\infty}n\{1-\prob(Z_1\leq a_{n,1}x_1,\ldots,Z_d\leq a_{n,d}x_d)\}$ for all $x=(x_1,\ldots,x_d)^\top\in \reald_+$.
Applying the conditional tail dependence function framework of \cite{nikoloulopoulos2009} it follows that
$$
V(x_j,i\in I)=\lim_{\goinf}\sum_{j=1}^dx_j^{-\nu}\prob(Z_{i}\leq a_{n,i}x_i, i\in I_j | Z_{j}=a_{n,j}x_j).
$$
From the conditional distribution in Proposition \ref{pro:prop_nskewt}(\ref{prop1A}) we have that
$$
\left\{
\left(
\frac{Z_i-a_{n,j}x_j}{\{\zeta_{n,j}(1-\omega_{i,j}^2)\}^{1/2}},i\in I_j
\right)^\top
|Z_j=a_{n,j}x_j
\right\}
\sim \skt_{d-1}
\left(
\bar{\Omega}^+_{j},\alpha^+_j,\tau_{n,j},\kappa_{n,j},\nu+1
\right),
$$
for $j\in\ldots{1,\ldots,d}$, where $\bar{\Omega}^+_{j}=\omega_{I_jI_j\cdot j}^{-1}\Omega_{I_jI_j\cdot j} \omega_{I_jI_j\cdot j}^{-1}$,
$\omega_{I_jI_j\cdot j}=\text{diag}(\Omega_{I_jI_j\cdot j})^{1/2}$, 
$\bar{\Omega}_{I_jI_j\cdot j}=\bar{\Omega}_{I_jI_j}-\bar{\Omega}_{I_jj}\bar{\Omega}_{jI_j}$,
$\alpha_j^+=\bar{\Omega}_{I_jI_j\cdot j}\alpha_{I_j}$
$
\zeta_{n,j}=[\nu+(a_{n,j}x_j)^2]/(\nu+1),$
$\tau_{n,j}=[(\bar{\Omega}_{jI_j}\alpha_{I_j}+\alpha_j)a_{n,j}x_j + \tau]/\zeta_{n,j}^{1/2}$ and
$\kappa_{n,j} = \kappa/\zeta_{n,j}^{1/2}.
$
Now, for any $j\in\{1,\ldots,d\} $ and all $i\in I_j$
$$
\frac{a_{n,i}x_i-a_{n,j}x_j}{\{\zeta_{n,j}(1-\omega_{i,j}^2)\}^{1/2}}\rightarrow 
\frac{(x^+_i/x^+_j-\omega_{i,j})(\nu+1)^{1/2}}{\{(1-\omega_{i,j})\}^{1/2}} \quad\mbox{as } n \rightarrow+\infty,
$$ 
where $\omega_{i,j}$ is the $(i,j)$-th element of $\bar{\Omega}$,
$x^+_j=x_j\cL^{1/\nu}(\alpha^*_{j},\tau^*_{j},\kappa^*_{j},\nu)$  
and
$\tau_{n,j}\rightarrow\tau^+_{j}=(\bar{\Omega}_{jI_j}\alpha_{I_j}+\alpha_j)(\nu+1)^{1/2},$
and $\kappa_{n,j}\rightarrow 0$ as $n \rightarrow+\infty$. 
As a consequence
$$
V(x_j, j\in I)=
\sum_{j=1}^d x_j^{-\nu} \Psi_{d-1}
\left(\left(
\sqrt{\frac{\nu+1}{1-\omega^2_{i,j}}}
\left(
\frac{x^+_i}{x^+_j} - \omega_{i,j}
\right),i\in I_j\right)^\top; \bar{\Omega}^+_j, \alpha^+_j, \tau^+_j,\nu+1
\right).
$$
%

%
\subsection{Proof of Proposition \ref{pro:tails_sn}}\label{ssec:tails_sn}

Recall that  if $Z\sim \skn_2(\bar{\Omega},\alpha)$, then $Z_j\sim\skn(\alpha^*_{j})$
and $Z_j|Z_{3-j}\sim\skn(\alpha_{j\cdot 3-j})$ for $j=1,2$ (e.g. \citet[][Ch. 2]{azzalini2013skew} or Proposition \ref{pro:prop_nskewt}), where
$$
\alpha^*_{j}=\frac{\alpha_j+\omega\alpha_{3-j}}{\sqrt{1+\alpha_{3-j}^2(1-\omega^2)}},\quad
\alpha_{j\cdot 3-j}=\alpha_j\sqrt{1-\omega^2}.
$$
Define $x_j(u)=\Phi^{\leftarrow}(1-u;\alpha^*_{j})$, for any $u\in[0,1]$, where
$\Phi^{\leftarrow}(\cdot;\alpha^*_{j})$ is the inverse  of the marginal distribution function
$\Phi(\cdot;\alpha^*_{j})$, $j=1,2$. The asymptotic behaviour of $x_j(u)$ as $u\rightarrow 0$ 
is
\begin{equation}\label{eq:quantile}
x_j(u)=\left\{
\begin{tabular}{lc}
$x(u)$, & if\, $\alpha^*_{j}\geq0$\\
$x(u)/\bar{\alpha}_{j}-\{2\log(1/u)\}^{-1/2}\log(\sqrt{\pi}\alpha^*_{j})$, & if\, $\alpha^*_{j}<0$
\end{tabular}
\right.
\end{equation}
for $j=1,2$, where $\bar{\alpha}_{j}=\{1+\alpha^{*2}_{j}\}^{1/2}$ and 
$
x(u) \approx \{2\log(1/u)\}^{1/2}-\{2\log(1/u)\}^{-1/2}\{\log\log(1/u)+\log(2\sqrt{\pi})\}
$
\citep[][]{padoan2011}.
The limiting behaviour of the joint survivor function of the bivariate skew-normal distribution is described by
\begin{equation}\label{eq:tail_prob}
p(u)=\prob\{Z_1>x_1(u),Z_2>x_2(u)\},\qquad u\rightarrow 0.
\end{equation}
For case (a), when $\alpha_1,\alpha_2>0$, then $x_1(u)=x_2(u)=x(u)$, 
and the joint upper tail \eqref{eq:tail_prob} behaves as
\begin{align}\label{eq:joint_tail}
p(u)&=\int_{x(u)}^\infty \left\{1-\Phi\left(\frac{y(u)-\omega v}{\sqrt{1-\omega^2}};\alpha_{1\cdot 2}\right)\right\}\phi(v;\alpha^*_{2})\der v\nonumber\\
&\approx \frac{\sqrt{1-\omega^2}}{x(u)}\int_0^\infty\frac{\phi_2(x(u),x(u)+t/x(u);\bar{\Omega},\alpha)}
{x(u)(1-\omega)-\omega t/x(u)}\der t\nonumber\\
&\approx \frac{e^{-x^2(u)/(1+\omega)}}{\pi(1-\omega)x^2(u)}\left(
\int_0^\infty e^{-t/(1+\omega)}\der t -\frac{e^{-x^2(u)(\alpha_1+\alpha_2)^2/2}}{\sqrt{2\pi}(\alpha_1+\alpha_2)x(u)}
\int_0^\infty e^{-t\{1/(1+\omega)+\alpha_2(\alpha_1+\alpha_2)\}} \der t
\right)\nonumber\\
&= 
\frac{e^{-x^2(u)/(1+\omega)}(1+\omega)}{\pi(1-\omega)x(u)^2}
\left(
1-\frac{e^{-x^2(u)(\alpha_1+\alpha_2)^2/2}}{\sqrt{2\pi}(\alpha_1+\alpha_2)\{1+\alpha_2(\alpha_1+\alpha_2)(1+\omega)\}x(u)}
\right),
\end{align}
as $u\rightarrow 0$.
The first approximation is obtained by using $1-\Phi(x;\alpha)\approx \phi(x;\alpha)/x$ as $x\rightarrow+\infty$, 
when $\alpha>0$ \citep[][]{padoan2011}. The second approximation uses $1-\Phi(x)\approx \phi(x)/x$ as 
$x\rightarrow+\infty$ \citep{feller1968}.
Let $X_j=\{-1/\log\Phi(Z_j;\alpha^*_{j})\}$, $j=1,2$.
Substituting $x(u)$ into \eqref{eq:joint_tail} substituting and using the approximation $1-\prob(X_j>x)\approx 1/x$ as $x\rightarrow\infty$, $j=1,2$,
we obtain that \eqref{eq:tail_prob} with common unit Fr\'{e}chet margins behaves asymptotically as 
$
\cL(x)\; x^{-2/(1+\omega)},\mbox{ as }x \rightarrow +\infty,
$
where
\begin{equation}\label{eq:first_case_sl_fun}
\cL(x)=\frac{2(1+\omega)(4\pi\log x)^{-\omega/(1+\omega)}}{1-\omega}
\left(
1-
\frac{(4\pi\log x)^{\{(\alpha_1+\alpha_2)^2-1\}/2}\,x^{-(\alpha_1+\alpha_2)^2}}
{(\alpha_1+\alpha_2)\{1+\alpha_2(\alpha_1+\alpha_2)(1+\omega)\}}
\right).
\end{equation}
As the second term in the parentheses in \eqref{eq:first_case_sl_fun} is 
$o(x^{(\alpha_1+\alpha_2)})$, then the quantity inside the parentheses $\rightarrow1$ rapidly as $x\rightarrow\infty$, and so $\cL(x)$ is well approximated by the first term in 
\eqref{eq:first_case_sl_fun}.
When $\alpha_2<0$ and $\alpha_1\geq-\alpha_2/\omega$, then $\alpha^*_{1},\alpha^*_{2}>0$ and we obtain the same outcome.

For case (b), when $\alpha_2<0$ and $-\omega,\alpha_2\leq \alpha_1<-\omega^{-1}\alpha_2$, then $\alpha^*_{1}\geq 0$ and
$\alpha^*_{2}<0$ and hence $x_1(u)=x(u)$ and $x_2(u)\approx x(u)/\bar{\alpha}_{2}$ as $u\rightarrow0$. 
When $\alpha_1>-\bar{\alpha}_{2}\alpha_2$, then following a similar derivation to those in
\eqref{eq:joint_tail}, we obtain that 
$$
p(u)\approx\frac{\bar{\alpha}_{2}^2(1-\omega^2)(1-\omega\bar{\alpha}_{2})^{-1}}
{\pi(\bar{\alpha}_{2}-\omega)x^2(u)}
\exp\left[-\frac{x^2(u)}{2}\left\{\frac{1-\omega^2+(\bar{\alpha}_{2}-\omega)^2}{(1-\omega^2)\bar{\alpha}_{2}^2}\right\}\right],
\quad\mbox{as } u\rightarrow 0.
$$
Similarly, when $\alpha_1<-\bar{\alpha}_{2}\alpha_2$, and noting that $\Phi(x)\approx -\phi(-x)/x$ as $x\rightarrow-\infty$,  then
$$
p(u)\approx\frac{-\bar{\alpha}_{2}^2\{1-\omega\bar{\alpha}_{2} +\alpha_2(\alpha_2+\alpha_1\bar{\alpha}_{2})
(1-\omega^2)\}^{-1}}
{\pi(\bar{\alpha}_{2}-\omega)(1-\omega^2)^{-1}(\alpha_1+\alpha_2/\bar{\alpha}_{2})x^3(u)}
e^{-\frac{x^2(u)}{2}\left\{\frac{1-\omega^2+(\bar{\alpha}_{2}-\omega)^2}{(1-\omega^2)\bar{\alpha}_{2}^2}
+\left(\alpha_1+\frac{\alpha_2}{\bar{\alpha}_{2}}\right)^2
\right\}}, \quad\mbox{ as } u\rightarrow 0.
$$
For case (c), when $\alpha_2<0$ and $0<\alpha_1<-\omega\alpha_2$, then $\alpha^*_{1},\alpha^*_{2}<0$ and hence 
$x_1(u)\approx x(u)/\bar{\alpha}_{1}$ and
$x_2(u)\approx x(u)/\bar{\alpha}_{2}$ as $u\rightarrow0$. Then as $u\rightarrow 0$ we have
\begin{align*}
p(u)&\approx\frac{-\bar{\alpha}_{2}^{3/2}\bar{\alpha}_{1}^2(1-\omega^2)
(\bar{\alpha}_{2}-\omega\bar{\alpha}_{1})^{-1}(\alpha_1\bar{\alpha}_{2}+\alpha_2\bar{\alpha}_{1})^{-1}}
{\pi\{1-\omega\bar{\alpha}_{2}+\alpha_2(\alpha_2+\alpha_1\bar{\alpha}_{2}/\bar{\alpha}_{1})(1-\omega^2)\}x^3(u)}\\
&\times
\exp\left[-\frac{x^2(u)}{2(1-\omega^2)}
\left(
\frac{\alpha_1^2(1-\omega^2)+1}{\bar{\alpha}_{1}^2}+
\frac{\alpha_2^2(1-\omega^2)+1}{\bar{\alpha}_{2}^2}+
\frac{2(\alpha_1\alpha_2(1-\omega^2)-\omega)}{\bar{\alpha}_{1}\bar{\alpha}_{2}}
\right)
\right]
\quad u\rightarrow 0.
\end{align*}
When $\alpha_1,\alpha_2<0$ and $\omega_2^{-1}\alpha_2\leq \alpha_1<0$  the same argument holds.
Finally, interchanging $\alpha_1$ with $\alpha_2$ produces the same results but 
substituting $\alpha_{j}$ and $\bar{\alpha}_{j}$ with $\alpha_{3-j}$ and $\bar{\alpha}_{3-j}$ respectively, for $j=1,2.$


%
%

%
\subsection{Proof of Theorem \ref{teo:main_res}}\label{ap:proof_th_1}

Let $Y(s)$ be a skew-normal process with finite dimensional distribution $\skn_d(\bar{\Omega},\alpha,\tau)$.
For any $j\in I=\{1,\ldots,d\}$ consider the partition $Y=(Y_j,Y_{I_j}^\top)^\top$, where $I_j=I\backslash j$, $Y_{j}=Y_{\{j\}}=Y(s_j)$ and $Y_{I_j}=(Y_i,i\in I_j)^\top$, and the respective partition of $(\bar{\Omega},\alpha)$.
The exponent function 
\eqref{eq:expo_mea} is
$$
V(x_j,j\in I)=\expect\left[\max_{j}\left\{\frac{(Y^+_j/x_j)^\xi}{m^+_j}\right\}\right]
=\int_{\reald}\max_{j}\left\{\frac{(y_j/x_j)^\xi}{m_j^+ },0\right\}\phi_d(y;\bar{\Omega};\alpha,\tau) \der y,
$$
where $x_j\equiv x(s_j)$, $y_j\equiv y(s_j)$ and $m_j^{+}\equiv m^+(s_j)$. Then
\begin{equation}\label{eq:expo_extskt}
V(x_j,j\in I)=\sum_{j=1}^dV_j,\quad
V_j=\frac{1}{m^{+}_j}
\int_0^\infty \left(\frac{y_j}{x_j}\right)^{\nu} \int_{-\infty}^{y_jx_{I_j}/x_j}
\phi_d(y;\bar{\Omega};\alpha,\tau) \der y_{I_j} \der y_j,
\end{equation}
where $x_{I_j}=(x_i, i\in I_j)^\top$ and $y_{I_j}=(y_i, i\in I_j)^\top$. 
As $Y_j\sim\skn(\alpha^*_{j},\tau^*_{j})$, where $\alpha^*_{j}=\alpha^*_{\{j\}}$ and $\tau_{j}^*=\tau^*_{\{j\}}$ are the marginal parameters derived from Proposition \ref{pro:prop_nskewt}(\ref{prop1A}), then
\begin{align*}
m_j^+=\int_0^\infty y_j^\nu\,
\phi(y_j;\alpha^*_{j},\tau^*_{j})\der y_j
&=
\frac{1}{\Phi\{\tau^*_{j}(1+\alpha^{*2}_{j})^{-1/2}\}}\int_0^\infty y_j^\nu\, \phi(y_j)\Phi(\alpha^*_{j}y_j+\tau^*_{j})
\der y_j\\
&=
\frac{2^{(\nu-2)/2} \Gamma\{(\nu+1)/2\}
\Psi(\alpha^*_{j}
\sqrt{\nu+1};-\tau^*_{j},\nu+1) 
}
{\sqrt{\pi}\Phi[\tau\{1+Q_{\bar{\Omega}}(\alpha)\}^{-1/2}]}
\end{align*}
by observing that
$
\tau^*_{j}\{1+\alpha^{*2}_{j}\}^{1/2}=\tau\{1+Q_{\bar{\Omega}}(\alpha)\}^{-1/2}$. 
For $j=1,\ldots,d$ define
$x^\circ_j=x_j(m^+_j)^{1/\nu}$ and $m_j^{+}=\bar{m}_j^{+}/\Phi[\tau\{1+Q_{\bar{\Omega}}(\alpha)\}^{-1/2}],
$
where
$
\bar{m}_j^{+}=(\pi)^{1/2}2^{(\nu-2)/2} \Gamma\{(\nu+1)/2\}
\Psi(\alpha^*_{j}
\sqrt{\nu+1};-\tau^*_{j},\nu+1).
$
Then, for any $j=1,\ldots,d$
\begin{align*}
V_j&=\frac{1}{m^{+}_j}
\int_0^\infty \left(\frac{y_j}{x_j}\right)^{\nu} \int_{-\infty}^{y_jx_{I_j}/x_j}
\phi_d(y;\bar{\Omega},\alpha,\tau) \der y_{I_j} \der y_j\\
&=\frac{1}{\bar{m}_j^{+}}
\int_0^{\infty} 
\left(\frac{y_j}{x_j}\right)^{\nu}
\int_{-\infty}^{y_jx_{I_j}/x_j}
\phi_d(y;\Omega)\Phi(\alpha^\top y+\tau) \der y_{I_j}\der y_j\\
&=\frac{1}{\bar{m}_j^{+}}
\int_0^{\infty} 
\left(\frac{y_j}{x_j}\right)^{\nu}\phi(y_j)
\int_{-\infty}^{y_jx_{I_j}/x_j}
\phi_{d-1}(y_{I_j}-y_j\bar{\Omega}_{j,I_j};\bar{\Omega}^\circ_j)\Phi(\alpha^\top y+\tau)\der y_{I_j}\der y_j\\
&=
\frac{1}{\bar{m}_j^+} \int_0^{\infty}
\left(\frac{y_j}{x_j}\right)^{\nu}
\phi(y_j) \Phi_d \left(y_j^{\circ}; \Omega_j^{\circ\circ}
 \right)\der y_j,
\end{align*}
where
\begin{equation*}
y_j^{\circ}=
\left(
\begin{tabular}{c}
$y_j\,\omega_{I_jI_j\cdot j}^{-1}(x^{\circ}_{I_j}/x^{\circ}_j-\bar{\Omega}_{I_j j})$\\
$y_j\alpha^*_{j}+\tau^*_{j}$
\end{tabular}
\right),
\end{equation*}
with $\omega_{I_jI_j\cdot j}=\text{diag}(\bar{\Omega}_{I_jI_j\cdot j})^{1/2}$,
$\bar{\Omega}_{I_jI_j\cdot j}=\bar{\Omega}_{I_jI_j}-\bar{\Omega}_{I_jj}\bar{\Omega}_{jI_j}$,
$
y_j\alpha^*_{j}+\tau^*_{j}
=
\frac{y_j(\alpha_j+\bar{\Omega}_{jj}^{-1}\bar{\Omega}_{jI_j}\alpha_{I_j})+\tau}
{\{1+Q_{\bar{\Omega}_{I_jI_j\cdot j}}(\alpha_{I_j})\}^{1/2}}
$
and
\begin{equation*}
\Omega_j^{\circ\circ}=\left(
\begin{array}{cc}
\bar{\Omega}_j^\circ & 
-\frac{\bar{\Omega}_{I_jI_j\cdot j}\omega_{I_jI_j\cdot j}^{-1}\alpha_{I_j}}
{\{1+ Q_{\bar{\Omega}_{I_jI_j\cdot j}}(\alpha_{I_j})\}^{1/2}}\\
-\left(\frac{\bar{\Omega}_{I_jI_j\cdot j}\omega_{I_jI_j\cdot j}^{-1}\alpha_{I_j}}
{\{1+ Q_{\bar{\Omega}_{I_jI_j\cdot j}}(\alpha_{I_j})\}^{1/2}}\right)^\top & 1
\end{array}
\right),
\end{equation*}
where 
$\bar{\Omega}_j^\circ=\omega_{I_jI_j\cdot j}^{-1}\,\bar{\Omega}_{I_jI_j\cdot j}\,\omega_{I_jI_j\cdot j}^{-1}$ and
$
\frac{\Omega_{I_jI_j\cdot j}\omega_{I_jI_j\cdot j}^{-1}\alpha_{I_j}}
{\{1+ Q_{\Omega_{I_jI_j\cdot j}}(\alpha_{I_j})\}^{1/2}}=
\frac{\Omega^\circ_{j}\,\omega_{I_jI_j\cdot j}\,\alpha_{I_j}}
{\{1+ Q_{\bar{\Omega}^\circ_{j}}(\omega_{I_jI_j\cdot j}\alpha_{I_j})\}^{1/2}}.
$

Applying Dutt's \citep{dutt1973} probability integrals we obtain
\begin{align*}
V_j&=
\frac{1}{\bar{m}_j^+} \int_0^{\infty}
\left(\frac{y_j}{x_j}\right)^{\nu}
\phi(y_j) \Phi_d \left(y_j^{\circ}; \Omega_j^{\circ\circ}
 \right)\der y_j,\\
&=\frac{1}{x_j^{\nu}}
 \frac{\Psi_{d+1} \left(
 \left(\left(
\sqrt{\frac{\nu+1}{1-\omega^2_{i,j}}}
\left(
\frac{x^\circ_i}{x^\circ_j} - \omega_{i,j}
\right),i\in I_j\right),
\alpha^*_{j}\sqrt{\nu+1}\right)^\top;
\Omega_j^{\circ\circ},
 \left(
 0, -\tau^*_{j}
 \right)^\top,
\nu+1
\right)}
{\Psi(\alpha^*_{j}\sqrt{\nu+1};-\tau^*_{j},\nu+1)}.
\end{align*}
This is recognised as the form of
a $(d-1)$-dimensional non-central extended skew-$t$ distribution with $\nu+1$ degrees of freedom \citep{jamalizadeh2009}, 
from which $V_j$ can be expressed as
$$
V_j=\frac{1}{x_j^{\nu}}\Psi_{d-1}
\left(
 \left(
\sqrt{\frac{\nu+1}{1-\omega^2_{i,j}}}
\left(
\frac{x^\circ_i}{x^\circ_j} - \omega_{i,j}
\right),i\in I_j\right)^\top;
\bar{\Omega}^\circ_j,\alpha^\circ_j,\tau^\circ_j,\kappa^\circ_j,\nu+1
\right)
$$
for $j=1,\ldots,d$ where
$\alpha_j^{\circ}=\omega_{I_jI_j\cdot j}\,\alpha_{I_j}$,
$\tau_{j}^{\circ}=(\bar{\Omega}_{jI_j}\alpha_{I_j}+\alpha_j)(\nu+1)^{1/2}$ and
$\kappa^{\circ}_j=-\{1+Q_{\bar{\Omega}_{I_jI_j\cdot j}}(\alpha_{I_j})\}^{-1/2}\tau.$
Substituting the expression for $V_j$ into \eqref{eq:expo_extskt} then gives the required
the exponent function.

\newpage

\section{Supplementary material for `Models for extremal dependence derived from skew-symmetric families' by B. Beranger, S. A. Padoan and S. A. Sisson}

This document/appendix contains technical details for deriving the bivariate, trivariate and quadrivariate
densities of the extremal-skew-$t$ model described in the paper, some graphical illustration and simulation results for
the extremal-$t$ process.

%
%
%
\subsection{Plots of the angular density of the extremal-skew-$t$ model}\label{sec:angular}

\begin{figure}[tbh]
\centering
\makebox{
\includegraphics[scale=0.25]{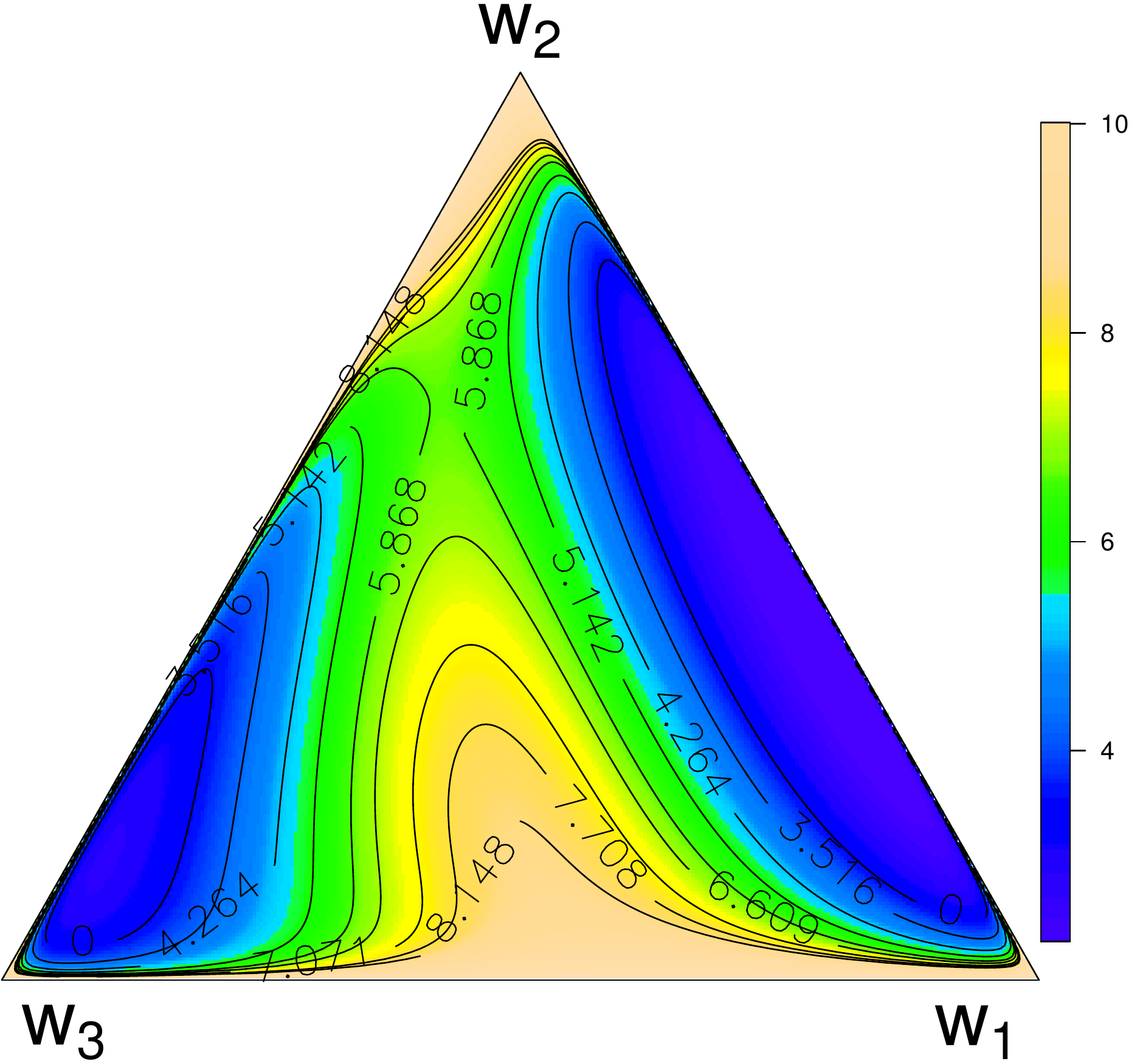}
\includegraphics[scale=0.25]{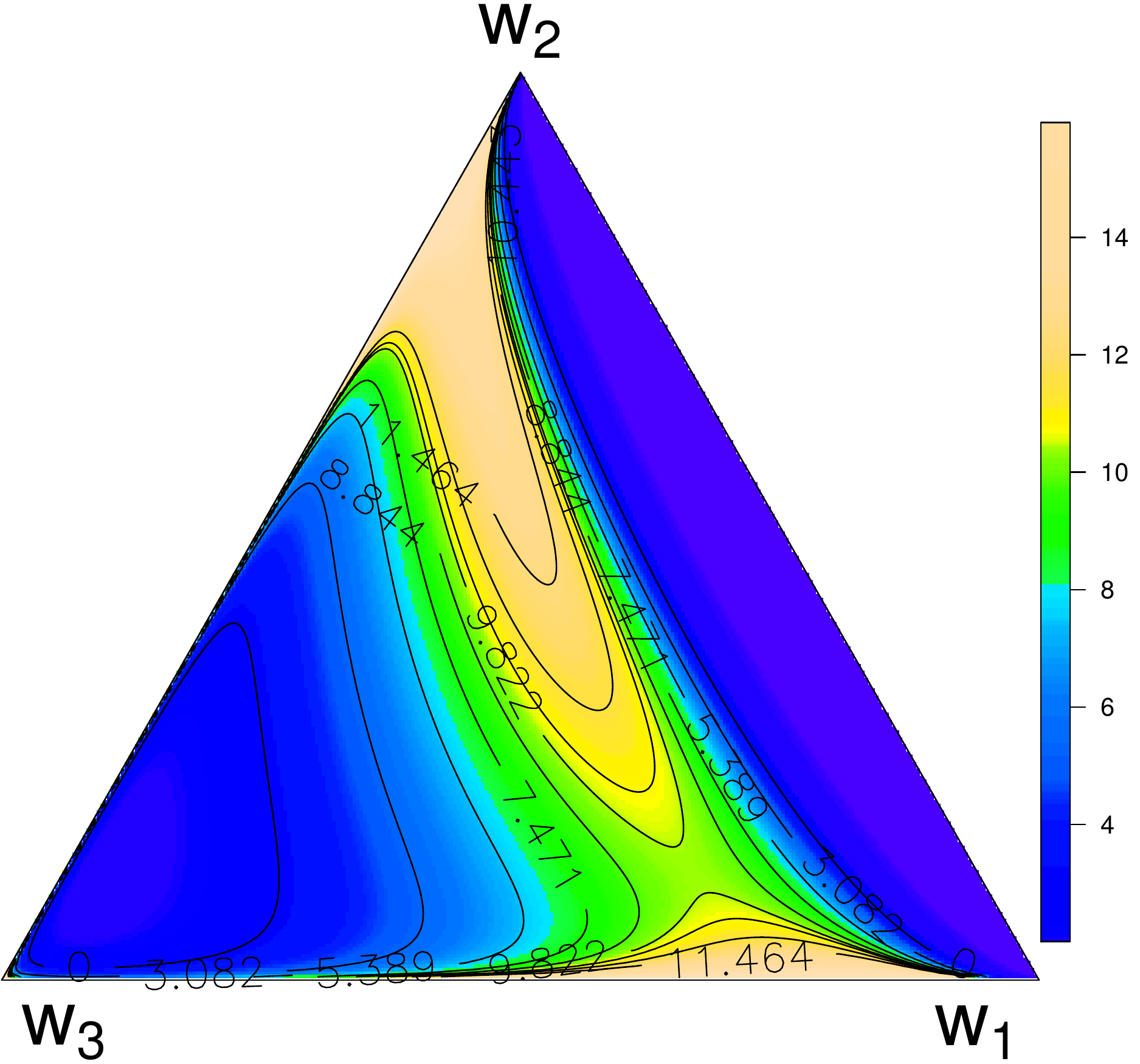}
\includegraphics[scale=0.25]{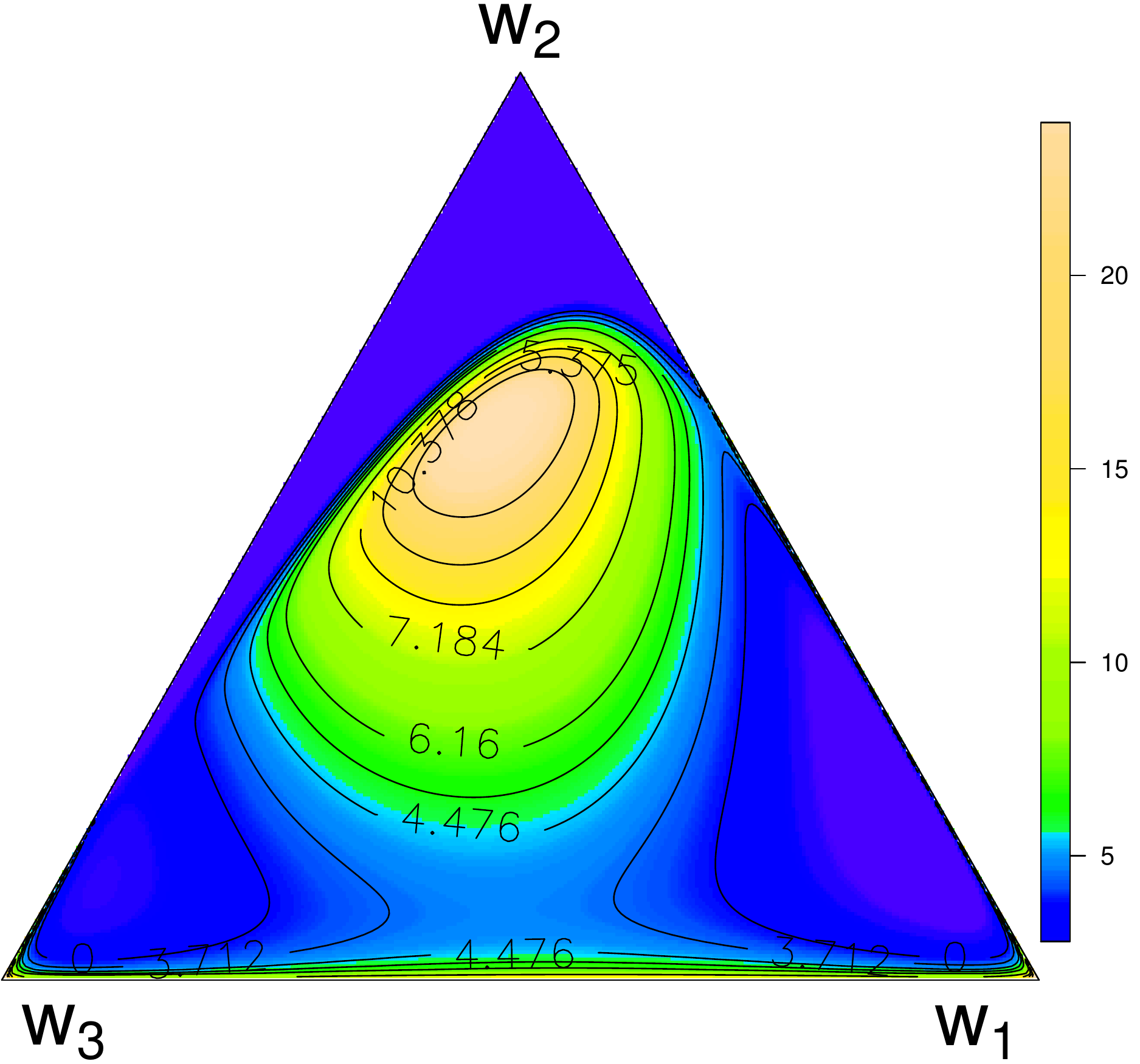}}\\\makebox{
\includegraphics[scale=0.25]{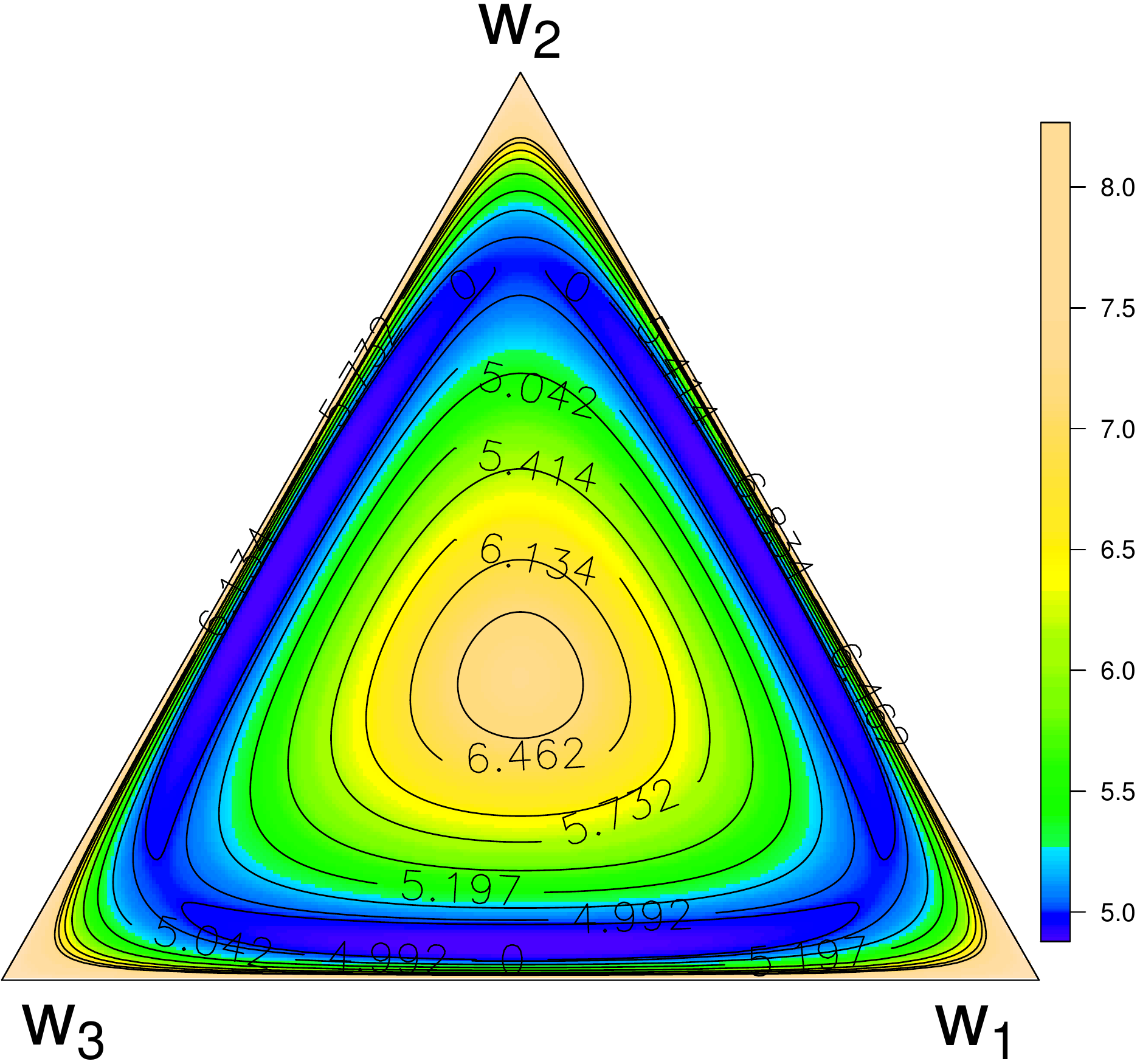}
\includegraphics[scale=0.25]{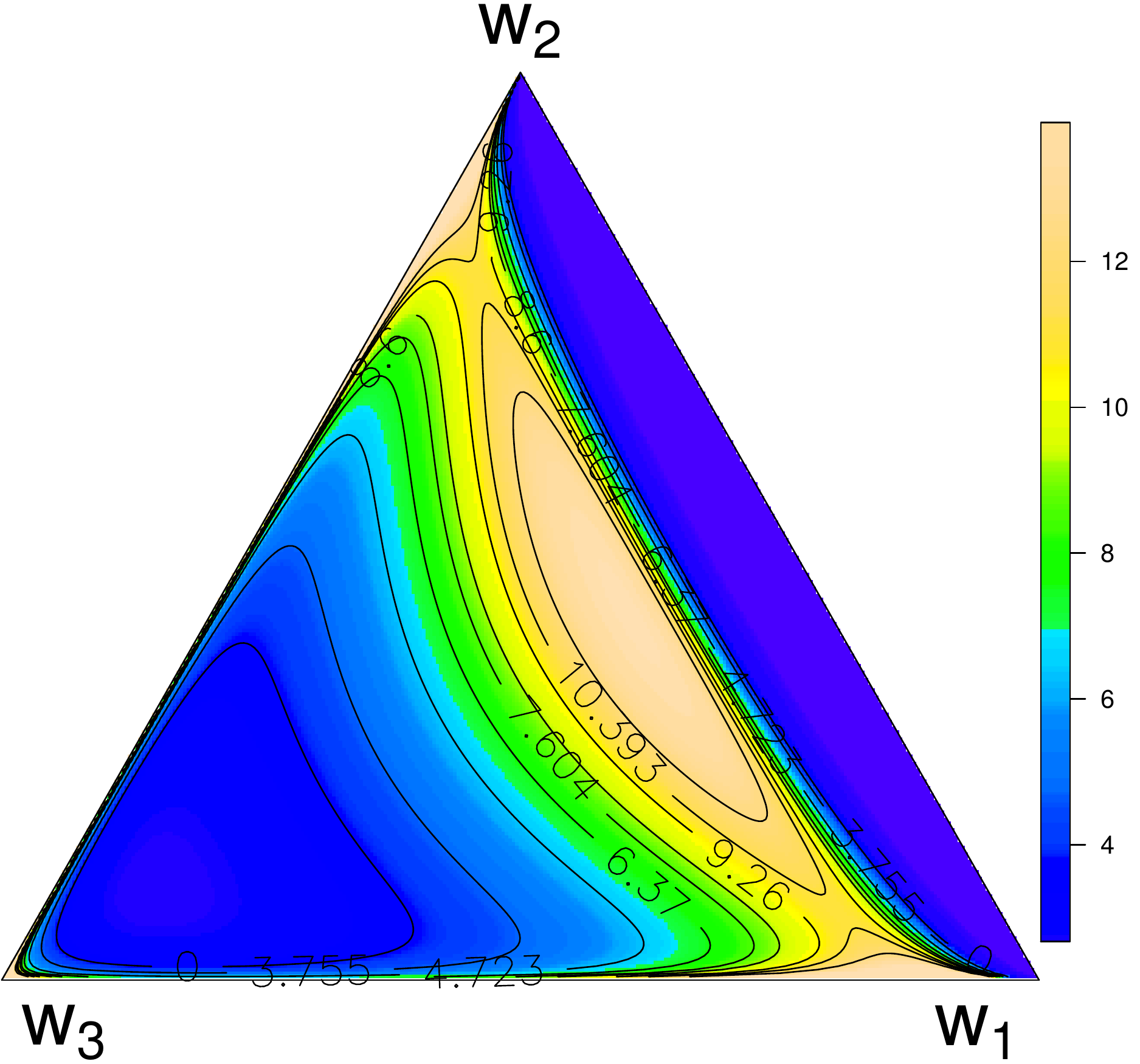}
\includegraphics[scale=0.25]{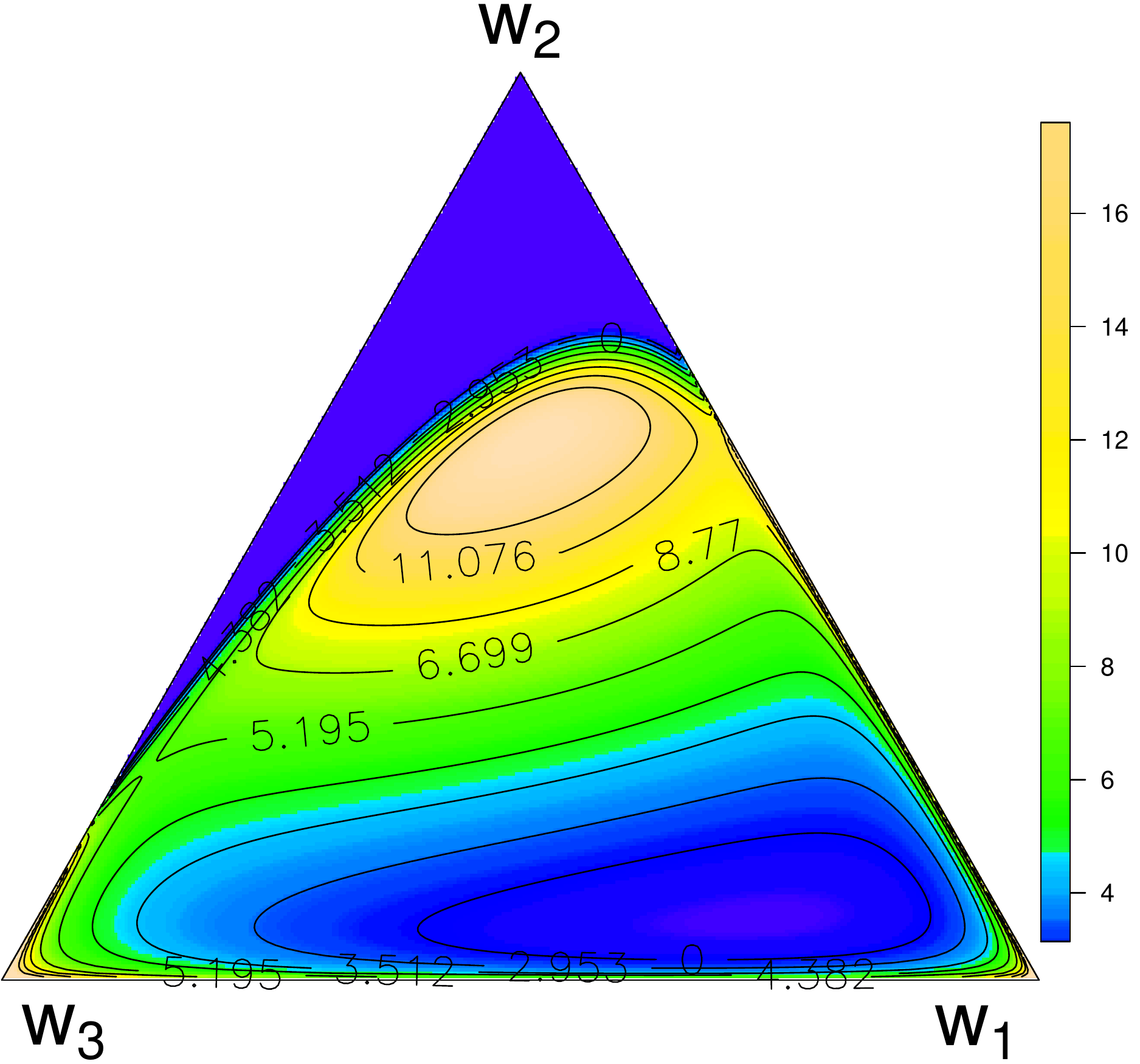}}
\caption{\label{fig:specden_plot}\small Trivariate extremal skew-$t$ angular densities with $\nu =3$ degrees of freedom. 
Correlation coefficients are $\omega = ( 0.6, 0.8, 0.7 )^\top$ for the top row and 
$\omega = ( 0.7, 0.7, 0.7 )^\top$ for the bottom. 
From left to right the skewness parameters are $\alpha = (0,0,0)^\top$, $\alpha = ( -3,-3,7 )^\top$ and $\alpha = ( 7,-10,3 )^\top$.
In all cases $\tau=0$ for simplicity.}
\end{figure}

Figure \ref{fig:specden_plot} illustrates some examples of the flexibility of the trivariate extremal-skew-$t$ 
 dependence structure. 
Here we write the correlation coefficients as $\omega=(\omega_{1,2},\omega_{1,3}, \omega_{2,3})^\top$ and the slant parameters as $\alpha=(\alpha_{1,2},\alpha_{1,3}, \alpha_{2,3})^\top$, and assume that $\nu=3$ and $\tau=0$ for simplicity.

The plots in the left column have $\alpha=(0,0,0)^\top$ and so correspond to the extremal-$t$ angular measure.
The density in the top-left panel, obtained with $\omega=(0.6, 0.8, 0.7)^\top$,
has mass concentrations mainly on the edge that links the first and the third variable, 
since they are the most dependent ($w_{1,3}=0.8$). 
Some mass is also placed on the corners of the second variable, indicating that this is less
dependent on the others ($w_{1,2}=0.6$ and $w_{2,3}=0.7$), and on the middle of the simplex, because a low degree of freedom ($\nu=3$)
pushes mass towards the centre of the simplex. 
The top-middle and top-right panels are extremal skew-$t$ angular densities obtained with $\alpha=(-3,-3,7)^\top$ and $\alpha=(7,-10,3)^\top$
respectively. Here the impact of the slant parameters is to increase the levels of dependence -- indeed the mass is clearly  pushed
towards the centre of the simplex. In the middle panel dependence between the second and third variables has increased, while
in the right panel all variables are strongly dependent with a greater dependence of the second variable on the others.

The bottom row in Figure \ref{fig:specden_plot} illustrates the spectral densities with correlation coefficients
$\omega=(0.7, 0.7, 0.7)^\top$. The bottom-left panel is the standard extremal-$t$ dependence (with $\alpha=(0,0,0)^\top$), which has a symmetric density with mass concentrated mainly in the centre of the simplex and on the vertices. 
The bottom-middle and bottom-right panels show extremal skew-$t$ densities, obtained with $\alpha=(-3,-3,7)^\top$ and $\alpha=(7,-10,3)^\top$ respectively. In this case
the impact of the slant parameters is to decrease the dependence -- here the mass is pushed
towards the edges of the simplex. 
In the middle panel the first and second variables have become  less dependent from the third variable, more so than between each other. In the right panel the first and third variables are less dependent on the second.
These examples illustrate the great flexibility of the extremal skew-$t$ model in capturing a wide range of extremal dependence behaviour above and beyond that of the standard extremal $t$ model.

%
%
%
\subsection{Display of the partitions of the three-dimensional simplex}\label{sec:simulation}

Figure \ref{fig:simplex_3d} displays 
the 
partitions of the three-dimensional simplex into three vertices (grey shading), edges (line shading) and the interior (no shading).
Observations where angular components fall into such areas are considered to belong 
to the corresponding subset of the simplex (vertex, edge or interior).
\begin{figure}[b!]
\centering $
\begin{array}{cc}
\includegraphics[width=0.5\textwidth]{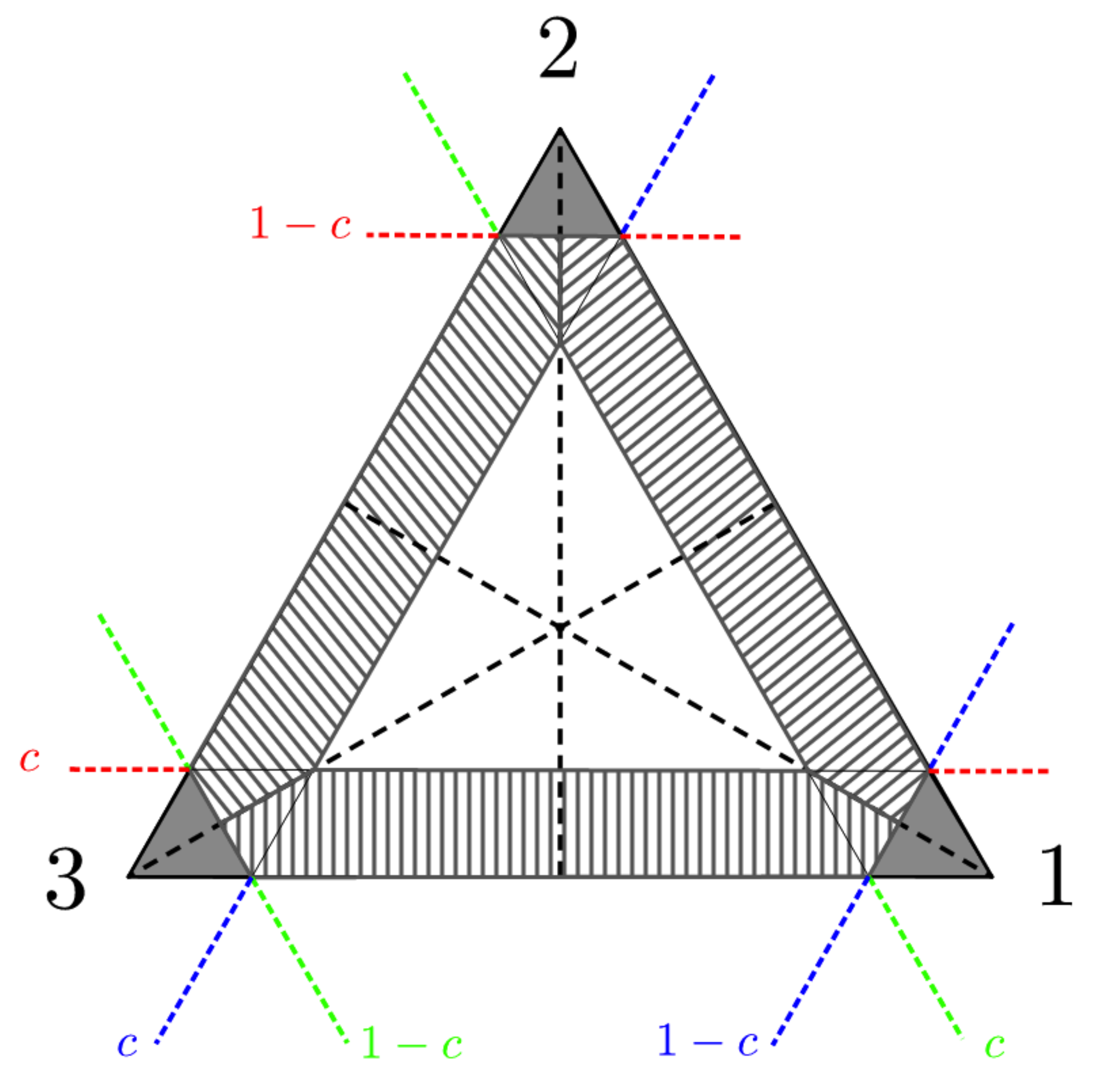}
\end{array} $
\caption{\label{fig:simplex_3d}\small Partitions of the three-dimensional simplex}
\end{figure}
For example, when $w_3>1-c$ (on the left of the green dashed line indicating the $1-c$ level for $w_3$), then $w=(w_1,w_2,w_3)$ is in the corner associated with the third component, which corresponds to the grey shaded triangle on the bottom left of the simplex.
Similarly, if both $w_1$ and $w_2$ are less than $1-c$ (i.e. to the left of the blue dashed line indicating the $1-c$ level of $w_1$ and below the red dashed line indicating the $1-c$ level of $w_2$), such that $w_1 > 1- 2 w_2$ and $w_2 > 1- 2 w_1$ (i.e. to the right of the black dashed line bisecting the corner of the second component and above the black dashed line bisecting the corner of the first component) and if $w_3<c$ (to the right of the green dashed line indicating the $1-c$ level of $w_3$) then $w=(w_1,w_2,w_3)$ is on the edge between the first and second component. This is indicated by the line-shaded area on the right hand side of the simplex.
Finally if $w_1,w_2,w_3 >c$ (i.e. to the right of the blue dashed line, above the red dashed line and to the left of the green dashed line, respectively indicating the $c$ levels of $w_1, w_2$ and $w_3$) then $w=(w_1,w_2,w_3)$ is in the interior of the simplex, represented by the white triangle in the centre of the simplex.

%
%
%
\subsection{Computation of $d$-dimensional extremal-skew-$t$ density for $d=2,3,4$.}
\label{sec:ap_mar_densities}

For clarity of exposition we focus on the finite dimensional distribution of the
extremal-$t$ process, denoted by $G$.
We initially assume that $\alpha=0$ and 
$\tau=0$ in (15) of the main paper (focusing on (16)), and relax this assumption later. 
For brevity the exponent function is written as
$$
V(x_j,j\in I) = \sum_{j\in I} x_j^{-1} T_j,\quad 
T_j = \Psi_{d-1} \left( u_j ; \bar{\Omega}^\circ_{j},\nu+1 \right) 
$$
where $I=\{1,\ldots,d\}$,
$
u_j = \left[
\sqrt{\frac{\nu+1}{1-\omega_{i,j}^2}} \left\{ \left(\frac{x_i}{x_j}\right)^{1/\nu} -\omega_{i,j} \right\},
i\in I_j \right]^\top
$ and where $I_j=I\backslash\{j\}$.
By successive differentiations the $2$-dimensional density ($d=2$) is
$$
f(x)=(-V_{12}+V_1V_2)G(x),\quad x\in\real_+^2,
$$
the $3$-dimensional density  ($d=3$) is 
\begin{align*}\label{eq:full_dens_ext}
f(x) = (-V_{123} + V_{1}V_{23} + V_{2}V_{13} + V_{3}V_{12} - V_{1}V_{2}V_{3} )G(x), \quad x\in\real_+^3
\end{align*}
and the $4$-dimensional density ($d=4$) is
\begin{align*}
f(x) &=  ( -V_{1234} + V_{1}V_{234} + V_{2}V_{134} + V_{3}V_{124} + V_{12}V_{34} + V_{13}V_{24} + V_{14}V_{23} \\
& \quad - V_{1}V_{2}V_{34} - V_{1}V_{3}V_{24} - V_{1}V_{4}V_{23} - V_{2}V_{3}V_{14} - V_{2}V_{4}V_{13} - V_{3}V_{4}V_{12} \\
& \quad + V_{1}V_{2}V_{3}V_{4} ) G(x),\quad x\in\real_+^4
\end{align*}
where $V_{i_1,\ldots,i_m} := \frac{\der^m V(x_j,j\in I) }{\der x_{i_1} \cdots \der x_{i_m} }$
for $i_k \in I$. The derivatives of the exponent function are given by
\begin{equation}\label{eq:deriv_exponent}
V_{i_1,\ldots,i_m} 
= \sum_{k=1}^d x_{i_k}^{-1} \frac{\der^m T_{i_k}}{\der x_{i_1} \cdots \der x_{i_m}}
- \sum_{\ell=1}^m x_{i_\ell}^{-2} \frac{\der^{m-1} T_{i_\ell}}
{\der x_{i_1} \cdots \der x_{i_{\ell-1}} \der x_{i_{\ell+1}} \cdots \der x_{i_m}}.
\end{equation}
In particular, when $m=d$ it follows that $\{ i_1, \ldots, i_d \} = \{ 1,\ldots,d \}$ and that
\begin{equation*} 
V_{1 \cdots d} = - 
(\nu x_1)^{-(d+1)}\psi_{d-1}
\left( 
u_1;\bar{\Omega}^{\circ}_1,\nu+1
\right)
\prod_{i=2}^d \sqrt{ \frac{\nu+1}{1-\omega_{i,1}^2 } }
 \left( \frac{x_i}{x_1}\right)^{\frac{1}{\nu}-1}.
\end{equation*}
When $d=2$ or $3$, the derivatives of $T_j$, for $j \in I$ are given by
\begin{align}
\frac{\der T_j}{\der x_{i_1}} &= \sum_{p=1}^{d-1} \frac{\der}{\der u_{p,j}} 
\Psi_{d-1} \left( u_j ; \bar{\Omega}^\circ_{j},\nu+1 \right) \frac{\der u_{p,j}}{\der x_{i_1} }, \label{eq:first_deriv_I}\\
\frac{\der^2 T_j}{\der x_{i_1} \der x_{i_2}} &= 
\sum_{p=1}^{d-1} \left( 
\frac{\der}{\der u_{p,j}} \Psi_{d-1} \left( u_j ; \bar{\Omega}^\circ_{j},\nu+1 \right) 
\frac{\der^2 u_{p,j}}{\der x_{i_1}\der x_{i_2} }
+ \frac{\der^2}{\der u_{p,j}^2} \Psi_{d-1} \left( u_j ; \bar{\Omega}^\circ_{j},\nu+1 \right) 
\frac{\der u_{p,j}}{\der x_{i_1} } \frac{\der u_{p,j}}{\der x_{i_2} }
\right) \nonumber \\
& + \sum_{p=1}^{d-2} \sum_{q=p+1}^{d-1} \frac{\der^2}{\der u_{p,j}\der u_{q,j}}
\Psi_{d-1} \left( u_j ; \bar{\Omega}^\circ_{j},\nu+1 \right)
\left[ \frac{\der u_{p,j}}{\der x_{i_1} } \frac{\der u_{q,j}}{\der x_{i_2} } 
+ \frac{\der u_{p,j}}{\der x_{i_2} } \frac{\der u_{q,j}}{\der x_{i_1} } \right] \label{eq:second_deriv_I},
\end{align}
where $u_{p,j}$ is the $p$-th element of $u_j$, and when $d=3$
\begin{align}\label{eq:third_deriv_I}
\frac{\der^3 T_j}{\der x_{i_1} \der x_{i_2} \der x_{i_3}} &=
\sum_{p=1}^2 \sum_{q=2}^3 \left( \frac{\der^2}{\der u_{p,j}\der u_{q,j}}
\Psi_{d-1} \left( u_j ; \bar{\Omega}^\circ_{j},\nu+1 \right) 
\sum_{\substack{r,s,t \in I \\ r \neq s \neq t}} \frac{\der u_{p,j}}{\der x_{i_r} } \frac{\der^2 u_{q,j}}{\der x_{i_s} \der x_{i_t} } 
+ \frac{\der u_{q,j}}{\der x_{i_r} } \frac{\der^2 u_{p,j}}{\der x_{i_s} \der x_{i_t} }
\right) \nonumber\\
& + \sum_{p=1}^3 \sum_{\substack{q=1 \\q \neq p}}^3 \frac{\der^3}{\der u_{p,j}^2 \der u_{q,j}}
\Psi_{d-1} \left( u_j ; \bar{\Omega}^\circ_{j},\nu+1 \right) 
\sum_{\substack{r,s,t \in I \\ r \neq s \neq t}} 
\frac{\der u_{p,j}}{\der x_{i_r} } \frac{\der u_{p,j}}{\der x_{i_s} } \frac{\der u_{q,j}}{\der x_{i_t} } \nonumber\\
& + \frac{\der^3}{\der u_{1,j} \der u_{2,j} \der u_{3,j}}
\Psi_{d-1} \left( u_j ; \bar{\Omega}^\circ_{j},\nu+1 \right) 
\sum_{\substack{r,s,t \in I \\ r \neq s \neq t}} 
\frac{\der u_{1,j}}{\der x_{i_r} } \frac{\der u_{2,j}}{\der x_{i_s} } \frac{\der u_{3,j}}{\der x_{i_t} }.
\end{align}
We provide the derivatives of the $d$-dimensional $t$ cdf below.
%
When $d=1$ and for all $x\in\real_+$
\begin{align*}
\frac{\der}{\der x} \Psi(x;\nu)&=\psi(x;\nu), \quad
\frac{\der^2}{\der x^2} \Psi(x;\nu)=-\frac{(\nu+1)x}{\nu+x^2}\psi(x;\nu), \\
\frac{\der^3}{\der x^3} \Psi(x;\nu)&=\frac{(\nu+1)(x^2-\nu+ (\nu+1)x^2)}{(\nu+x^2)^2}\psi(x;\nu).
\end{align*}
%
When $d=2$ and for all $x\in \real_+^2$,
\begin{align*}
\frac{\der}{\der x_1} \Psi_2 ( x; \bar{\Omega}, \nu ) 
&= \psi(x_1; \nu) \Psi\left( v_{2\cdot 1};\nu +1 \right),\\
\frac{\der^2}{\der x_1^2} \Psi_2 (x; \bar{\Omega}, \nu )
&= -\psi(x_1;\nu) 
\left\{ 
\frac{(\nu +1)x_1}{\nu +x_1^2 } 
\Psi \left( v_{2\cdot 1};\nu +1 \right)+\sqrt{\frac{\nu +1}{1 - \omega^2}} 
\left( \frac{\omega \nu  + x_2 x_1}{(\nu +x_1^2)^{3/2}} \right)
\psi \left( v_{2\cdot 1};\nu +1 \right)
\right\},\\
\frac{\der^2}{\der x_1 \der x_2} \Psi_2 (x; \bar{\Omega}, \nu )
&= \psi_2 (x; \bar{\Omega}, \nu),
\end{align*}
where 
$
v_{i\cdot j}=\sqrt{\frac{\nu +1}{\nu +x_j^2}} \frac{x_i - \omega_{i,j} x_1}{\sqrt{1-\omega_{i,j}^2}},\quad j\in I, j\in I_j,
$
%
{\small
\begin{align*}
\frac{\der^3}{\der x_1^3} \Psi_2(x; \bar{\Omega}, \nu )
&= 
\Psi\left( v_{2\cdot 1};\nu +1 \right)
\psi(x_1;\nu) \left\{ \frac{(\nu +1)^2 x_1^2 - (\nu +1) (\nu - x_1^2) }{(\nu +x_1^2)^2} \right\} \\
&+ \psi \left( v_{2\cdot 1}; \nu +1 \right)
\psi(x_1;\nu) \sqrt{\frac{\nu +1}{1 - \omega^2}} \frac{1}{(\nu +x_1^2)^{5/2}} \\
&\times \left\{ x_1 (\omega \nu + x_2 x_1) ( 2\nu -1) - x_2 ( \nu +x_1^2 ) \right.\\
&- \left.\frac{\big ( \omega ( \nu + x_1^2 ) +(x_2 - \omega x_1)x_1 \big ) (\nu + 2) ( x_2 - \omega x_1) ( \omega \nu + x_2 x_1 ) }
{(\nu + x_1^2) (1- \omega^2) + (x_2 - \omega x_1)^2 }
\right\},\\
\frac{\der^3}{\der x_1^2 \der x_2} \Psi_2 (x; \bar{\Omega}, \nu )
&= - \frac{(\nu +2) ( x_1 - \omega x_2 )}{2 \pi \nu (1-\omega^2)^{3/2}}
\bigg( 1 + \frac{x_1^2 - 2 \omega x_1 x_2 + x_2^2}{\nu ( 1- \omega^2 )} \bigg)^{-(\frac{\nu}{2}+1)}.
\end{align*}
}
When $d=3$ and for all $x\in \real_+^3$,
\begin{align*}
\frac{\der}{\der x_1} \Psi_3 (x; \bar{\Omega}, \nu) &= \psi(x; \nu)
\Psi_2 \left\{ (v_{2\cdot 1}, v_{3\cdot 1})^\top;\bar{\Omega}^{\circ}_1, \nu +1 \right\},
\end{align*}
%
{\small
\begin{align*}
\frac{\der^2}{\der x_1^2} &\Psi_3 (x; \bar{\Omega}, \nu) = 
\frac{-\psi(x_1; \nu)}{\nu +x_1^2} \left[ (\nu +1)x_1
\times \Psi_2 \left\{ (v_{2\cdot 1}, v_{3\cdot 1})^\top;\bar{\Omega}^{\circ}_1, \nu +1 \right\}\right.\\
&  
+ \psi\left( v_{2\cdot 1}; 
\nu +1 \right) \sqrt{\frac{\nu +1}{1-\omega_{12}^2}} \frac{ x_2x_1 +\omega_{12}\nu}{\sqrt{\nu +x_1^2}}\\
& 
\times \Psi \left(
\frac{\sqrt{\nu+2} \left\{ (x_3-\omega_{13}x_1)(1-\omega_{12}^2)-(\omega_{23}-\omega_{12}\omega_{13})(x_2-\omega_{12}x_1) \right\}}
{\sqrt{\left\{ (1-\omega_{12}^2)(\nu +x_1^2)+(x_2-\omega_{12}x_1)^2 \right\}
\left\{ (1-\omega_{12}^2)(1-\omega_{13}^2)-(\omega_{23}-\omega_{12}\omega_{13})^2 \right\}}}
;\nu +2\right) \\
& 
+ \psi \left( v_{3\cdot 1}; 
\nu +1 \right) \sqrt{\frac{\nu +1}{1-\omega_{13}^2}} \frac{ x_3x_1 +\omega_{13}\nu}{\sqrt{\nu +x_1^2}} \\
& 
\times \left.\Psi \left(
\frac{\sqrt{\nu+2} \left\{ (x_2-\omega_{12}x_1)(1-\omega_{13}^2)-(\omega_{23}-\omega_{12}\omega_{13})(x_3-\omega_{13}x_1) \right\}}
{\sqrt{\left\{ (1-\omega_{13}^2)(\nu +x_1^2)+(x_3-\omega_{13}x_1)^2 \right\}
\left\{ (1-\omega_{12}^2)(1-\omega_{13}^2)-(\omega_{23}-\omega_{12}\omega_{13})^2 \right\}}}
;\nu +2\right) \right]
\end{align*}
}
{\small
\begin{align*}
\frac{\der^2}{\der x_1 \der x_2} &\Psi_3 (x; \bar{\Omega}, \nu) = 
\psi(x_2; \nu) \psi\left( v_{1\cdot 2} ; \nu +1 \right)
\sqrt{\frac{\nu+1}{(1-\omega_{12}^2)(\nu +x_2^2)}} \\
& 
\times \Psi \left(
\frac{\sqrt{\nu+2} \left\{ (x_3-\omega_{23}x_2)(1-\omega_{12}^2)-(\omega_{13}-\omega_{12}\omega_{23})(x_1-\omega_{12}x_2) \right\}}
{\sqrt{\left\{ (1-\omega_{12}^2)(\nu +x_1^2)+(x_1-\omega_{12}x_2)^2 \right\}
\left\{ (1-\omega_{12}^2)(1-\omega_{23}^2)-(\omega_{13}-\omega_{12}\omega_{23})^2 \right\}}}
;\nu +2\right)
\end{align*}
}
{\small
\begin{align*}
\frac{\der^3}{\der x_1^2 \der x_2} &\Psi_3 (x; \bar{\Omega}, \nu) = 
- \psi(x_3; \nu ) 
\psi \left( v_{1\cdot 3} ;\nu +1 \right)
\sqrt{\frac{\nu+1}{( 1-\omega_{13}^2 )( \nu +x_3^2 )}}
\left[ \frac{(\nu +2)(x_1-\omega_{12}x_2)}{(1-\omega_{12}^2)(\nu +x_2^2)+(x_1-\omega_{12}x_2)^2} \right.\\
& 
\times \Psi \left(
\frac{\sqrt{\nu+2} \left\{ (x_3-\omega_{23}x_2)(1-\omega_{12}^2)-(\omega_{13}-\omega_{12}\omega_{23})(x_1-\omega_{12}x_2) \right\}}
{\sqrt{\left\{ (1-\omega_{12}^2)(\nu +x_1^2)+(x-\omega_{12}x_2)^2 \right\}
\left\{ (1-\omega_{12}^2)(1-\omega_{23}^2)-(\omega_{13}-\omega_{12}\omega_{23})^2 \right\}}}
;\nu +2\right) \\
& 
+ \frac{\sqrt{\nu+2} ( 1-\omega_{12}^2 )}
{\sqrt{(1-\omega_{12}^2)(1-\omega_{23}^2)-(\omega_{13}-\omega_{12}\omega_{23})^2}}
\frac{(\omega_{13}-\omega_{12}\omega_{23})-(x_1-\omega_{12}x_2)(x_3-\omega_{23}x_2)}
{\left\{ (1-\omega_{12}^2)(\nu +x_2^2) + (x_1-\omega_{12}x_2)^2\right\}^{3/2}} \\
& 
\times \left.\psi \left(
\frac{\sqrt{\nu+2} \left\{ (x_3-\omega_{23}x_2)(1-\omega_{12}^2)-(\omega_{13}-\omega_{12}\omega_{23})(x_1-\omega_{12}x_2) \right\}}
{\sqrt{\left\{ (1-\omega_{12}^2)(\nu +x_1^2)+(x_1-\omega_{12}x_2)^2 \right\}
\left\{ (1-\omega_{12}^2)(1-\omega_{23}^2)-(\omega_{13}-\omega_{12}\omega_{23})^2 \right\}}}
;\nu +2\right) \right]
\end{align*}
}
{\small
\begin{align*}
\frac{\der^3}{\der x_1^3} &\Psi_3 (x; \bar{\Omega}, \nu) = 
-\frac{\psi(x_1;\nu)}{(\nu+x_1^2)}\left[
\left( \frac{\nu+3}{\nu +x_1^2} \right)(1-x_1^2) (\nu+1)
\Psi_2 \left\{ (v_{2\cdot 1}, v_{3\cdot 1})^\top;\bar{\Omega}^{\circ}_1, \nu +1 \right\}\right.\\
& 
+ \Psi \left(
\frac{\sqrt{\nu+2} \left[ (x_3-\omega_{13}x_1)(1-\omega_{12}^2)-(\omega_{23}-\omega_{12}\omega_{13})(x_2-\omega_{12}x_1) \right]}
{\sqrt{\left[ (1-\omega_{12}^2)(\nu +x_1^2)+(x_2-\omega_{12}x_1)^2 \right]
\left[ (1-\omega_{12}^2)(1-\omega_{13}^2)-(\omega_{23}-\omega_{12}\omega_{13})^2 \right]}}
;\nu +2\right) \\
&  
\times \psi \left( v_{2\cdot 1}; 
\nu +1 \right)
\sqrt{\frac{\nu+1}{1-\omega_{12}^2}} \frac{2(x_2x_1+\omega_{12}\nu)(\nu+2)x_1-\nu(x_2-\omega_{12}x_1)}{(\nu+x_1^2)^{3/2}} \\
&
\times \frac{(\nu+2) (x_2-\omega_{12}x_1) \sqrt{\nu+1} (x_2x_1 +\omega_{12}\nu)^2 }
{ \sqrt{1-\omega_{12}^2} (\nu+x_1^2)^{3/2} \left((1-\omega_{12}^2)(\nu + x_1^2) + (x_2-\omega_{12}x_1)^2\right)} \\
& 
+ \Psi \left(
\frac{\sqrt{\nu+2} \left[ (x_2-\omega_{12}x_1)(1-\omega_{13}^2)-(\omega_{23}-\omega_{12}\omega_{13})(x_3-\omega_{13}x_1) \right]}
{\sqrt{\left[ (1-\omega_{13}^2)(\nu +x_1^2)+(x_3-\omega_{13}x_1)^2 \right]
\left[ (1-\omega_{12}^2)(1-\omega_{13}^2)-(\omega_{23}-\omega_{12}\omega_{13})^2 \right]}}
;\nu +2\right) \\
&  
\times \psi \left( v_{3\cdot 1}; 
\nu +1 \right)
\sqrt{\frac{\nu+1}{1-\omega_{13}^2}} \frac{2(x_3x_1+\omega_{13}\nu)(\nu+2)x_1-\nu(x_3-\omega_{13}x_1)}{(\nu+x_1^2)^{3/2}} \\
&
\times \frac{(\nu+2) (x_3-\omega_{13}x_1) \sqrt{\nu+1} (x_3x_1 +\omega_{13}\nu)^2 }
{ \sqrt{1-\omega_{13}^2} (\nu+x_1^2)^{3/2} \left((1-\omega_{13}^2)(\nu + x_1^2) + (x_3-\omega_{13}x_1)^2\right)} \\
& 
+ \psi \left(
\frac{\sqrt{\nu+2} \left[ (x_3-\omega_{13}x_1)(1-\omega_{12}^2)-(\omega_{23}-\omega_{12}\omega_{13})(x_2-\omega_{12}x_1) \right]}
{\sqrt{\left[ (1-\omega_{12}^2)(\nu +x_1^2)+(x_2-\omega_{12}x_1)^2 \right]
\left[ (1-\omega_{12}^2)(1-\omega_{13}^2)-(\omega_{23}-\omega_{12}\omega_{13})^2 \right]}}
;\nu +2\right) \\
&  
\times \psi \left( v_{2\cdot 1}; 
\nu +1 \right) 
\sqrt{\frac{(1-\omega_{12}^2)(\nu+2)}{(1-\omega_{12}^2)(1-\omega_{13}^2)-(\omega_{23}-\omega_{12}\omega_{13})^2}} \\
&
\times \frac{\sqrt{\nu+1}(x_2x_1 + \omega_{12}\nu)}
{\sqrt{\nu+x_1^2}((1-\omega_{12}^2)(\nu+x_1^2)+(x_2-\omega_{12}x_1)^2)^{3/2}}
\bigg[ ((1-\omega_{12}^2)(\nu+x_1^2)+(x_2-\omega_{12}x_1)^2) \\
&
\times \left(\omega_{12}\frac{\omega_{23}-\omega_{12}\omega_{13}}{1-\omega_{12}^2}-\omega_{13} \right)
- \left((x_3-\omega_{13}x_1)-\frac{\omega_{23}-\omega_{12}\omega_{13}}{1-\omega_{12}^2}(x_2-\omega_{12}x_1)\right)
(x_1-\omega_{12}x_2)\bigg] \\
& 
+ \psi \left(
\frac{\sqrt{\nu+2} \left[ (x_2-\omega_{12}x_1)(1-\omega_{13}^2)-(\omega_{23}-\omega_{12}\omega_{13})(x_3-\omega_{13}x_1) \right]}
{\sqrt{\left[ (1-\omega_{13}^2)(\nu +x_1^2)+(x_3-\omega_{13}x_1)^2 \right]
\left[ (1-\omega_{12}^2)(1-\omega_{13}^2)-(\omega_{23}-\omega_{12}\omega_{13})^2 \right]}}
;\nu +2\right) \\
&  
\times \psi \left( v_{3\cdot 1}; 
\nu +1 \right) 
\sqrt{\frac{(1-\omega_{13}^2)(\nu+2)}{(1-\omega_{12}^2)(1-\omega_{13}^2)-(\omega_{23}-\omega_{12}\omega_{13})^2}} \\
&
\times \frac{\sqrt{\nu+1}(x_3x_1 + \omega_{13}\nu)}
{\sqrt{\nu+x_1^2}((1-\omega_{13}^2)(\nu+x_1^2)+(x_3-\omega_{13}x_1)^2)^{3/2}}
\bigg[ ((1-\omega_{13}^2)(\nu+x_1^2)+(x_3-\omega_{13}x_1)^2) \\
&
\times \left(\omega_{13}\frac{\omega_{23}-\omega_{12}\omega_{13}}{1-\omega_{13}^2}-\omega_{12} \right)
- \left((x_2-\omega_{12}x_1)-\frac{\omega_{23}-\omega_{12}\omega_{13}}{1-\omega_{13}^2}(x_3-\omega_{13}x)\right)
(x_1-\omega_{13}x_3)\bigg].
\end{align*}
}
%
Combining the derivatives of the $t$ cdf with equations \eqref{eq:deriv_exponent}--\eqref{eq:third_deriv_I} provides the full  $d$-dimensional densities of the extremal-$t$ process.
Returning to the extremal skew-$t$ case (i.e. when $\alpha\neq 0$ and $\tau\neq 0$), 
it is sufficient to consider the following changes. Firstly, rewrite
$$
T_j = \frac{\Psi_{d} \left\{
\left( \begin{array}{c} u_j \\ \bar{\tau}_j \end{array}\right);
\left( \begin{array}{cc} \bar{\Omega}^\circ_{j} & -\delta_j \\ -\delta_j^\top & 1 \end{array} \right), \nu +1
\right\} }
{\Psi_1\left(\bar{\tau}_j; \nu+1 \right)},\quad j\in I,
$$
where
$
u_j = \left[
\sqrt{\frac{\nu+1}{1-\omega_{i,j}^2}} \left\{ \left( \frac{x_i^{\circ}}{x_j^{\circ}} \right)^{1/\nu} -\omega_{i,j} \right\},
i\in I_j \right]^\top,
$
following Definition 1 of the main paper.
It can then be shown that
{\small
\begin{equation*} \label{eq:density_extremal_extend_st}
V_{1\cdots d} = -
(\nu x_1)^{-(d+1)} \psi_{d-1}(u_1; 
\bar{\Omega}^{\circ}_1, \alpha_1^\circ, \tau_1^\circ, \kappa_1^\circ, \nu+1)
\prod_{i=2}^d \sqrt{ \frac{\nu+1}{1-\omega_{i,1}^2 } }
 \left( \frac{x^{\circ}_i}{x^{\circ}_1}\right)^{\frac{1}{\nu}-1} \frac{m_i^+}{m_1^+}
\end{equation*}}
%
%
following Theorem 1 of the main paper.
Note that equations \eqref{eq:deriv_exponent}--\eqref{eq:third_deriv_I} are still valid in this case, through the redefinition of $d\leftarrow d+1$ and $u_j \leftarrow ( u_j, \bar{\tau}_j )^\top$. 
This in combination  with the above derivatives of the $t$ cdfs leads to the 
$d$-dimensional densities of the extremal-skew-$t$ process.

%
%
%
%
\begin{table}[b!]
\begin{center}
\begin{tabular}{ c c c c c c }
\hline
\multicolumn{6}{c}{$\nu=1$}\\
\multicolumn{6}{l}{$n=20$}\\
$\lambda \backslash \xi$ & $0.5$ & $1$ & $1.5$ & $1.9$ & $2$ \\ 
$14$ & $89/94/89$ & $84/97/93$ & $83/69/79$ & $81/82/84$ & $78/64/72$ \\
$28$ & $76/100/98$ & $59/100/69$ & $73/86/73$ & $74/66/75$ & $34/75/26$ \\
$42$ & $81/100/100$ & $51/96/89$ & $51/80/88$ & $43/63/79$ & $33/51/72$ \\
\hline
\multicolumn{6}{l}{$n=50$}\\
$\lambda \backslash \xi$ & $0.5$ & $1$ & $1.5$ & $1.9$ & $2$ \\ 
$14$ & $85/81/84$ & $87/78/86$ & $76/67/78$ & $66/56/72$ & $52/47/62$ \\
$28$ & $64/100/81$ & $81/79/82$ & $73/72/78$ & $72/66/74$ & $34/68/24$ \\
$42$ & $71/100/97$ & $33/61/59$ & $17/42/40$ & $17/34/37$ & $2/18/7$ \\
\hline
\multicolumn{6}{l}{$n=70$}\\
$\lambda \backslash \xi$ & $0.5$ & $1$ & $1.5$ & $1.9$ & $2$ \\ 
$14$ & $80/87/83$ & $81/76/80$ & $74/65/77$ & $62/57/70$ & $47/42/60$ \\
$28$ & $51/100/68$ & $82/82/84$ & $72/72/77$ & $71/66/73$ & $54/53/62$ \\
$42$ & $56/93/89$ & $28/52/48$ & $13/40/14$ & $12/28/27$ & $8/23/26$ \\
\hline\hline
\multicolumn{6}{c}{$\nu=3$}\\
\multicolumn{6}{l}{$n=20$}\\
$\lambda \backslash \xi$ & $0.5$ & $1$ & $1.5$ & $1.9$ & $2$ \\ 
$14$ & $93/100/96$ & $93/96/91$ & $88/84/83$ & $84/83/84$ & $78/77/82$ \\
$28$ & $86/100/100$ & $72/97/75$ & $90/91/89$ & $87/85/86$ & $39/78/50$ \\
$42$ & $78/100/100$ & $72/97/100$ & $58/71/74$ & $51/68/95$ & $44/58/84$ \\
\hline
\multicolumn{6}{l}{$n=50$}\\
$\lambda \backslash \xi$ & $0.5$ & $1$ & $1.5$ & $1.9$ & $2$ \\ 
$14$ & $91/85/89$ & $92/89/92$ & $86/81/88$ & $82/78/86$ & $64/64/74$ \\
$28$ & $70/100/81$ & $74/87/63$ & $83/81/84$ & $80/74/82$ & $77/75/81$ \\
$42$ & $69/100/100$ & $47/70/75$ & $36/53/64$ & $30/40/61$ & $38/32/33$ \\
\hline
\multicolumn{6}{l}{$n=70$}\\
$\lambda \backslash \xi$ & $0.5$ & $1$ & $1.5$ & $1.9$ & $2$ \\ 
$14$ & $93/93/94$ & $89/88/87$ & $81/77/85$ & $81/74/84$ & $58/58/71$ \\
$28$ & $94/94/94$ & $85/87/89$ & $81/77/86$ & $79/75/82$ & $81/77/84$ \\
$42$ & $65/94/95$ & $44/57/62$ & $29/45/49$ & $25/35/50$ & $20/28/38$ \\
\hline
\end{tabular}
\end{center}
\caption{\small Efficiency of maximum triplewise likelihood estimators relative to maximum pairwise likelihood estimators  for the Extremal-$t$ process, based on $300$ replicate simulations. Simulated datasets of size $n=20, 50, 70$ are generated at  $20$ random sites in $\spa=\left[ 0,100\right]^2$, given
power exponential dependence function parameters $\vartheta=(\lambda,\xi)$. 
Relative efficiencies are $RE_\xi$/$RE_\lambda$/$RE_{(\lambda,\xi)}$ ($\times100$) where $RE_\xi=\widehat{\text{var}}(\hat{\xi}_3 ) / \widehat{\text{var}} ( \hat{\xi}_2 )$,  $RE_\lambda=\widehat{\text{var}}(\hat{\lambda}_3 ) / \widehat{\text{var}} ( \hat{\lambda}_2 )$ and
$RE_{(\lambda,\xi)}=\widehat{\text{cov}}(\hat{\lambda}_3,\hat{\xi}_3 ) / \widehat{\text{cov}} ( \hat{\lambda}_2,\hat{\xi}_2 )$, where $(\hat{\lambda}_m,\hat{\xi}_m)$ are the $m$-wise maximum composite likelihood estimates ($m=2,3$), and $\widehat{\text{var}}$ and $\widehat{\text{cov}}$ denote sample variance and covariance over replicates.
}
\label{table:efficiency_extremalt}
\end{table}
\subsection{Composite likelihood simulation study}\label{sec:simulation}
We compare the efficiency of the maximum triplewise composite likelihood estimator with that based on the pairwise composite likelihood, 
discussed in Section 4 of the main paper, when data are drawn from an extremal-$t$ process.
We generate 300 replicate samples of size $n=20,50$ and $70$ from the extremal-$t$ process with correlation function (10) in Section
 2.2 of the main paper, with varying parameters, over  $20$ random spatial points on $\spa=[0,100]^2$. 
Table \ref{table:efficiency_extremalt} presents the resulting relative efficiencies 
 $RE_\xi$/$RE_\lambda$/$RE_{(\lambda,\xi)}$ ($\times100$), where $RE_\xi=\widehat{\text{var}}(\hat{\xi}_3 ) / \widehat{\text{var}} ( \hat{\xi}_2 )$,  $RE_\lambda=\widehat{\text{var}}(\hat{\lambda}_3 ) / \widehat{\text{var}} ( \hat{\lambda}_2 )$ and
$RE_{(\lambda,\xi)}=\widehat{\text{cov}}(\hat{\lambda}_3,\hat{\xi}_3 ) / \widehat{\text{cov}} ( \hat{\lambda}_2,\hat{\xi}_2 )$, where $(\hat{\lambda}_m,\hat{\xi}_m)$ are the $m$-wise maximum composite likelihood estimates ($m=2,3$), and $\widehat{\text{var}}$ and $\widehat{\text{cov}}$ denote sample variance and covariance over replicates.
Perhaps unsurprisingly, the triplewise estimates are at worst just as efficient as the pairwise estimates ($RE\leq 100$) but are frequently much more efficient. However this is balanced computationally as there is a corresponding increase in the number of components in the triplewise composite likelihood function.
For each $\nu$,  there is a general gain in efficiency when the smoothing parameter $\xi$ increases for each fixed scale parameter $\lambda$. 
There is a similar gain when increasing $\lambda$ for fixed $\xi$.
These gains become progressively pronounced with increasing sample size $n$, and when there is stronger dependence present (i.e. smaller degrees of freedom $\nu$). 
However, we note that there are a number of instances where the efficiency gain goes against this general trend, which indicates that there are some subtleties involved.

%
%
%
\subsection{Marginal analysis of wind speed data}\label{sec:real_data}

The maximum daily observations of wind speed ($1564$ observations per station) are considered for each of the $4$ monitoring stations 
CLOU, CLAY, SALL and PAUL.  
The $t$ and skew-$t$ distributions are fitted to the data using the maximum likelihood approach and a chi-square test is performed in order to investigate wether the slant parameter of the skew-$t$ distribution is significantly different from zero.
Additionally the Fisher-Pearson coefficient of skewness ($\gamma$) is calculated.

\begin{table}[b!]
\begin{center}
\begin{tabular}{@{\extracolsep{4pt}}ccccccccc@{}}
\setlength{\tabcolsep}{0.1pt}
Station & Model & $\hat{\mu}$ & $\hat{\sigma}$ & $\hat{\alpha}$ & $\hat{\nu}$ & $p$-value & $\gamma$ \\
 \hline
CLOU & $t$ & $11.84$ & $2.75$ & $-$ & $5.78$ & $-$ & $-$ \\
 & skew-$t$ & $8.51$ & $20.24$ & $2.79$ & $11.21$ & $0$ & $1.17$ \\
 \hline
CLAY & $t$ & $12.63$ & $3.50$ & $-$ & $6.40$ & $-$ & $-$ \\
 & skew-$t$ & $8.23$ & $35.53$ & $3.28$ & $16.61$ & $0$ & $1.12$  \\
 \hline
SALL & $t$ & $14.66$ & $4.27$ & $-$ & $7.47$ & $-$ & $-$ \\
 & skew-$t$ & $9.02$ & $58.76$ & $4.20$ & $50.98$ & $0$ & $0.92$  \\
 \hline
PAUL & $t$ & $15.76$ & $4.25$ & $-$ & $9.31$ & $-$ & $-$ \\
 & skew-$t$ & $11.43$ & $38.55$ & $1.78$ & $17.81$ & $0$ & $0.79$  \\
 \hline
\end{tabular}
\end{center}
\caption{\small Outcome of the marginal analysis of the four stations.}
\label{tabMargin}
\end{table}

The marginal estimation results are collected in Table \ref{tabMargin}. The estimated parameters  are location $\mu$, scale $\sigma$ and degrees of freedom $\nu$ for the $t$ distribution and in addition the slant $\alpha$ for the 
skew-$t$ distributions. The Table also displays the $p$-value of a chi-square test of $\alpha=0$ for each station. 
With a $p$-value of effectively zero, the marginal skewness of the data is established for each station. 
\begin{figure}[t!]
\centering 
\makebox{
\includegraphics[width=0.4\textwidth]{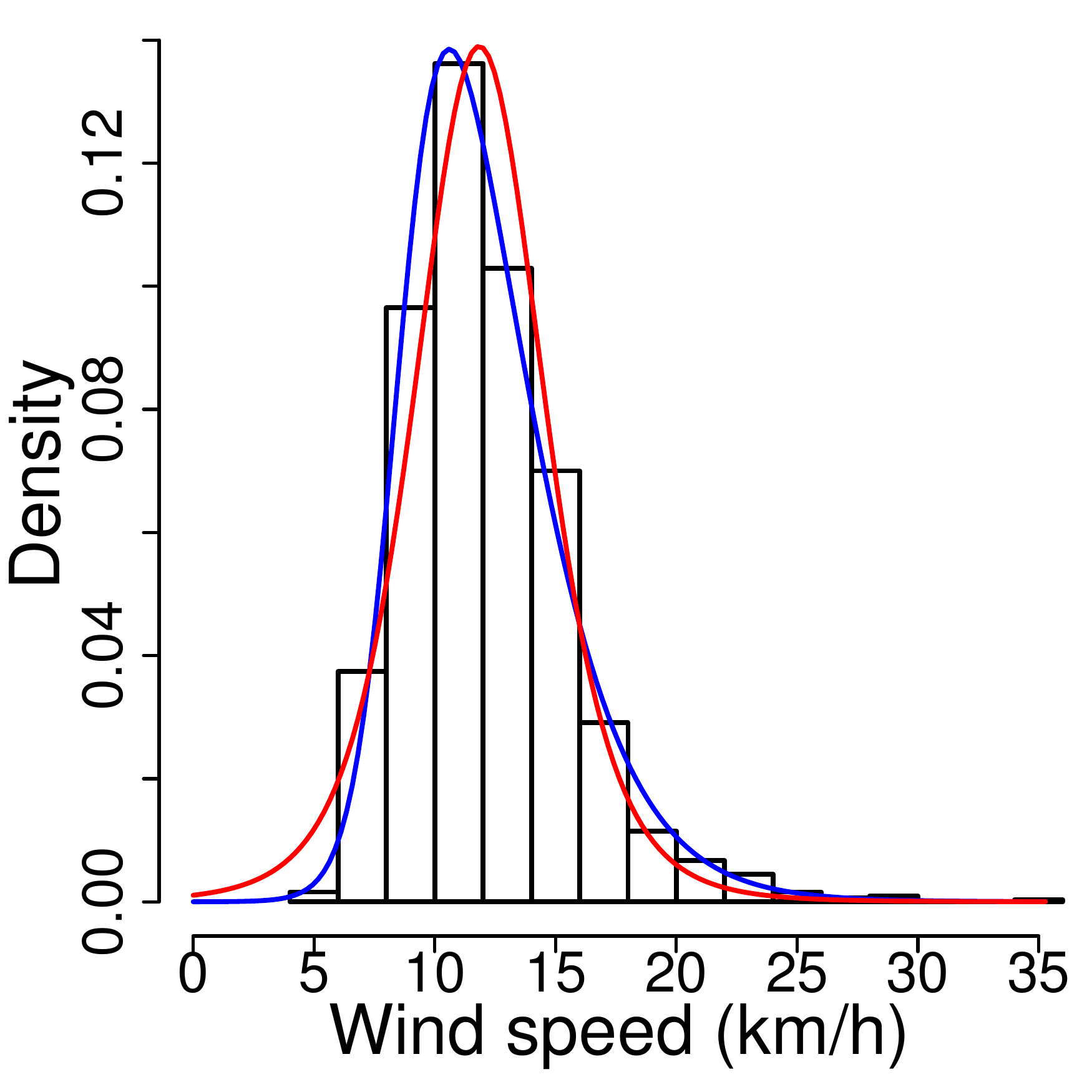} 
\includegraphics[width=0.4\textwidth]{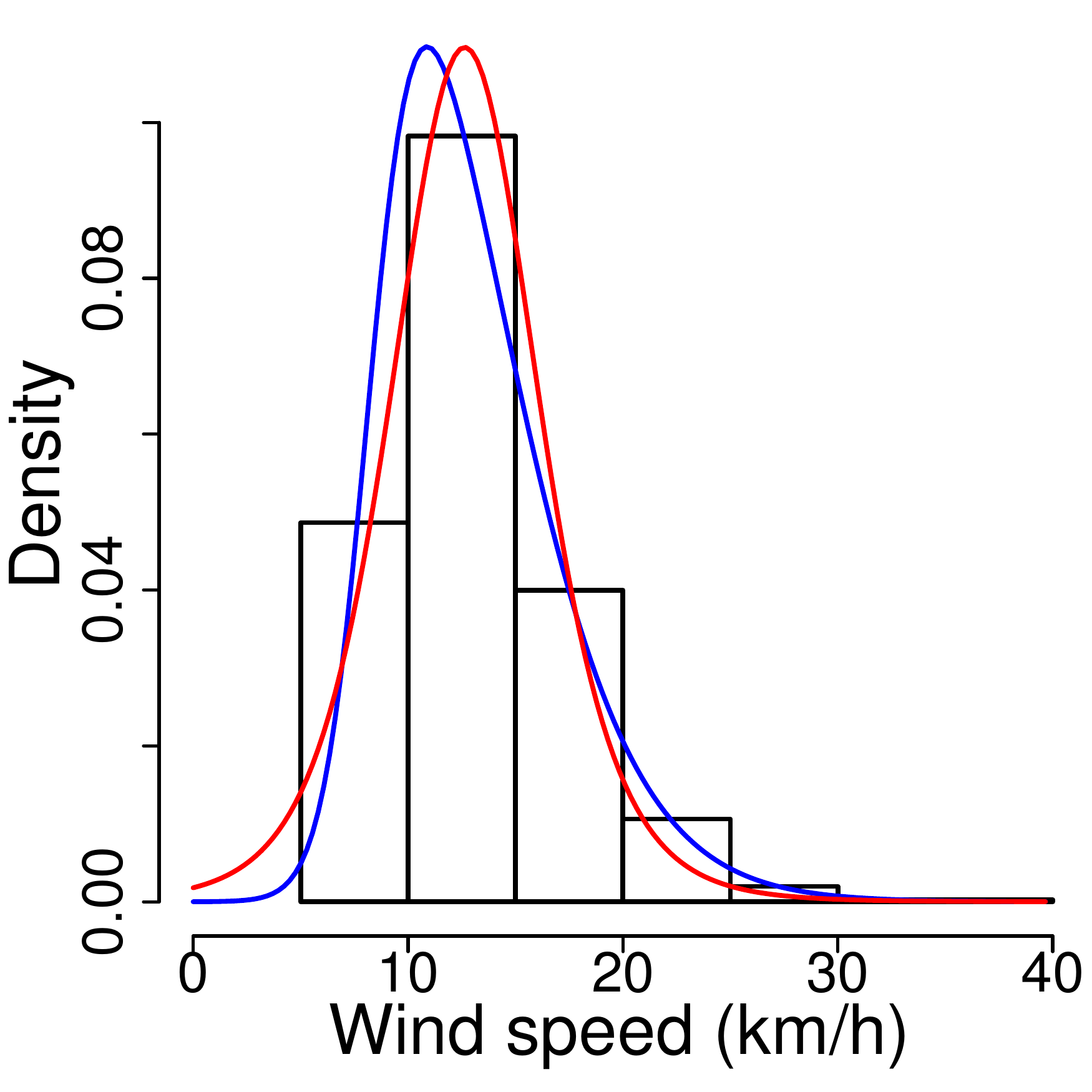}}\\
\makebox{
\includegraphics[width=0.4\textwidth]{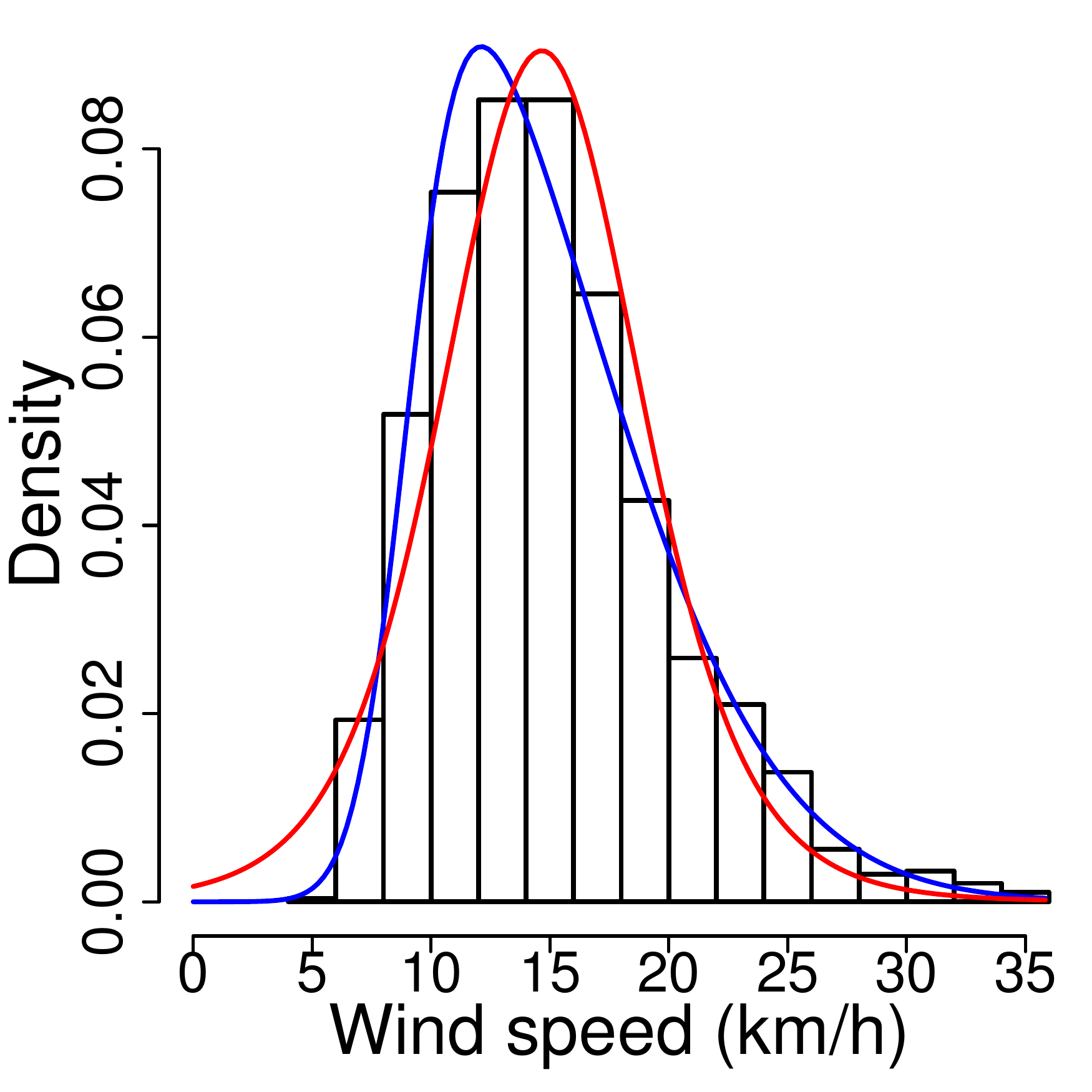} 
\includegraphics[width=0.4\textwidth]{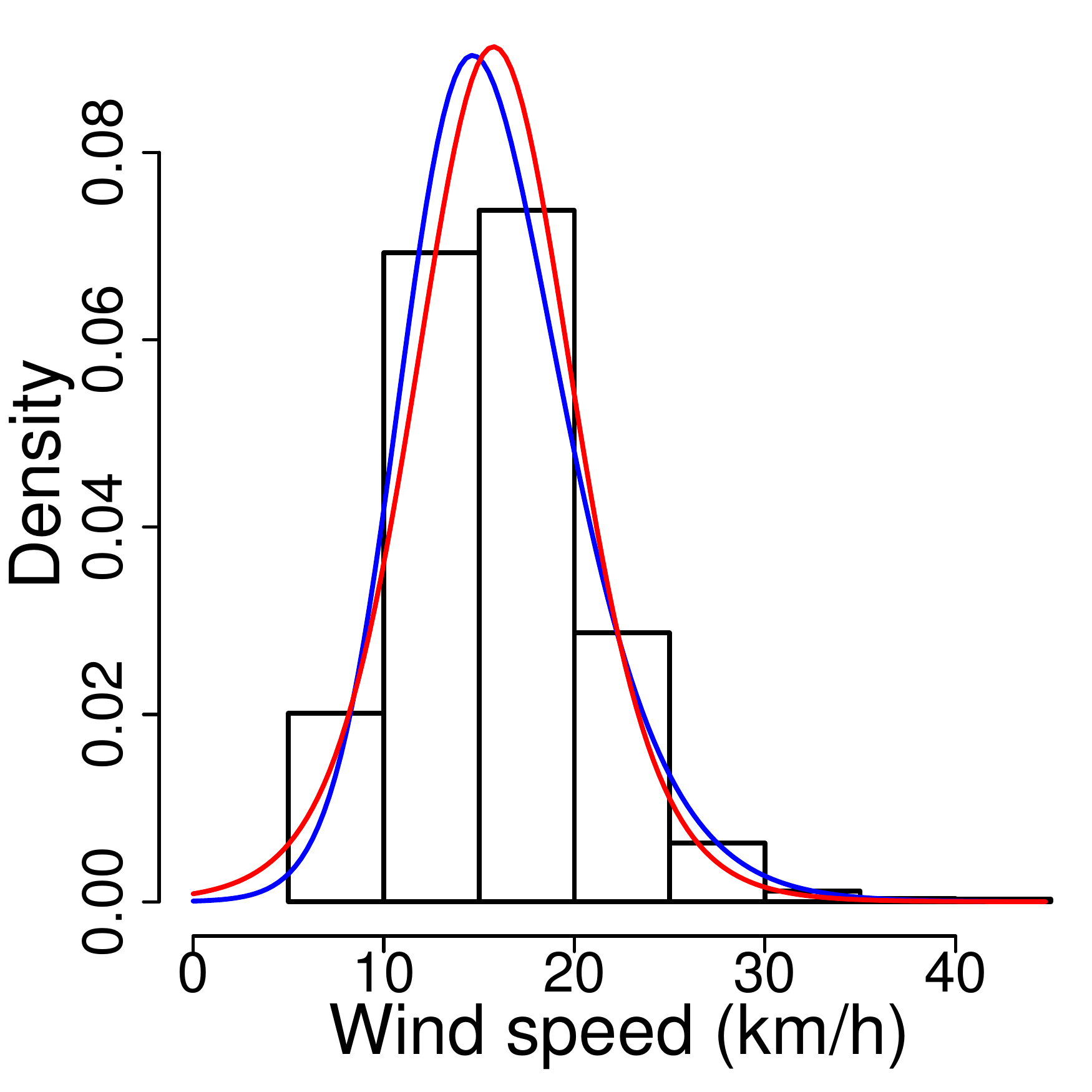} 
}
\caption{\label{fig:figMargin}\small Histogram of daily windspeed data, fitted $t$ (red solid line) and skew-$t$ (blue solid line) densities for each of the four stations CLOU (top-left), CLAY (top-right), SALL (bottom-left) and PAUL (bottom-right).}
\end{figure}

The red and blue solid lines in Figure \ref{fig:figMargin} respectively show the fitted $t$ and skew-$t$ densities compared to the histogram of the daily observations for each of the four monitoring stations.
Each of the plots clearly shows that the datasets are right skewed and that the model with the ability to handle skewness provides a better fit.

\end{document}